\documentclass[11pt,a4paper]{article}
\usepackage{geometry}
\usepackage{url}
\usepackage{hyperref}
\usepackage{graphicx}
\usepackage{amsthm}
\usepackage{amsmath}
\usepackage{amssymb}
\usepackage[ruled,linesnumbered,noend]{algorithm2e}
\usepackage{diagbox}
\usepackage{makecell}
\usepackage{comment}
\usepackage{color}
\usepackage{float}

\usepackage{thmtools}
\usepackage{thm-restate}

\usepackage{todonotes}

\newtheorem{theorem}{Theorem}
\newtheorem{lemma}[theorem]{Lemma}
\newtheorem{definition}[theorem]{Definition}
\newtheorem{observation}[theorem]{Observation}
\newtheorem{proposition}[theorem]{Proposition}

\newtheorem{problem}[theorem]{Problem}
\newtheorem{claim}[theorem]{Claim} 
\newtheorem{corollary}[theorem]{Corollary}

\newcommand{\falsevalue}{0}
\newcommand{\truevalue}{1}

\newcommand{\treeautomaton}{\mathcal{A}}

\newcommand{\statestreeautomaton}{Q}
\newcommand{\transitionstreeautomaton}{\Delta}
\newcommand{\finalstatestreeautomaton}{F}
\newcommand{\lang}{\mathcal{L}}

\newcommand{\transitionsdpcore}{\Delta}
\newcommand{\widthdecomposition}[2]{w_{#1}(#2)}

\newcommand{\treewidthvalue}{k}

\newcommand{\alphabet}{\Sigma}
\newcommand{\analgorithm}{\mathfrak{A}}
\newcommand{\decompositiongraph}[1]{\mathcal{G}(#1)}

\newcommand{\agraph}{G}
\newcommand{\bgraph}{H}
\newcommand{\Nplus}{\mathbb{N}_{+}}
\newcommand{\N}{\mathbb{N}}
\newcommand{\Z}{\mathbb{Z}}

\newcommand{\abag}{\mathfrak{b}}
\newcommand{\bagset}{B}

\newcommand{\allsubterms}{\mathrm{Sub}} 

\newcommand{\aterm}{\tau}

\newcommand{\nodes}[1]{\mathsf{Nodes}(#1)}
\newcommand{\anode}{p}

\newcommand{\aposition}{p}

\newcommand{\vertexsetname}{V}
\newcommand{\vertexset}[1]{\vertexsetname_{#1}}
\newcommand{\edgesetname}{E}
\newcommand{\edgeset}[1]{\edgesetname_{#1}}
\newcommand{\incidencerelationname}{\rho}
\newcommand{\incidencerelation}[1]{\incidencerelationname_{#1}}
\newcommand{\edgeendpointsname}{\mathrm{endpts}}
\newcommand{\edgeendpoints}[1]{\edgeendpointsname_{#1}}
\newcommand{\anedge}{e}
\newcommand{\avertex}{v}

\newcommand{\introvertextype}[1]{\mathtt{IntroVertex}\{#1\}}
\newcommand{\introedgetype}[2]{\mathtt{IntroEdge}\{#1,#2\}}
\newcommand{\forgetvertextype}[1]{\mathtt{ForgetVertex}\{#1\}}
\newcommand{\jointype}{\mathtt{Join}}
\newcommand{\leaftype}{\mathtt{Leaf}}

\newcommand{\introvertexgeneric}[1]{\mathtt{IntroVertex}\{#1\}}
\newcommand{\introedgegeneric}[2]{\mathtt{IntroEdge}\{#1,#2\}}
\newcommand{\forgetvertexgeneric}[1]{\mathtt{ForgetVertex}\{#1\}}
\newcommand{\joingeneric}{\mathtt{Join}}

\newcommand{\finalwitnessgeneric}{\mathtt{Final}}

\newcommand{\joingenericcore}{\mathtt{Join}}
\newcommand{\initialsetgenericcore}{\mathtt{LeafSet}}
\newcommand{\finalwitnessgenericcore}{\mathtt{Final}}
\newcommand{\cleaningfunctioncore}{\mathtt{Clean}}
\newcommand{\accepteddecompositions}[1]{\mathsf{Acc}(#1)}
\newcommand{\accepteddecompositionsboundedwidth}[1]{\mathsf{Acc}(#1)} 

\newcommand{\numberusefulsetssimple}[1]{\delta_{#1}}

\newcommand{\maxsizeusefulsetsimple}[1]{\mu_{#1}}
\newcommand{\bitlengthusefulwitnesssimple}[1]{\beta_{#1}}

\newcommand{\children}{\mathrm{Children}}

\newcommand{\asymbol}{a}
\newcommand{\arity}{\mathfrak{r}}
\newcommand{\arityvalue}{r}
\newcommand{\labelingfunction}{\lambda}

\newcommand{\topbagname}{B}
\newcommand{\topbag}[1]{\topbagname(#1)}

\newcommand{\allabstractdecompositionspathwidth}[1]{\mathsf{IPD}_{#1}}
\newcommand{\allabstractdecompositionstreewidth}[1]{\mathsf{ITD}_{#1}}
\newcommand{\instructivetreedecompositionclass}{\mathsf{ITD}}

\newcommand{\abstractdecomposition}{\tau}
\newcommand{\sigmaabstractdecomposition}{\sigma}

\newcommand{\rootnode}[1]{\mathsf{root}(#1)}

\newcommand{\sizedecomposition}{n}

\newcommand{\vertexone}{u}
\newcommand{\vertextwo}{v}

\newcommand{\allwitnesses}{\mathcal{W}}
\newcommand{\awitness}{\mathsf{w}}

\newcommand{\dynamizationfunctionname}[1]{\Gamma[{#1}]}
\newcommand{\dynamizationfunction}[2]{\dynamizationfunctionname{#1,#2}}

\newcommand{\abstractalphabet}[1]{\Sigma_{#1}}
\newcommand{\atree}{T}

\newcommand{\finitepowerset}[1]{\mathcal{P}_{\mathrm{fin}}(#1)}
\newcommand{\powerset}[1]{\mathcal{P}(#1)}
\newcommand{\powersetchoosek}[2]{\mathcal{P}(#1,#2)}

\newcommand{\Terms}[1]{\mathrm{Terms}(#1)}

\newcommand{\astate}{q}

\newcommand{\graphproperty}{\mathbb{P}}
\newcommand{\dpcoregraphproperty}[1]{\mathbb{G}[#1]}

\newcommand{\dpcoregraphpropertyclass}[2]{\mathbb{G}[#1,#2]}

\newcommand{\boldS}{\mathbf{S}}

\newcommand{\asetgraphs}{S}

\newcommand{\dpcore}{\mathsf{D}}

\newcommand{\numberproperties}{\ell}

\newcommand{\realizationclass}{\mathcal{A}}

\newcommand{\true}{\mathit{true}}

\newcommand{\witnessset}{S}

\newcommand{\allgraphs}{\textsc{Graphs}}

\newcommand{\allgraphstreewidth}[1]{\textsc{GraphsTw}[#1]}

\newcommand{\invariantmeasure}{\mathcal{M}}
\newcommand{\invariant}{\mathcal{I}}

\newcommand{\languageclass}{\mathsf{L}}
\newcommand{\graphfunction}{\mathcal{G}}
\newcommand{\decompositionclass}{\mathsf{C}}
\newcommand{\alphabetclass}{\Sigma}

\newcommand{\isomorphic}{\sim}
\newcommand{\isomorphismclosure}[1]{\mathsf{ISO}(#1)} 
\newcommand{\isomorphism}{\phi}
\newcommand{\isomorphismvertices}{\phi_1}
\newcommand{\isomorphismedges}{\phi_2}

\newcommand{\maximumvertex}{a}
\newcommand{\maximumedge}{b}

\newcommand{\topmapname}{\theta}
\newcommand{\topmap}[1]{\topmapname[#1]}

\newcommand{\invariantCore}{\mathtt{Inv}}

\newcommand{\combinator}{\mathcal{C}}

\newcommand{\indicatorfunction}[2]{#1(#2)}

\newcommand{\combinatorproperty}[1]{\hat{#1}}

\newcommand{\domain}{Dom}
\newcommand{\image}{Im}

\newcommand{\dprefutation}{R}

\newcommand{\xvertex}{x}

\usepackage{bbm}

\newcommand{\numbercolors}{c}
\newcommand{\vertexcoloring}{\alpha}

\newcommand{\flowValue}{m}
\newcommand{\flowWitness}{\awitness}
\newcommand{\flowAssignedValue}{f}
\newcommand{\flow}{\mathsf{Flow}}

\newcommand{\headmap}{h}
\newcommand{\tailmap}{t}
\newcommand{\labels}{\mathsf{Labels}}

\newcommand{\vertexcover}{X}
\newcommand{\vertexcoverParameter}{r}
\newcommand{\vertexcoverProperty}{\texttt{\mbox{$\mathtt{VertexCover}$}}}
\newcommand{\vertexcoverCore}{\texttt{\mbox{$\mathtt{C}$-$\mathtt{VertexCover}$}}}
\newcommand{\vertexcoverPred}{\texttt{\mbox{$\mathtt{P}$-$\mathtt{VertexCover}$}}}
\newcommand{\vertexcoverset}{R}
\newcommand{\vertexcoversize}{s}

\newcommand{\qconn}{\gamma}
\newcommand{\Pconn}{P}

\newcommand{\relabelingfunction}{\gamma}

\title{From Width-Based Model Checking to \\ Width-Based Automated Theorem Proving}
\date{$ $}
\author{
Mateus de Oliveira Oliveira$^{1,2}$ and Sam Urmian$^2$ \\ 
\\ 
$^1$Stockholm University, Sweden \\
$^2$University of Bergen, Norway \\ 
\\
{\small \texttt{oliveira@dsv.su.se, sam.urmian@uib.no}} \\ 
}

\begin{document}

\maketitle
\begin{abstract}
In the field of parameterized complexity theory, the study of  graph width measures has been 
intimately connected with the development of width-based model checking algorithms for combinatorial
properties on graphs. In this work, we introduce a general framework to convert a large class of width-based
model-checking algorithms into algorithms that can be used to test the validity of graph-theoretic conjectures on classes of graphs 
of bounded width. Our framework is modular and can be applied with respect to several well-studied width
measures for graphs, including treewidth and cliquewidth.

As a quantitative application of our framework, we prove analytically that for several long-standing graph-theoretic
conjectures, there exists an algorithm that takes a number $k$ as input and correctly
determines in time double-exponential in $k^{O(1)}$ whether the conjecture is valid on
all graphs of treewidth at most $k$. These upper bounds, which may be regarded as upper-bounds
on the size of proofs/disproofs for these conjectures on the class of graphs of treewidth at most $k$, 
improve significantly on theoretical upper bounds obtained using previously available techniques.
\end{abstract}

\section{Introduction} \label{section:Introduction}
\subsection{Motivation}
\label{subsection:Motivation}
When mathematicians are not able to solve a conjecture about a given class of
mathematical objects, it is natural to try to test the validity of the
conjecture on a smaller, or better behaved class of objects. In the realm of
graph theory, a common approach is to try to analyze the conjecture on 
restricted classes of graphs, defined by fixing some structural parameter.
In this work, we push forward this approach from a computational perspective 
by focusing on parameters derived from graph width measures. Prominent
examples of such parameters are the {\em treewidth} of a graph, which intuitively
quantifies how much a graph is similar to a tree 
\cite{robertson1984graph,bertele1973non,halin1976s}
and the cliquewidth of a graph, which intuitively quantifies how much a graph
is similar to a clique \cite{courcelle2000upper}.  
More specifically, we are concerned with the following problem. 

\begin{problem}[Width-Based ATP]
\label{problem:MainQuestionWBATP}
Given a graph property $\graphproperty$ and a positive integer $\treewidthvalue$, is it the case that every graph of width at most $\treewidthvalue$ belongs to $\graphproperty$? 
\end{problem}

Problem \ref{problem:MainQuestionWBATP} provides a width-based approach to the field of automated theorem proving (ATP).
For instance, consider Tutte's celebrated $5$-flow conjecture \cite{tutte1954contribution}, which states that every bridgeless graph has a nowhere-zero $5$-flow. Let $\texttt{HasBridge}$ be the graph property comprising all graphs that have a bridge and $\texttt{NZFlow}(5)$ be graph property comprising all graphs that admit a nowhere-zero $5$-flow. Then, proving Tutte's $5$-flow conjecture is equivalent to showing
that every graph belongs to the graph property $\texttt{HasBridge} \vee \texttt{NZFlow}(5)$. Since Tutte's conjecture has been 
unresolved for many decades, one possible approach for gaining understanding about this conjecture is 
to try to determine, for gradually increasing values of $k$,
whether every graph of width at most $\treewidthvalue$, with respect to some fixed width-measure, belongs to  $\texttt{HasBridge} \vee \texttt{NZFlow}(5)$. What makes this kind of question interesting from a proof
theoretic point of view is that several important classes of graphs have small width with respect to some width measure.
For instance, trees and forests have treewidth at most $1$, series-parallel graphs and outerplanar graphs have treewidth at most $2$, $k$-outerplanar graphs have treewidth at most $3k-1$, co-graphs have cliquewidth at most $2$, any distance hereditary graph has cliquewidth $3$, etc \cite{biedl2015triangulating,bodlaender1998partial,brandstadt1999graph,bodlaender1986classes,bodlaender2008combinatorial,kammer2007determining}. Therefore, proving the validity of a given conjecture on classes of graphs of small width corresponds to proving the conjecture on interesting classes of graphs. 

\subsection{Our Results}
\label{subsection:Results}
In this work, we introduce a general and modular framework that allows one to convert width-based dynamic programming algorithms
for the model checking of graph properties into algorithms that can be used to address Problem \ref{problem:MainQuestionWBATP}. 
More specifically, our main contributions are threefold. 

\begin{enumerate}
\item We start by defining the notions of a {\em treelike decomposition class} (Definition \ref{definition:TreelikeDecompositionClass}) 
and of a {\em treelike width-measure} (Definition \ref{definition:WidthMeasure}). These two notions can be used to express 
several well studied width measures for graphs, such as treewidth \cite{bodlaender1997treewidth}, pathwidth \cite{korach1993tree}, carving width \cite{thilikos2000constructive}, 
cutwidth \cite{chung1989graphs,thilikos2005cutwidth}, bandwidth \cite{chung1989graphs}, cliquewidth \cite{courcelle2000upper}, etc, and some more recent measures 
such as ODD-width \cite{andrade2019width}. 
It is worth noting that treelike width measures strictly generalize width measures defined by some prominent formalisms, such as families 
of graph grammars \cite{BauderonC87,BlumensathC06}. On the one hand, each graph grammar in such a family corresponds to a class of graphs of {\em constant} cliquewidth \cite{glikson2003nce}. 
On the other hand, the class of hypercube graphs has unbounded cliquewidth \cite{bonomo2016graph} but constant ODD-width \cite{andrade2019width}, 
and ODD-width is a treelike width measure.   
\item Subsequently, we introduce the notion of a
{\em treelike dynamic programming core} (Definition \ref{definition:DPCore}), 
a formalism for the specification of dynamic programming algorithms
operating on treelike decompositions. Our formalism combines two points 
of view for dynamic programming: the {\em symbolic} point of view 
\cite{DBLP:journals/tcs/CourcelleD16} that allows one to use symbolic tree 
automata to reason about classes of graphs of bounded width, 
and the {\em combinatorial} point of view \cite{baste2022diversity}
that provides a general methodology for the specification
of state-of-the-art combinatorial algorithms for model-checking graph properties
on classes of bounded width. Treelike DP-cores satisfying certain 
{\em coherency} and {\em finiteness} properties are suitable for the 
study of the computational complexity of the problem of determining whether
certain graph-theoretic conjectures are valid on the class of graphs 
of width at most $k$ (for several suitable notions of width).
Additionally, our formalism allows one to establish upper bounds on
the size of a hypothetical counter-example of width at most $k$
in terms of the computational complexity of algorithms for model-checking
the graph properties arising in the statement of the conjecture. 
\item Our main result (Theorem \ref{theorem:InclusionTestCombinators}) states that if a graph property $\graphproperty$ is a Boolean
combination (see Section \ref{subsection:CombinatorsCombination}) 
of graph properties $\graphproperty_1,\dots,\graphproperty_{\numberproperties}$ 
that can be decided by coherent and finite DP-cores $\dpcore_1,\dpcore_{2},\dots,\dpcore_{\numberproperties}$, then the process of 
determining whether every graph of width at most $\treewidthvalue$ belongs to $\graphproperty$ can be decided 
roughly\footnote{The precise statement of Theorem \ref{theorem:InclusionTestCombinators} involves other parameters that are negligible in typical applications.} in time
\[
2^{O(\beta(\treewidthvalue)\cdot \mu(\treewidthvalue))} 
\leq 2^{2^{O(\beta(\treewidthvalue))})},
\]
where $\mu(\treewidthvalue)$ and $\beta(\treewidthvalue)$ are respectively the maximum multiplicity and 
the maximum bitlength of a DP-core from the list $\dpcore_{1},\dots,\dpcore_{\numberproperties}$ (see Section \ref{subsection:ComplexityMeasures}).
Additionally, if a counterexample of width at most $\treewidthvalue$ exists, then a term of height at most $2^{O(\beta(\treewidthvalue)\cdot \mu(\treewidthvalue))}$ 
representing such a counterexample can be constructed (Corollary \ref{corollary:SizeCounterexample}). 
\end{enumerate}

We illustrate our approach by specializing on graphs of bounded treewidth. Mores specifically, we show that several long-standing conjectures in graph theory 
can be tested on the class of graphs of treewidth at most $\treewidthvalue$ 
in time double exponential in $k^{O(1)}$.
Examples of such conjectures include Hadwiger conjecture
\cite{hadwiger1943klassifikation}, Tutte's flow conjectures
\cite{tutte1954contribution} and  Barnette's conjecture \cite{tutte1969recent}
(Section \ref{section:Applications}). Our upper bounds improve significantly on upper bounds obtained using previous techniques. For instance, we show
that Hadwiger's conjecture for $c$ colors can be verified on the class of graphs
of treewidth at most $k$ in time $2^{2^{O(k\log k + c^2)}}$, while previously, the 
best known upper bound was of the form $p^{p^{p^{p}}}$ 
for $p=(k+1)^{c-1}$ \cite{kawarabayashi2009hadwiger}. 

\subsection{Related Work}
\label{subsection:RelatedWork}

General automata-theoretic frameworks for the development of dynamic programming algorithms have been introduced under a wide variety of contexts 
\cite{gnesi1981dynamic,papusha2016automata,pardo1997automata,kumar1988cdp,parker1989partial,morales2000parallel,leon2009tools,bannach2019positive,DBLP:journals/algorithms/BannachB19,baste2022diversity}.
In most of these contexts, automata are used to encode the space of solutions of combinatorial problems when a graph $G$ is given at the input. 
For instance, given a tree decomposition of width $\treewidthvalue$ of a graph $\agraph$, one can construct a tree automaton representing the set of proper $3$-colorings of $\agraph$
\cite{DBLP:journals/algorithms/BannachB19}. 

In our framework, treelike DP-cores are used to define 
graph properties. For instance, one can define a treelike DP-core $\dpcore$, where for each 
$\treewidthvalue\in \N$, $\dpcore[\treewidthvalue]$ is a finite representation of the set of all graphs of treewidth at most $\treewidthvalue$ that are $3$-colorable. 
In our setting, it is essential that graphs of width $\treewidthvalue$ are encoded as terms over a finite alphabet whose size only depends on $\treewidthvalue$. 
We note that the idea of representing families of graphs as tree languages over a finite alphabet has been used in a wide variety of contexts
\cite{PilipczukBojanczyk2016,AdlerGroheKreutzer2008,CourcelleEngelfriet2012,Elberfeld2016,Flum2002query,DBLP:journals/tcs/CourcelleD16}. Nevertheless,
the formalisms arising in these contexts are usually designed to be compatible with logical algorithmic meta-theorems, and for this reason, 
tree automata are meant to be compiled from logical specifications, rather than to be programmed.
In contrast, our combinatorial framework allows one to easily specify
state-of-the-art dynamic programming algorithms operating on treelike decompositions, and 
to safely combine these algorithms (just like plugins) for the purpose of 
width-based automated theorem proving. 

The monadic second-order logic of graphs (MSO$_2$ logic) extends first-order logic by introducing quantification over sets of
vertices and over sets of edges. This logic is powerful enough to express several well studied graph-theoretic 
properties such as connectivity, Hamiltonicity, $3$-colorability, and many others. Given that for each $\treewidthvalue\in \N$, 
the MSO$_2$ theory of graphs of treewidth at most $\treewidthvalue$
is decidable \cite{seese1991structure,Courcelle1990MSO}, we have that if a graph property $\graphproperty$ is
definable in MSO$_2$ logic, then there is an algorithm that takes an integer $\treewidthvalue$ as input, and
correctly determines whether every graph of treewidth at most $\treewidthvalue$ belongs to $\graphproperty$.
A similar result can be proved with respect to graphs of constant cliquewidth using MSO$_1$ logic \cite{seese1991structure}, and for 
graphs of bounded ODD width using FO logic \cite{andrade2019width}. One issue with addressing 
Problem \ref{problem:MainQuestionWBATP} using this logic-theoretic approach is that algorithms obtained
in this way are usually based on quantifier-elimination. As a consequence, the 
function upper-bounding the running time of these algorithms in terms of the width parameter grows 
as a tower of exponentials whose height depends on the number of quantifier alternations of the logical sentence
given as input to the algorithm. For instance, the algorithm obtained using this approach to test the validity of 
Hadwiger's conjecture restricted to $c$ colors on graphs of treewidth at most $k$ has a very large
dependency on the width parameter. In \cite{kawarabayashi2009hadwiger}, the time necessary to perform
this task was estimated in $f(\numbercolors,\treewidthvalue)\leq p^{p^{p^{p}}}$,
where $p=(\treewidthvalue+1)^{(\numbercolors-1)}$ \cite{kawarabayashi2009hadwiger}.
Our approach yields a much smaller upper bound of the form $2^{2^{O(\treewidthvalue \cdot \log \treewidthvalue + c^2)}}$. 
Significant reductions in complexity, with respect to the logical approach, can also be observed for other conjectures.

Algorithms with optimal dependency on the width parameter have been obtained for many graph properties 
\cite{DBLP:journals/eatcs/LokshtanovMS11,pilipczuk2011problems} under standard complexity theoretic assumptions, such as 
the exponential time hypothesis (ETH) \cite{DBLP:journals/jcss/ImpagliazzoPZ01,DBLP:journals/jcss/ImpagliazzoP01} and 
related conjectures \cite{calabro2009complexity,DBLP:conf/innovations/CarmosinoGIMPS16}. It is worth noting that in many cases, 
the development of such optimal algorithms requires the use of advanced techniques borrowed from diverse subfields
of mathematics, such as structural graph theory \cite{robertson1986graph,DBLP:conf/soda/BasteST20}, 
algebra \cite{DBLP:conf/esa/RooijBR09,bodlaender2015deterministic,DBLP:conf/stacs/NederlofPSW22}, 
combinatorics \cite{raymond2017recent,DBLP:journals/siamdm/LokshtanovMSZ19}, etc. 
Our framework allows one to incorporate several of these techniques in the development of faster algorithms for width-based automated theorem 
proving.

\section{Preliminaries} \label{section:Preliminaries} 
We let $\N$ denote the set of natural numbers and $\Nplus$ denote the set of
positive natural numbers. We let $[0] = \emptyset$, and for each $n\in
\Nplus$, we define $[n] = \{1,...,n\}$. 
Given a set $S$, the set of finite subsets of $S$ is denoted by
$\finitepowerset{S}$.

In this work, a {\em graph} is a triple $\agraph =
(\vertexsetname,\edgesetname,\incidencerelationname)$ where $\vertexsetname
\subseteq \N$ is a {\em finite} set of {\em vertices}, $\edgesetname \subseteq
\N$ is a finite set {\em edges}, and $\incidencerelationname \subseteq
\edgesetname\times \vertexsetname$ is an {\em incidence relation}.  For each
edge $\anedge\in \edgesetname$, we let $\edgeendpointsname(\anedge) =
\{\avertex\in \vertexsetname\;:\;
(\anedge,\avertex)\in\incidencerelationname\}$ be the set of vertices incident
with $\anedge$.  In what follows, we may write $\vertexset{\agraph}$,
$\edgeset{\agraph}$ and $\incidencerelation{\agraph}$ to denote the sets
$\vertexsetname$, $\edgesetname$ and $\incidencerelationname$ respectively. We
let $|\agraph| = |\vertexset{\agraph}|+|\edgeset{\agraph}|$ be the {\em size}
of $\agraph$.  We let $\allgraphs$ denote the set of all graphs. For us, the
{\em empty graph} is the graph $(\emptyset,\emptyset,\emptyset)$ with no
vertices, no edges, and no incidence pairs.

An {\em isomorphism} from a graph $\agraph$ to a graph $\bgraph$ is a pair
$\isomorphism  = (\isomorphismvertices,\isomorphismedges)$ where
$\isomorphismvertices:\vertexset{\agraph}\rightarrow \vertexset{\bgraph}$ is a
bijection from the vertices of $\agraph$ to the vertices of $\bgraph$ and
$\isomorphismedges:\edgeset{\agraph}\rightarrow \edgeset{\bgraph}$ is a
bijection from the edges of $\agraph$ to the edges of $\bgraph$ with the
property that for each 
vertex $\avertex\in \vertexset{\agraph}$
and each 
edge $\anedge \in \edgeset{\agraph}$, 
$(\anedge,\avertex)\in \incidencerelation{\agraph}$ if and only if
$(\isomorphismedges(\anedge),\isomorphismvertices(\avertex))\in \incidencerelation{\bgraph}$. 
If such a bijection exists, we say that $\agraph$
and $\bgraph$ are {\em isomorphic}, and denote this fact by $\agraph\sim \bgraph$. 

A {\em graph property} is any subset $\graphproperty \subseteq \allgraphs$
closed under isomorphisms. That is to say, for each two isomorphic graphs
$\agraph$ and $\bgraph$ in $\allgraphs$, $\agraph \in\graphproperty$ if and
only if $\bgraph\in\graphproperty$. Note that the sets $\emptyset$ and
$\allgraphs$ are graph properties. Given a set $\asetgraphs$ of graphs, the
{\em isomorphism closure} of $\asetgraphs$ is defined as the set
$\isomorphismclosure{\asetgraphs} = \{\agraph\in \allgraphs \;:\; \exists
\bgraph \in \asetgraphs, \agraph\isomorphic \bgraph\}.$

Given a graph property $\graphproperty$, a {\em $\graphproperty$-invariant} is
a function $\invariant:\graphproperty \rightarrow S$, for some set $S$, 
that is invariant under graph isomorphisms.  More precisely, $\invariant(\agraph) =
\invariant(\bgraph)$ for each two isomorphic graphs $\agraph$ and $\bgraph$ in
$\graphproperty$. If $\graphproperty=\allgraphs$, we may say that $\invariant$
is simply a {\em graph invariant}. For instance, chromatic number, clique
number, dominating number, etc., as well as width measures such as treewidth and 
cliquewidth, 
are all graph invariants. In this work, the set $S$ will be typically $\N$,
when considering width measures, or $\{0,1\}^*$ when considering other invariants encoded in binary.
In order to avoid confusion we may use the letter $\invariantmeasure$ to denote invariants
corresponding to width measures, and the letter $\invariant$ to denote general invariants.

A {\em ranked alphabet} is a finite non-empty set $\alphabet$ together with
function $\arity:\alphabet\rightarrow \N$ that specifies the arity of
each symbol in $\alphabet$. The arity of $r$ is the maximum arity 
of a symbol in $\alphabet$. 
A term over $\alphabet$ is a pair $\aterm = (\atree,\labelingfunction)$ where $\atree$ is a
rooted tree, where the children of each node are totally ordered,
and $\labelingfunction:\nodes{\atree}\rightarrow \alphabet$ is a
function that labels each node $\anode$ in $\nodes{\atree}$ with a
symbol from $\alphabet$ of arity $|\children(\anode)|$, i.e., the number of children
of $\anode$. In particular, leaf nodes
are labeled with symbols of arity $0$.
We may write $\nodes{\aterm}$ to refer to $\nodes{\atree}$. We
write $|\aterm|$ to denote $|\nodes{\atree}|$. The height of $\aterm$ is
defined as the height of $\atree$.  We denote by $\Terms{\alphabet}$
the set of all terms over $\alphabet$. 
If $\aterm_1=(\atree_1,\labelingfunction_1),...,\aterm_{\arityvalue}=(\atree_\arityvalue,\labelingfunction_\arityvalue)$
are terms in $\Terms{\alphabet}$, and $\asymbol\in \alphabet$ is a symbol of arity $\arityvalue$,
then we let $\asymbol(\aterm_1,...,\aterm_\arityvalue)$ denote the term $\aterm =
(\atree,\labelingfunction)$ where $\nodes{\atree} = \{u\}\cup
\nodes{\atree_1}\cup \dots \cup \nodes{\atree_{\arityvalue}}$ for some fresh node $u$,
$\rootnode{\atree} = u$, $\labelingfunction(u) = \asymbol$, and
$\labelingfunction|_{\nodes{\atree_j}} = \labelingfunction_j$ for each $j\in
[\arityvalue]$.
A tree automaton is a tuple $\treeautomaton =
(\alphabet,\statestreeautomaton,\finalstatestreeautomaton,\transitionstreeautomaton)$
where $\alphabet$ is a ranked alphabet, $\statestreeautomaton$ is a finite set
of states, $\finalstatestreeautomaton$ is a final set of states, and
$\transitionstreeautomaton$ is a set of transitions (i.e. rewriting rules) of the form
$\asymbol(\astate_1,\dots,\astate_{r})\rightarrow \astate$, where $\asymbol$ is
a symbol of arity $r$, and $\astate_1,\dots,\astate_{r},\astate$ are states in
$\statestreeautomaton$. A term $\aterm$ is accepted by $\treeautomaton$ if it
can be rewritten into a final state in $\finalstatestreeautomaton$ by
transitions in $\transitionstreeautomaton$. 
The language of $\treeautomaton$, denoted by $\lang(\treeautomaton)$, is the set of all terms accepted by
$\treeautomaton$.
We refer to \cite{TATA2008} for basic concepts on tree automata theory.

\section{Treelike Width Measures}
\label{section:TreelikeWidthMeasures}

In this section, we introduce the notion of a {\em treelike width measure}. Subsequently, 
we show that prominent width measures such as treewidth and cliquewidth 
fulfil the conditions of our definition. We start by introducing the notion of a {\em treelike decomposition class}. 

\begin{definition}
\label{definition:TreelikeDecompositionClass}
Let $\arityvalue\in \N$. A treelike
decomposition class of arity $\arityvalue$ is a sequence 
$\decompositionclass =
\{(\alphabetclass_{\treewidthvalue},\languageclass_{\treewidthvalue},
\graphfunction_{\treewidthvalue})\}_{\treewidthvalue\in
\N},
$
where for each $\treewidthvalue\in\N$,
$\alphabetclass_{\treewidthvalue}$ is a ranked alphabet of arity at most $\arityvalue$, 
$\languageclass_{\treewidthvalue}$ is a regular tree language over $\alphabetclass_{\treewidthvalue}$, and 
$\graphfunction_{\treewidthvalue}:\languageclass_{\treewidthvalue}\rightarrow \allgraphs$ is a function that assigns
a graph $\graphfunction_{\treewidthvalue}(\abstractdecomposition)$ to each $\abstractdecomposition\in \languageclass_{\treewidthvalue}$.
Additionally, we require that for each $\treewidthvalue\in \N$, $\alphabetclass_{\treewidthvalue}\subseteq \alphabetclass_{\treewidthvalue+1}$, 
$\languageclass_{\treewidthvalue}\subseteq \languageclass_{\treewidthvalue+1}$, and 
$\graphfunction_{k+1}|_{\languageclass_{k}} = \graphfunction_{\treewidthvalue}$. 
\end{definition}

Terms in the set $\languageclass(\decompositionclass) = \bigcup_{\treewidthvalue\in \N}
\languageclass_{\treewidthvalue}$ are called {\em {$\decompositionclass$-decompositions}}.
For each such a term $\abstractdecomposition$, we may write simply $\graphfunction(\abstractdecomposition)$
to denote $\graphfunction_k(\abstractdecomposition)$. The {\em $\decompositionclass$-width} of a 
$\decompositionclass$-decomposition $\abstractdecomposition$, denoted by $\widthdecomposition{\decompositionclass}{\abstractdecomposition}$, 
is the minimum $\treewidthvalue$ such that $\abstractdecomposition\in \languageclass_{\treewidthvalue}$. 
The {\em $\decompositionclass$-width} of a graph $\agraph$, denoted by $\widthdecomposition{\decompositionclass}{\agraph}$,
is the minimum $\decompositionclass$-width of a $\decompositionclass$-decomposition $\abstractdecomposition$ 
with $\graphfunction(\abstractdecomposition) \simeq \agraph$. 
We let $\widthdecomposition{\decompositionclass}{\agraph} =\infty$ if no such minimum $\treewidthvalue$
exists.

For each $\treewidthvalue\in \N$, we may write $\decompositionclass_{\treewidthvalue}=(\alphabetclass_{\treewidthvalue},\languageclass_{\treewidthvalue},\graphfunction_{\treewidthvalue})$
to denote the $\treewidthvalue$-th triple in $\decompositionclass$. The {\em graph property defined 
by $\decompositionclass_\treewidthvalue$} is the set $\dpcoregraphproperty{\decompositionclass_\treewidthvalue} =
\isomorphismclosure{\{\graphfunction(\abstractdecomposition)\;:\;
\abstractdecomposition\in \languageclass_{\treewidthvalue}\}}$. Note that every graph in 
$\dpcoregraphproperty{\decompositionclass_\treewidthvalue}$ has $\decompositionclass$-width
at most $\treewidthvalue$, and that 
$\dpcoregraphproperty{\decompositionclass_\treewidthvalue}\subseteq \dpcoregraphproperty{\decompositionclass_{\treewidthvalue+1}}$. We let $\dpcoregraphproperty{\decompositionclass} = \bigcup_{\treewidthvalue\in \N}
\dpcoregraphproperty{\decompositionclass_{\treewidthvalue}}$ be the graph property defined by 
$\decompositionclass$. We note that the $\decompositionclass$-width of any graph in 
$\dpcoregraphproperty{\decompositionclass}$ is finite.

\begin{definition}[Treelike Width Measure]
\label{definition:WidthMeasure}
Let $\graphproperty$ be a graph property and $\invariantmeasure:\graphproperty \rightarrow \N$ be 
a $\graphproperty$-invariant. We say that $\invariantmeasure$ is a {\em treelike width measure} if there is a treelike 
decomposition class $\decompositionclass$ such that $\graphproperty = \dpcoregraphproperty{\decompositionclass}$, 
and for each graph $\agraph\in \graphproperty$, $\widthdecomposition{\decompositionclass}{\agraph} = \invariantmeasure(\agraph)$.
In this case, we say that $\decompositionclass$ is a {\em realization} of $\invariantmeasure$. 
\end{definition}

The next theorem states that several well studied width measures for graphs are treelike.
The proof of this theorem can be found in Appendix \ref{appendix:proofTheoremTreewidthTreelike}. 

\begin{theorem}
\label{theorem:TreewidthTreelike}
The width measures treewidth, pathwidth, carving width, cutwidth, cliquewidth and ODD width\cite{andrade2019width} are treelike width measures. 
\end{theorem}

In our results related to width-based automated theorem proving, we will need to take into consideration the time necessary to construct a description 
of the languages associated with a treelike decomposition class. 
Let $\decompositionclass = \{(\alphabet_{\treewidthvalue},\languageclass_{\treewidthvalue},\graphfunction_{\treewidthvalue})\}_{\treewidthvalue\in \N}$ 
be a treelike decomposition class of arity $\arityvalue$. An automation for $\decompositionclass$ is a sequence
$\realizationclass = \{\realizationclass_{\treewidthvalue}\}_{\treewidthvalue\in \N}$ of tree automata where for each
$\treewidthvalue\in \N$, $\lang(\realizationclass_{\treewidthvalue}) = \languageclass_{\treewidthvalue}$. We say that $\realizationclass$ has complexity 
$f:\N\rightarrow \N$ if for each $\treewidthvalue\in \N$, $\realizationclass_{\treewidthvalue}$ has at most $f(\treewidthvalue)$ states, and 
there is an algorithm $\analgorithm$ that takes a number $\treewidthvalue\in \N$ as input, and constructs $\treeautomaton_{\treewidthvalue}$ in 
time $\treewidthvalue^{O(1)}\cdot f(\treewidthvalue)^{O(\arityvalue)}$.

\section{A DP-Friendly Realization of Treewidth}
\label{DPRealization}

As stated in Theorem \ref{theorem:TreewidthTreelike}, treewidth fulfills our definition of a treelike width measure. Indeed, suitable 
encodings of graphs of treewidth at most $\treewidthvalue$ as terms over a finite alphabet (whose size may depend on $\treewidthvalue$) 
can be used to realize treewidth as a treelike width measure. Examples of such encodings have been considered in the literature under a
wide variety of contexts \cite{downey2012parameterized,PilipczukBojanczyk2016,AdlerGroheKreutzer2008,CourcelleEngelfriet2012,Elberfeld2016,Flum2002query}.
In this section, we introduce the notion of a {\em $\treewidthvalue$-instructive tree decomposition}, and use this notion to define an alternative
realization of treewidth as a treelike width measure. 
At the same time that $\treewidthvalue$-instructive tree decompositions provide a representation of graphs of treewidth at most $\treewidthvalue$
as terms over a finite alphabet, 
such decompositions have a close correspondence with the notion edge-introducing 
tree decompositions, which are often used to define dynamic programming algorithms parameterized by treewidth. This correspondence allows one 
to translated dynamic programming algorithms that decide graph properties by processing edge-introducing tree decompositions of width at most $\treewidthvalue$ into 
dynamic programming algorithms that decide such properties by processing $\treewidthvalue$-instructive tree decompositions. 
These translated algorithms can then be used in the context of automated theorem proving.

\begin{definition}
\label{definition:kInstructiveAlphabet}
For each $\treewidthvalue\in \N$, we let 
$$
\begin{array}{l}
\abstractalphabet{\treewidthvalue}   = \{\leaftype,\; \introvertextype{\vertexone},\;
\forgetvertextype{\vertexone},\; \\
\hphantom{a}\hspace{5.1cm}  \introedgetype{\vertexone}{\vertextwo}, \; \jointype  
  \;:\; \vertexone,\vertextwo\in [\treewidthvalue+1], \vertexone\neq \vertextwo\}. 
\end{array}
$$	
where $\leaftype$ is a symbol of arity $0$, $\introvertextype{\vertexone}$, $\forgetvertextype{\vertexone}$
and $\introedgetype{\vertexone}{\vertextwo}$ are symbols of arity $1$, and $\jointype$ is a symbol of arity $2$. 
We call $\abstractalphabet{\treewidthvalue}$ the $\treewidthvalue$-instructive alphabet.
\end{definition}

Intuitively, the elements of $\abstractalphabet{\treewidthvalue}$ should be regarded as instructions
that can be used to construct graphs inductively. Each such a graph has an associated set $\abag\subseteq [\treewidthvalue+1]$ 
of {\em active labels}. In the base case, the instruction $\leaftype$ creates an empty graph with an empty set of active labels. 
Now, let $\agraph$ be a graph with set of active 
labels $\abag$. For each $\vertexone\in [\treewidthvalue+1]\backslash \abag$, the instruction $\introvertextype{\vertexone}$ adds a new vertex
to $G$, labels this vertex with $\vertexone$, and adds $\vertexone$ to $\abag$. For each $\vertexone\in \abag$, the 
instruction $\forgetvertextype{\vertexone}$ erases the label from the current vertex labeled with $\vertexone$, and 
removes $\vertexone$ from $\abag$. The intuition is that the label $\vertexone$ is now free and may be used later in the 
creation of another vertex. For each $\vertexone,\vertextwo\in \abag$, the instruction $\introedgetype{\vertexone}{\vertextwo}$
introduces a new edge between the current vertex labeled with $\vertexone$ and the current vertex labeled with $\vertextwo$.
We note that multiedges are allowed in our graphs.  
Finally, if $\agraph$ and $\agraph'$ are two graphs, each having $\abag$ as the set of active labels, then the instruction
$\jointype$ creates a new graph by identifying, for each $\vertexone\in \abag$, the vertex of $\agraph$ labeled 
with $\vertexone$ with the vertex of $\agraph'$ labeled with $u$. 

\begin{figure}[h]
\centering
\includegraphics[scale=0.3]{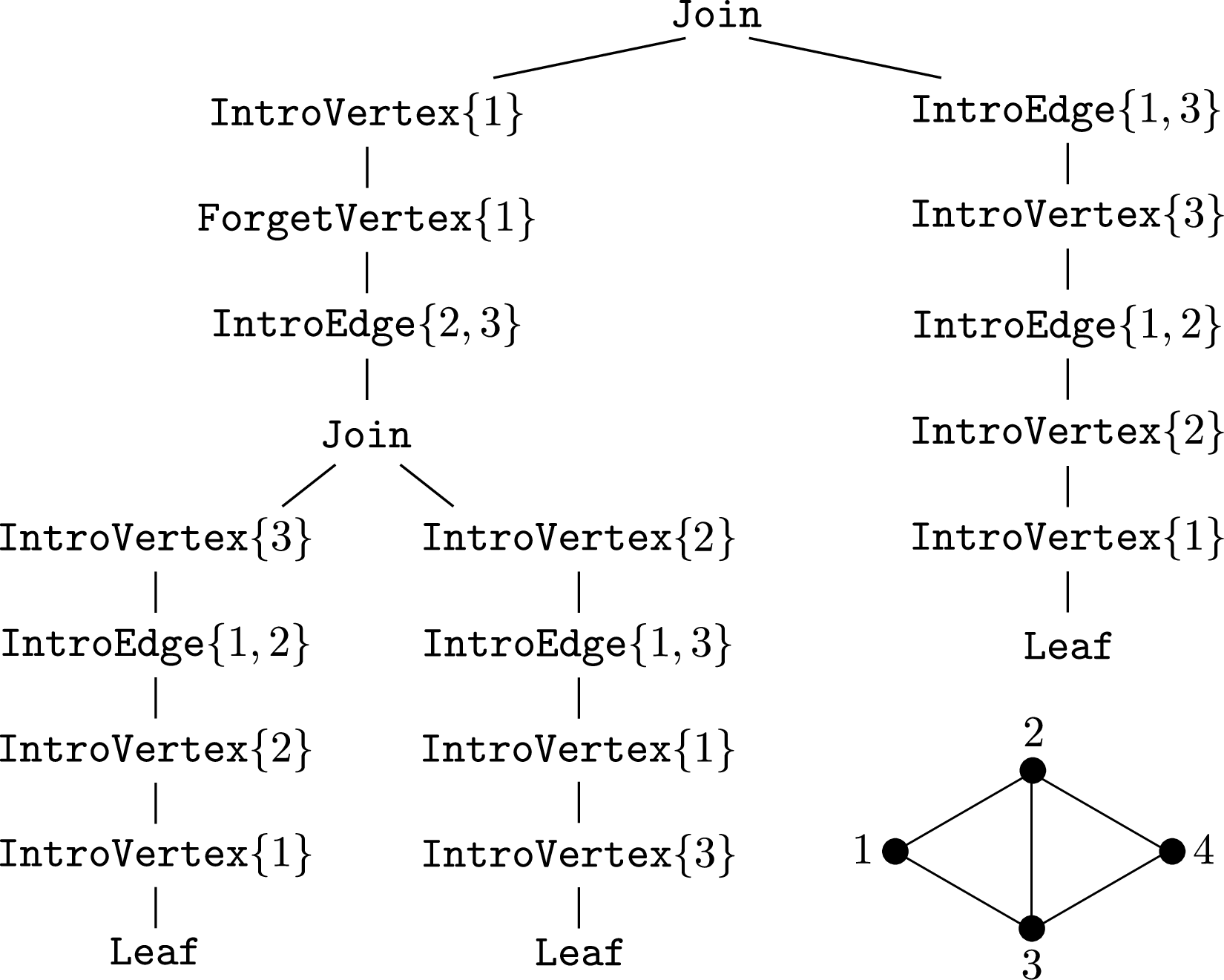}
\caption{A $2$-instructive tree decomposition $\abstractdecomposition$, and the 
graph $\decompositiongraph{\abstractdecomposition}$ associated with $\abstractdecomposition$.
Note that the graph has four vertices even though only labels from the set $\{1,2,3\}$ are used in the instructions occurring in the tree. Intuitively, once a vertex has been forgotten, its label can be reused when introducing a new vertex. \label{figure:TreeDecomposition}}
\end{figure}

A graph constructed according to the process described above can be encoded by a term over the alphabet $\abstractalphabet{\treewidthvalue}$. We let 
$\allabstractdecompositionstreewidth{\treewidthvalue}$ be the set of all terms over 
$\abstractalphabet{\treewidthvalue}$ that encode the construction of some graph.
The terms in $\allabstractdecompositionstreewidth{\treewidthvalue}$ are called $\treewidthvalue$-instructive tree decompositions (see Appendix \ref{formalDefinitionITD} for a formal definition). The following lemma, 
whose proof can be found in Appendix
\ref{section:ProofLemmaTreewidthIsTreelikeConfirmation},
establishes a close correspondence between graphs of treewidth at most $k$ and $k$-instructive tree decompositions. 

\begin{lemma}
\label{lemma:TreewidthIsTreelikeConfirmation}
Let $\agraph\in \allgraphs$ and $\treewidthvalue\in \N$. 
Then $\agraph$ has treewidth at most $\treewidthvalue$ 
if and only if there exists a $\treewidthvalue$-instructive tree decomposition $\abstractdecomposition$ 
such that $\graphfunction(\abstractdecomposition) \simeq \agraph$. 
\end{lemma}

\section{Treelike Dynamic Programming Cores}
\label{section:DPCores}

In this section, we introduce the notion of a {\em treelike dynamic-programming core}
(treelike DP-core), a formalism intended to capture the behavior of dynamic programming
algorithms operating on treelike decompositions. Our formalism generalizes and refines
the notion of dynamic programming core introduced in \cite{baste2022diversity}.
There are two crucial differences. First, our framework can be used to define DP-cores
for classes of dense graphs, such as graphs of constant cliquewidth, whereas the DP-cores
devised in \cite{baste2022diversity} are specialized to work on tree decompositions. 
Second, and most importantly, in our framework, graphs of width $\treewidthvalue$ can be 
represented as terms over ranked alphabets whose size depends only on $\treewidthvalue$.
This property makes our framework modular and particularly suitable for applications
in the realm of automated theorem proving. 

\begin{definition}[Treelike DP-Cores]
\label{definition:DPCore}
A treelike dynamic programming core is a sequence of $6$-tuples
$\dpcore =
\{(\alphabetclass_k,\allwitnesses_{\treewidthvalue},\finalwitnessgenericcore_{\treewidthvalue},\transitionsdpcore_{\treewidthvalue},\cleaningfunctioncore_{\treewidthvalue}, \invariantCore_{\treewidthvalue})\}_{\treewidthvalue\in \N}$ where for each $\treewidthvalue\in \N$,
\begin{enumerate}
\item $\alphabetclass_k$ is a ranked alphabet;
\item $\allwitnesses_k$ is a decidable subset of $\{0,1\}^*$;
\item $\finalwitnessgenericcore_{\treewidthvalue}: \allwitnesses_{\treewidthvalue}\rightarrow \{\falsevalue,\truevalue\}$ is a function; 
\item $\transitionsdpcore_{\treewidthvalue}$ is a set containing
\begin{itemize}
\item 
a {\em finite} set 
$\hat{\asymbol}\subseteq \finitepowerset{\allwitnesses_{\treewidthvalue}}$ 
for each symbol $\asymbol$ of arity $0$, 
\item 
a function $\hat{\asymbol}:\allwitnesses_{\treewidthvalue}^{\times \arity(\asymbol)}\rightarrow \finitepowerset{\allwitnesses_{\treewidthvalue}}$ for each symbol $\asymbol$ of arity $\arity(\asymbol)\geq 1$; 
\end{itemize}
\item $\cleaningfunctioncore_{\treewidthvalue}:\finitepowerset{\allwitnesses_{\treewidthvalue}} \rightarrow \finitepowerset{\allwitnesses_{\treewidthvalue}}$ is a function;
\item $\invariantCore_{\treewidthvalue}:\finitepowerset{\allwitnesses_{\treewidthvalue}} \rightarrow \{0,1\}^*$ is a function. 
\end{enumerate}
\end{definition}

We let $\dpcore[\treewidthvalue]= (\alphabetclass_k,\allwitnesses_{\treewidthvalue},\finalwitnessgenericcore_{\treewidthvalue},
\transitionsdpcore_{\treewidthvalue},\cleaningfunctioncore_{\treewidthvalue}, \invariantCore_{\treewidthvalue})$ denote the 
$\treewidthvalue$-th tuple of $\dpcore$. We may write $\dpcore[\treewidthvalue].\alphabetclass$ to denote the set
$\alphabetclass_{\treewidthvalue}$, $\dpcore[\treewidthvalue].\allwitnesses$ to denote the set $\allwitnesses_{\treewidthvalue}$,
and so on. 

Intuitively, for each $\treewidthvalue$, $\dpcore[\treewidthvalue]$ is a description of a dynamic programming algorithm that
operates on terms from $\Terms{\alphabet_k}$. This algorithm processes such a term $\abstractdecomposition$ from the leaves towards the root, and assigns a set of local witnesses to each node of $\abstractdecomposition$.
The algorithm starts by assigning the set $\dpcore[\treewidthvalue].\hat{\asymbol}$ to each leaf node labeled with symbol $\asymbol$. 
Subsequently, the set of local witnesses to be assigned to each internal node $\aposition$ is computed by taking into 
consideration the label of the node, and the set (sets) of local witnesses assigned to the child (children) of $\aposition$.
The algorithm accepts $\abstractdecomposition$ if at the end of the process, the set of local witnesses associated with the
root node $\rootnode{\abstractdecomposition}$ has some final local witness, i.e., some local witness
$\awitness\in \allwitnesses$ such that
$\dpcore[\treewidthvalue].\finalwitnessgeneric(\awitness) = \truevalue$.

The function $\dpcore[\treewidthvalue].\cleaningfunctioncore$ is used to 
remove redundant local witnesses during the processing of a term. 
The function $\dpcore[\treewidthvalue].\invariantCore$ is useful in the context of optimization problems. For instance, 
given a set $\witnessset$ of local witnesses encoding weighted partial solutions to a given problem, 
$\dpcore[\treewidthvalue].\invariantCore(\witnessset)$ may return (a binary encoding of) the minimum/maximum weight of a
partial solution in the set. 

The process described above is formalized in our framework using the notion
of {\em $\treewidthvalue$-th dynamization} of a dynamic core 
$\dpcore$, which is a function $\dynamizationfunction{\dpcore}{\treewidthvalue}$
that assigns a set $\dynamizationfunction{\dpcore}{\treewidthvalue}(\abstractdecomposition)$ of local witnesses to each term 
$\abstractdecomposition\in \Terms{\treewidthvalue}$. 
Given a symbol $\asymbol$ of arity $\arityvalue$ in the set $\dpcore[\treewidthvalue].\alphabetclass$, and subsets 
$\witnessset_1,\dots,\witnessset_{r} \subseteq \dpcore[\treewidthvalue].\allwitnesses$,
we let
$\dpcore[\treewidthvalue].\hat{\asymbol}(\witnessset_1,\dots,\witnessset_{r})$
denote the following set: 
\begin{equation}
\label{equation:NextStepDynamization}
\dpcore[\treewidthvalue].\cleaningfunctioncore\left(\bigcup_{i\in [r],\awitness_i\in \witnessset_i}
\dpcore[\treewidthvalue].\hat{\asymbol}(\awitness_1,\dots,\awitness_r)\right).
\end{equation}
Using this notation, for each $\treewidthvalue\in \N$,
the function
$\dynamizationfunction{\dpcore}{\treewidthvalue}$ is defined by induction on the structure of 
$\abstractdecomposition$ as follows. 

\begin{definition}[Dynamization]
\label{definition:Dynamization}
Let $\dpcore$ be a treelike DP-core. For each $\treewidthvalue\in \N$, the $\treewidthvalue$-th 
dynamization of $\dpcore$ is the function
$\dynamizationfunction{\dpcore}{\treewidthvalue}:\Terms{\dpcore[\treewidthvalue].\alphabetclass} \rightarrow 
\finitepowerset{\dpcore[\treewidthvalue].\allwitnesses}$ inductively defined as follows. 
\begin{enumerate}
\setlength\itemsep{0em}
\item If $\aterm = \asymbol$ for some symbol $\asymbol\in \dpcore[\treewidthvalue].\alphabetclass$ of arity $0$, then 
	$\dynamizationfunction{\dpcore}{\treewidthvalue}(\abstractdecomposition)=\dpcore[\treewidthvalue].\hat{\asymbol}$. 
\item If $\abstractdecomposition = \asymbol(\abstractdecomposition_1,\dots,\abstractdecomposition_{r})$ for some
$\asymbol\in \dpcore[\treewidthvalue].\alphabetclass$ of arity $\arityvalue$,
and some terms $\abstractdecomposition_1,\dots,\abstractdecomposition_{\arityvalue}$ in $\Terms{\dpcore[\treewidthvalue].\alphabetclass}$, then 
$\dynamizationfunction{\dpcore}{\treewidthvalue}(\abstractdecomposition) = \dpcore[\treewidthvalue].\hat{\asymbol}(\dynamizationfunction{\dpcore}{\treewidthvalue}(\abstractdecomposition_1),\dots,\dynamizationfunction{\dpcore}{\treewidthvalue}(\abstractdecomposition_{\arityvalue})).$
\end{enumerate}
\end{definition}

For each $\treewidthvalue\in \N$, we say that a term $\abstractdecomposition\in \Terms{\dpcore[\treewidthvalue].\alphabetclass}$
is {\em accepted} by $\dpcore[\treewidthvalue]$ if the set $\dynamizationfunction{\dpcore}{\treewidthvalue}(\aterm)$ contains a final local
witness, i.e., a local witness $\awitness$ with $\dpcore[\treewidthvalue].\finalwitnessgenericcore(\awitness) = 1$. We let
$\accepteddecompositions{\dpcore[\treewidthvalue]}$ denote the set of all terms accepted by $\dpcore[\treewidthvalue]$. We let 
$\accepteddecompositions{\dpcore} = \bigcup_{\treewidthvalue\in \N} \accepteddecompositions{\dpcore[\treewidthvalue]}$.
The combination of the notion of a DP-core with the notion of a treelike 
decomposition class can be used to define graph properties. 

\begin{definition}[Graph Property of a DP-Core]
\label{definition:GraphClassFromCore}
Let $\decompositionclass$ be a treelike decomposition class, and $\dpcore$ be a treelike DP-core. For each $\treewidthvalue\in \N$, 
the graph property of $\dpcore[\treewidthvalue]$ is the set  
$$\dpcoregraphpropertyclass{\dpcore[\treewidthvalue]}{\decompositionclass} = 
\isomorphismclosure{\{\graphfunction(\abstractdecomposition)\;:\;
\abstractdecomposition\in \languageclass_k \cap 
\accepteddecompositions{\dpcore[\treewidthvalue]}\}}.$$
The graph property defined by 
$\dpcore$ is the set
$$\dpcoregraphpropertyclass{\dpcore}{\decompositionclass} =
\bigcup_{\treewidthvalue}
\dpcoregraphpropertyclass{\dpcore[\treewidthvalue]}{\decompositionclass}.$$
\end{definition}

We note that for each $\treewidthvalue\in \N$, $\dpcoregraphpropertyclass{\dpcore[\treewidthvalue]}{\decompositionclass}\subseteq \dpcoregraphproperty{\decompositionclass_{\treewidthvalue}}$, 
and hence, $\dpcoregraphpropertyclass{\dpcore}{\decompositionclass} \subseteq \dpcoregraphproperty{\decompositionclass}$. 

\subsection{Coherency}

In order to be useful in the context of model-checking and automated theorem proving, DP-cores need to
behave coherently with respect to distinct treelike decompositions of the same graph. This intuition is 
formalized by the following definition. 

\begin{definition}[Coherency]
\label{definition:Coherency}
Let $\decompositionclass = \{(\alphabetclass_k,\languageclass_k,\graphfunction_k)\}_{\treewidthvalue\in \N}$ be a treelike
decomposition class, and $\dpcore$ be a treelike DP-core. We say that $\dpcore$ is $\decompositionclass$-coherent
if for each $\treewidthvalue\in \N$, $\alphabetclass_k = \dpcore[\treewidthvalue].\alphabet$, and for each
$k,k'\in \N$, and each $\abstractdecomposition\in \languageclass_{\treewidthvalue}$ and $\abstractdecomposition'\in \languageclass_{\treewidthvalue'}$ 
with $\graphfunction(\abstractdecomposition)\simeq\graphfunction(\abstractdecomposition')$,
\begin{enumerate}
\item\label{coherencyOne} $\abstractdecomposition\in \accepteddecompositions{\dpcore[\treewidthvalue]}$ if and only if
$\abstractdecomposition'\in \accepteddecompositions{\dpcore[\treewidthvalue']}$, and
\item\label{coherencyTwo} $\dpcore[\treewidthvalue].\invariantCore(\dynamizationfunction{\dpcore}{\treewidthvalue}(\abstractdecomposition)) =
	\dpcore[\treewidthvalue'].\invariantCore(\dynamizationfunction{\dpcore}{\treewidthvalue'}(\abstractdecomposition'))$.
\end{enumerate}
\end{definition}

Let $\dpcore$ be a $\decompositionclass$-coherent treelike DP-core. Condition \ref{coherencyOne} of Definition \ref{definition:Coherency} guarantees
that if a graph $\agraph$ belongs to $\dpcoregraphproperty{\dpcore,\decompositionclass}$, then for each $\treewidthvalue\in \N$ and each 
$\decompositionclass$-decomposition $\abstractdecomposition$ of width at most $\treewidthvalue$ such that $\decompositiongraph{\abstractdecomposition}\simeq G$, 
we have that $\abstractdecomposition\in \accepteddecompositionsboundedwidth{\dpcore[\treewidthvalue]}$. On the other hand, if $G$
does not belong to $\dpcoregraphproperty{\dpcore,\decompositionclass}$, then no $\decompositionclass$-decomposition $\abstractdecomposition$ with 
$\decompositiongraph{\abstractdecomposition}\simeq G$ belongs to $\accepteddecompositionsboundedwidth{\dpcore}$. This discussion is formalized
in the following proposition. 

\begin{proposition}
\label{proposition:CoherencySequenceTuples}
Let $\decompositionclass = \{(\alphabetclass_k,\languageclass_k,\graphfunction_k)\}_{\treewidthvalue\in \N}$ be a treelike
decomposition class, and $\dpcore$ be a $\decompositionclass$-coherent treelike DP-core. Then for each $\treewidthvalue\in \N$, 
and each $\abstractdecomposition\in \languageclass_{\treewidthvalue}$, we have that
$\decompositiongraph{\abstractdecomposition}\in \dpcoregraphproperty{\dpcore,\decompositionclass}$
if and only if $\abstractdecomposition\in \accepteddecompositionsboundedwidth{\dpcore[\treewidthvalue]}$. 
\end{proposition} 
\begin{proof}
Let $\treewidthvalue\in \N$ and $\abstractdecomposition\in \languageclass_{\treewidthvalue}$. 
Suppose that $\abstractdecomposition\in \accepteddecompositionsboundedwidth{\dpcore[\treewidthvalue]}$. Then, 
by Definition \ref{definition:GraphClassFromCore}, $\decompositiongraph{\abstractdecomposition}\in \dpcoregraphproperty{\dpcore[\treewidthvalue],\decompositionclass}$, 
and therefore,
we have that 
$\decompositiongraph{\abstractdecomposition}\in \dpcoregraphproperty{\dpcore,\decompositionclass}$. Note that this direction holds even 
if $\dpcore$ is not coherent. 

For the converse, we do need coherency. Suppose $\decompositiongraph{\abstractdecomposition}\in \dpcoregraphproperty{\dpcore,\decompositionclass}$. 
Then, there is some $\treewidthvalue'\in \N$ and some $\abstractdecomposition'\in \accepteddecompositionsboundedwidth{\dpcore[\treewidthvalue']}$ 
such that $\decompositiongraph{\abstractdecomposition'}\simeq \decompositiongraph{\abstractdecomposition}$. Since $\dpcore$ is $\decompositionclass$-coherent,
we can infer from Definition \ref{definition:Coherency}.\ref{coherencyOne} that 
$\abstractdecomposition\in \accepteddecompositionsboundedwidth{\dpcore[\treewidthvalue]}$. 
\end{proof}

Coherent DP-cores may be used to define not only graph properties but also graph invariants, as specified in Definition \ref{definition:DPInvariant}.

\begin{definition}[Invariant of a DP-Core]
\label{definition:DPInvariant}
Let $\decompositionclass$ be a decomposition class and $\dpcore$ be a 
$\decompositionclass$-coherent treelike DP-core. The $\dpcoregraphproperty{\dpcore,\decompositionclass}$-invariant defined by $\dpcore$ is
the function 
$\invariant[\dpcore,\decompositionclass]:\dpcoregraphpropertyclass{\dpcore}{\decompositionclass}\rightarrow \{0,1\}^*$
that assigns to each graph $\agraph\in \dpcoregraphpropertyclass{\dpcore}{\decompositionclass}$, the string 
	$\dpcore[\treewidthvalue].\invariantCore(\dynamizationfunction{\dpcore}{\treewidthvalue}(\abstractdecomposition))$, 
where
$\abstractdecomposition$ is an arbitrary $\decompositionclass$-decomposition
such that $\decompositiongraph{\abstractdecomposition}\simeq \agraph$, and 
$\treewidthvalue = \widthdecomposition{\decompositionclass}{\abstractdecomposition}$ is the $\decompositionclass$-width of $\abstractdecomposition$.
\end{definition}

We note that Condition \ref{coherencyTwo} of Definition \ref{definition:Coherency}
requires that 
for any two $\decompositionclass$-decompositions $\abstractdecomposition$ and $\abstractdecomposition'$, of $\decompositionclass$-width 
$\treewidthvalue = \widthdecomposition{\decompositionclass}{\abstractdecomposition}$ and $\treewidthvalue'=\widthdecomposition{\decompositionclass}{\abstractdecomposition'}$, respectively, if $\decompositiongraph{\abstractdecomposition}\simeq \decompositiongraph{\abstractdecomposition'}$, then

\[\dpcore[\treewidthvalue].\invariantCore(\dynamizationfunction{\dpcore}{\treewidthvalue}(\abstractdecomposition)) = 
\dpcore[\treewidthvalue'].\invariantCore(\dynamizationfunction{\dpcore}{\treewidthvalue'}(\abstractdecomposition')).\]

Therefore, for each graph $\agraph\in \dpcoregraphproperty{\dpcore,\decompositionclass}$, the value 
$\invariant[\dpcore,\decompositionclass](\agraph)$ is well-defined, and invariant under graph isomorphism. This value is also independent of the $\decompositionclass$-decomposition of $\agraph$ chosen to perform the computation of the invariant.

In Appendix \ref{subsection:CombinatorsCombination}, we illustrate the notion of a coherent DP-core using the vertex cover problem as an example. More specifically, for each $r\in \N$, we define an $\instructivetreedecompositionclass$-coherent DP-core $\vertexcoverCore_{\vertexcoverParameter}$ corresponding to the graph property $\vertexcoverProperty_{\vertexcoverParameter}$, the set of graphs that have a vertex-cover of size at most $r$. Here, $\instructivetreedecompositionclass$ is the instructive tree decomposition class defined in Section \ref{DPRealization}, which realizes the width measure treewidth. For each $\treewidthvalue\in \N$, 
$\vertexcoverCore_{\vertexcoverParameter}[\treewidthvalue]$ accepts a 
$\treewidthvalue$-instructive tree 
decomposition $\abstractdecomposition$ 
if and only if the graph $\decompositiongraph{\abstractdecomposition}$ has a vertex cover of size at most $\vertexcoverParameter$. By adapting $\vertexcoverCore_{\vertexcoverParameter}$, 
we obtain another DP-core $\dpcore$ which can be used to compute the vertex cover number of a graph. More specifically, for each $\treewidthvalue\in \N$, and each $\treewidthvalue$-instructive tree decomposition $\abstractdecomposition$, the string $\dpcore[\treewidthvalue].\invariantCore(\dynamizationfunction{\dpcore}{\treewidthvalue}(\abstractdecomposition))$ 
is (a binary encoding of) the minimum size of a vertex cover in the graph
$\decompositiongraph{\abstractdecomposition}$.

\subsection{Complexity Measures}
\label{subsection:ComplexityMeasures}

In order to analyze the behavior of treelike DP-cores 
from a quantitative point of view, we define the notions of 
{\em bitlength} and {\em multiplicity} of a DP-core $\dpcore$. 
We say that a set $\witnessset$ of local witnesses is {\em $(\dpcore,\treewidthvalue,\sizedecomposition)$-useful}
if there is some $\abstractdecomposition\in \Terms{\dpcore[\treewidthvalue].\alphabet}$ of size $|\abstractdecomposition|\leq \sizedecomposition$ such that
$\dynamizationfunction{\dpcore}{\treewidthvalue}(\abstractdecomposition) =
\witnessset$. 
The {\em bitlength} of $\dpcore$ is the function
$\bitlengthusefulwitnesssimple{\dpcore}$ that assigns to each pair
$(\treewidthvalue,\sizedecomposition)$ the maximum number of bits 
$\bitlengthusefulwitnesssimple{\dpcore}(\treewidthvalue,\sizedecomposition)$
in a $(\dpcore,\treewidthvalue,\sizedecomposition)$-useful witness, while 
the {\em multiplicity} $\maxsizeusefulsetsimple{\dpcore}$ of $\dpcore$ is 
the function that assigns to each pair $(\treewidthvalue,\sizedecomposition)$
the maximum number of elements 
$\maxsizeusefulsetsimple{\dpcore}(\treewidthvalue,\sizedecomposition)$ 
in a $(\dpcore,\treewidthvalue,\sizedecomposition)$-useful set. 
Finally, we let $\numberusefulsetssimple{\dpcore}(\treewidthvalue,\sizedecomposition)$ denote the number of 
$(\dpcore,\treewidthvalue,\sizedecomposition)$-useful sets. We call $\numberusefulsetssimple{\dpcore}$ the {\em deterministic state complexity} 
of $\dpcore$. 

An important class of DP-cores is the class of cores
where the maximum number of bits in a useful local witness corresponding to a
term $\abstractdecomposition$ is independent of the size of 
$\abstractdecomposition$. In other words, the number of bits may 
depend on $\treewidthvalue$ but not on $|\abstractdecomposition|$.  

\begin{definition}[Finite DP-cores]
We say that a treelike DP-core $\dpcore$ is
{\em finite} if there is a function $f:\N\rightarrow \N$ such that 
for each $n\in \N$, $\bitlengthusefulwitnesssimple{\dpcore}(\treewidthvalue,n)\leq f(\treewidthvalue)$. 
\end{definition} 

If $\dpcore$ is a finite DP-core then we may write simply
$\bitlengthusefulwitnesssimple{\dpcore}(\treewidthvalue)$
and $\maxsizeusefulsetsimple{\dpcore}(\treewidthvalue)$
to denote the functions 
$\bitlengthusefulwitnesssimple{\dpcore}(\treewidthvalue,\sizedecomposition)$, 
and $\maxsizeusefulsetsimple{\dpcore}(\treewidthvalue,\sizedecomposition)$
respectively. 
\\ 

We say that a DP-core is {\em internally polynomial}, if there is an 
algorithm $\analgorithm$ that when given a number $\treewidthvalue\in \N$ as input,
simulates the functions in $\dpcore[\treewidthvalue]$ in time
polynomial in $\treewidthvalue$ and in the size of the input to these functions
(see Appendix \ref{section:formalDefinitionInternallyPolynomial} for a formal definition).
We note that typical dynamic programming algorithms operating on tree-like 
decompositions give rise to internally polynomial DP-cores.
Note that the fact that $\dpcore$ is internally polynomial does not imply that one can determine whether a given
term $\abstractdecomposition$ is accepted by $\dpcore$ in time polynomial in
$|\abstractdecomposition|$. The complexity of this test 
is governed by the bitlength and multiplicity of the DP-core in question
(see Theorem \ref{theorem:ModelChecking}).

\subsection{Model Checking and Invariant Computation}
\label{subsection:ModelChecking}

Let $\decompositionclass = \{(\alphabet_{\treewidthvalue},\languageclass_{\treewidthvalue},\graphfunction_{\treewidthvalue})\}_{\treewidthvalue\in \N}$ 
be a treelike decomposition class, and $\dpcore$ be a $\decompositionclass$-coherent treelike DP-core. Given a $\decompositionclass$-decomposition of width 
at most $\treewidthvalue$, we can use the notion of dynamization (Definition \ref{definition:Dynamization}), to check whether 
the graph $\decompositiongraph{\abstractdecomposition}$ encoded by $\abstractdecomposition$ belongs to the graph property 
$\dpcoregraphproperty{\dpcore,\decompositionclass}$ represented by $\dpcore$. 
The next theorem states that the complexity of this model-checking task is essentially governed by the bitlength and by the multiplicity of $\dpcore$.
We note that in typical applications the arity $\arityvalue$ of a decomposition class is a constant (most often 1 or 2), and the width $\treewidthvalue$ 
is smaller than $\bitlengthusefulwitnesssimple{\dpcore}(\treewidthvalue,\sizedecomposition)$ for each $n\in N$. 
Nevertheless, for completeness, we explicitly include the dependence on $\treewidthvalue^{O(1)}$ and $\arityvalue^{O(1)}$ in the calculation of the running time. 

\begin{theorem}[Model Checking]
\label{theorem:ModelChecking}
Let $\decompositionclass = \{(\alphabet_{\treewidthvalue},\languageclass_{\treewidthvalue},\graphfunction_{\treewidthvalue})\}_{\treewidthvalue\in \N}$
be a treelike decomposition class of arity $r$, $\dpcore$ be an internally polynomial $\decompositionclass$-coherent treelike
DP-Core, and let $\abstractdecomposition$ be a $\decompositionclass$-decomposition of $\decompositionclass$-width at most $\treewidthvalue$ and size $|\abstractdecomposition|=n$. 
\begin{enumerate}
\item One can determine whether $\decompositiongraph{\abstractdecomposition}
\in \dpcoregraphpropertyclass{\dpcore}{\decompositionclass}$ in time 
$$T(\treewidthvalue,\sizedecomposition) = \sizedecomposition \cdot \treewidthvalue^{O(1)} \cdot  \arityvalue^{O(1)}\cdot \bitlengthusefulwitnesssimple{\dpcore}(\treewidthvalue,\sizedecomposition)^{O(1)} \cdot \maxsizeusefulsetsimple{\dpcore}(\treewidthvalue,\sizedecomposition)^{\arityvalue+ O(1)}.$$
\item One can compute the invariant $\invariant[\dpcore,\decompositionclass](\decompositiongraph{\abstractdecomposition})$ in time
$$T(\treewidthvalue,\sizedecomposition)+ \treewidthvalue^{O(1)}\cdot \bitlengthusefulwitnesssimple{\dpcore}(\treewidthvalue,\sizedecomposition)^{O(1)}\cdot \maxsizeusefulsetsimple{\dpcore}(\treewidthvalue,\sizedecomposition)^{O(1)}.$$
\end{enumerate}
\end{theorem}

The proof of Theorem \ref{theorem:ModelChecking} can be found in Appendix \ref{section:ProofTheoremModelChecking}.

\subsection{Inclusion Test}
\label{subsection:InclusionTest}

Let $\decompositionclass$ be a treelike decomposition class, and $\dpcore$ be a treelike DP-core.
As discussed in the introduction, the problem of determining whether $\dpcoregraphproperty{\decompositionclass}\subseteq \dpcoregraphproperty{\dpcore,\decompositionclass}$ 
can be regarded as a task in the realm of automated theorem proving. A width-based approach to testing whether this inclusion holds is to test for increasing values of 
$\treewidthvalue$, whether the inclusion $\dpcoregraphproperty{\decompositionclass_{\treewidthvalue}} \subseteq \dpcoregraphproperty{\dpcore,\decompositionclass}$ holds.
It turns out that if $\dpcore$ is $\decompositionclass$-coherent, then testing whether $\dpcoregraphproperty{\decompositionclass_{\treewidthvalue}} \subseteq \dpcoregraphproperty{\dpcore,\decompositionclass}$ 
reduces to testing whether all $\decompositionclass$-decompositions of width at most $\treewidthvalue$ are accepted by $\dpcore[\treewidthvalue]$, as stated in Lemma \ref{lemma:UniversalityInclusion} below. 
We note that this is not necessarily true if $\dpcore$ is not $\decompositionclass$-coherent. 

\begin{lemma}
\label{lemma:UniversalityInclusion}
Let $\decompositionclass = \{(\alphabet_{\treewidthvalue},\languageclass_{\treewidthvalue},\graphfunction_{\treewidthvalue})\}_{\treewidthvalue\in \N}$
be a treelike decomposition class and $\dpcore$ be a $\decompositionclass$-coherent treelike DP-core. Then, for each $\treewidthvalue\in \N$, 
$\dpcoregraphproperty{\decompositionclass_{\treewidthvalue}}\subseteq \dpcoregraphpropertyclass{\dpcore}{\decompositionclass}$ if and only if
$\languageclass_{\treewidthvalue} \subseteq \accepteddecompositions{\dpcore[\treewidthvalue]}$. 
\end{lemma}
\begin{proof}
Suppose that $\languageclass_\treewidthvalue \subseteq \accepteddecompositionsboundedwidth{\dpcore[\treewidthvalue]}$.
Let $\agraph\in \dpcoregraphproperty{\decompositionclass_{\treewidthvalue}}$. Then, there is some $\treewidthvalue\in \N$, and 
some $\treewidthvalue$-instructive tree decomposition $\abstractdecomposition \in \languageclass_{\treewidthvalue}$, 
such that $\decompositiongraph{\abstractdecomposition}$ is isomorphic to $\agraph$. 
Since, by assumption, $\abstractdecomposition$ also belongs to $\accepteddecompositionsboundedwidth{\dpcore[\treewidthvalue]}$, 
we have that both $\decompositiongraph{\abstractdecomposition}$ and $\agraph$ belong to $\dpcoregraphproperty{\dpcore[\treewidthvalue],\decompositionclass}$. Therefore, 
$\agraph$ also belongs to $\dpcoregraphproperty{\dpcore,\decompositionclass}$.
Since $\agraph$ was chosen to be an arbitrary graph in $\dpcoregraphproperty{\decompositionclass_{\treewidthvalue}}$, we have that
$\dpcoregraphproperty{\decompositionclass_{\treewidthvalue}} \subseteq \dpcoregraphproperty{\dpcore,\decompositionclass}$. We note that for this direction, we did not need the assumption that $\dpcore$ is $\decompositionclass$-coherent. 

For the converse, we do need the assumption that $\dpcore$ is coherent. 
	Suppose that $\dpcoregraphproperty{\decompositionclass_{\treewidthvalue}} \subseteq
\dpcoregraphproperty{\dpcore,\decompositionclass}$. 
Let $\abstractdecomposition\in \languageclass_{\treewidthvalue}$. Then, by 
the definition of graph property associated with a treelike decomposition class, we have 
that $\decompositiongraph{\abstractdecomposition}$ belongs to $\dpcoregraphproperty{\decompositionclass_{\treewidthvalue}}$. 
Therefore, by our supposition, $\decompositiongraph{\abstractdecomposition}$ also belongs to 
$\dpcoregraphproperty{\dpcore,\decompositionclass}$. This means 
that for some $\treewidthvalue'$, there is some $\treewidthvalue'$-instructive tree decomposition $\abstractdecomposition'$ in 
$\accepteddecompositionsboundedwidth{\dpcore[\treewidthvalue']}$ such that 
$\decompositiongraph{\abstractdecomposition'}$ is isomorphic to
$\decompositiongraph{\abstractdecomposition}$. But since $\dpcore$ is coherent,
this implies that $\abstractdecomposition$ also belongs to
$\accepteddecompositionsboundedwidth{\dpcore[\treewidthvalue]}$. Since $\abstractdecomposition$ was chosen to 
be an arbitrary treelike decomposition in $\languageclass_{\treewidthvalue}$, we have that 
$\languageclass_{\treewidthvalue} \subseteq
\accepteddecompositionsboundedwidth{\dpcore[\treewidthvalue]}$.
\end{proof}

Lemma \ref{lemma:UniversalityInclusion} implies that if $\dpcore$ is coherent, then in order to show that 
$\dpcoregraphproperty{\decompositionclass_{\treewidthvalue}}\nsubseteq \dpcoregraphproperty{\dpcore,\decompositionclass}$
it is enough to show that there is some $\decompositionclass$-decomposition $\abstractdecomposition$ of width at most
$\treewidthvalue$ that belongs to $\languageclass_{\treewidthvalue}$ but not to $\accepteddecompositionsboundedwidth{\dpcore[\treewidthvalue]}{}$.
We will reduce this later task to the task of constructing a dynamic programming refutation (Definition \ref{definition:DPRefutation}). 

Let $\decompositionclass$ be a decomposition class with automation $\realizationclass$, and 
let $\dpcore$ be a $\decompositionclass$-coherent treelike DP-core. An $(\realizationclass,\dpcore,\treewidthvalue)$-pair
is a pair of the form $(\astate,\witnessset)$ where $\astate$ is a state of $\treeautomaton_{\treewidthvalue}$
and $\witnessset\subseteq \dpcore[\treewidthvalue].\allwitnesses$. We say that such pair 
$(\astate,\witnessset)$ is {\em $(\realizationclass,\dpcore,\treewidthvalue)$-inconsistent} if $\astate$ is a 
final state of $\realizationclass_{\treewidthvalue}$, but $\witnessset$ has no final local witness for $\dpcore$.

\begin{definition}[DP-Refutation]
\label{definition:DPRefutation}
Let $\decompositionclass = \{(\alphabet_{\treewidthvalue},\languageclass_{\treewidthvalue},\graphfunction_{\treewidthvalue})\}_{\treewidthvalue\in \N}$ 
be a decomposition class, $\realizationclass$ be an automation for $\decompositionclass$, $\dpcore$ be a $\decompositionclass$-coherent treelike DP-core, and 
$\treewidthvalue\in \N$. An $(\realizationclass,\dpcore,\treewidthvalue)$-refutation is a sequence of $(\realizationclass,\dpcore,\treewidthvalue)$-pairs
$$\dprefutation \equiv 
(\astate_1,\witnessset_1)(\astate_2,\witnessset_2)\dots (\astate_m,\witnessset_m)$$ satisfying the following conditions: 
\begin{enumerate}
\setlength\itemsep{0.5em}
\item \label{DPRefutation:Two} \label{condition: dp-refutation-one} $(\astate_m,\witnessset_m)$ is $(\realizationclass,\dpcore,\treewidthvalue)$-inconsistent. 
\item \label{DPRefutation:Three} \label{condition:dp-refutation-two} For each $i\in [m]$, 
\begin{enumerate}
	\item either $(\astate_i,\witnessset_i) = (\astate,\dpcore[\treewidthvalue].\hat{\asymbol})$ for some symbol $\asymbol$ of arity $0$ in $\alphabet_{\treewidthvalue}$, and some 
state $\astate$ such that $\asymbol\rightarrow \astate$ is a transition of $\treeautomaton_{\treewidthvalue}$, or 
\item $(\astate_i,\witnessset_i) = (\astate, \dpcore[\treewidthvalue].\hat{\asymbol}(\witnessset_{j_1},\dots,\witnessset_{j_{\arity(a)}}))$, for
some $j_1,\dots,j_{\arity(\asymbol)}<i$, some symbol $\asymbol\in \alphabet_{\treewidthvalue}$ of arity $\arity(\asymbol)>0$, and some state $\astate$ such that 
$\asymbol(\astate_{j_1},\dots,\astate_{j_{\arity(\asymbol)}})\rightarrow \astate$ is a transition of $\realizationclass_{\treewidthvalue}$. 
\end{enumerate}
\end{enumerate}
\end{definition}

%
The following theorem shows that if $\decompositionclass$ is a decomposition class with automation $\realizationclass$ and $\dpcore$ is a 
$\decompositionclass$-coherent treelike DP-core, then showing that 
$\dpcoregraphproperty{\decompositionclass}\nsubseteq \dpcoregraphproperty{\dpcore,\decompositionclass}$, 
is equivalent to showing the existence of some $(\realizationclass,\dpcore,\treewidthvalue)$-refutation.

\begin{theorem}
\label{theorem:DP-CoreRefutation}
Let $\decompositionclass = \{(\alphabet_{\treewidthvalue},\languageclass_{\treewidthvalue},\graphfunction_{\treewidthvalue})\}_{\treewidthvalue\in \N}$ 
be a decomposition class with automation $\realizationclass$, and $\dpcore$ be a $\decompositionclass$-coherent treelike DP-core. For each $\treewidthvalue\in \N$, 
we have that $\dpcoregraphproperty{\decompositionclass_{\treewidthvalue}}\nsubseteq \dpcoregraphproperty{\dpcore,\decompositionclass}$
if and only if some $(\realizationclass,\dpcore,\treewidthvalue)$-refutation exists. 
\end{theorem}
\begin{proof}
Suppose that $\dpcoregraphproperty{\decompositionclass_{\treewidthvalue}}\nsubseteq \dpcoregraphproperty{\dpcore,\decompositionclass}$.
Then, there is a graph $\agraph$ that belongs to $\dpcoregraphproperty{\decompositionclass_{\treewidthvalue}}$,
but not to $\dpcoregraphproperty{\dpcore,\decompositionclass}$. 
Since $\agraph\in \dpcoregraphproperty{\decompositionclass_{\treewidthvalue}}$, we have that for some 
$\abstractdecomposition\in \languageclass_{\treewidthvalue}$, $\agraph$ is isomorphic to 
$\decompositiongraph{\abstractdecomposition}$. 
Since 
$\agraph\notin \dpcoregraphpropertyclass{\dpcore}{\decompositionclass}$
we have that 
$\abstractdecomposition \notin \accepteddecompositions{\dpcore}$, and therefore, 
$\abstractdecomposition \notin \accepteddecompositions{\dpcore[\treewidthvalue]}$.

Let $\allsubterms(\abstractdecomposition) =\{\sigmaabstractdecomposition \;:\;\sigmaabstractdecomposition
\mbox{ is a subterm of }\abstractdecomposition\}$ be the set 
of subterms\footnote{Note that $|\allsubterms(\abstractdecomposition)|$ may be smaller than $|\abstractdecomposition|$,
since a given subterm may occur in several positions of $\abstractdecomposition$.} of $\abstractdecomposition$, and 
$\sigmaabstractdecomposition_1,\sigmaabstractdecomposition_2,\dots,\sigmaabstractdecomposition_{m}$ be a
topological ordering of the elements in $\allsubterms(\abstractdecomposition)$. Since we ordered the subterms
topologically, for each $i,j\in [m]$, if $\sigmaabstractdecomposition_i$ is a subterm of $\sigmaabstractdecomposition_j$, then $i\leq j$. Additionally, 
$\sigmaabstractdecomposition_m = \abstractdecomposition$.
Now, consider the 
sequence 
$$
\dprefutation \equiv (\astate_1,\witnessset_1)(\astate_2,\witnessset_2)\dots (\astate_m,\witnessset_m).
$$ 

Since $\sigmaabstractdecomposition_1,\sigmaabstractdecomposition_2,\dots,\sigmaabstractdecomposition_m$ are subterms of $\abstractdecomposition$ and ordered topologically, 
we have that for each $i\in[m]$, $\sigmaabstractdecomposition_i$ is either a symbol of arity zero or there is a symbol
$\asymbol$ of arity $\arity(\asymbol)>0$, and $j_{1},\dots,j_{\arity(\asymbol)} < i$ 
such that $\sigmaabstractdecomposition_i = \asymbol(\sigmaabstractdecomposition_{j_1},\dots,\sigmaabstractdecomposition_{j_{\arity(\asymbol)}})$. 
\begin{itemize}
\item If $\sigmaabstractdecomposition_i = \asymbol$ has arity zero, and $\astate$ is the unique state of $\treeautomaton_{\treewidthvalue}$ such that 
$\asymbol\rightarrow \astate$ is a transition of $\treeautomaton_{\treewidthvalue}$, then we set $(\astate_i,\witnessset_i) = (\astate,\dpcore[\treewidthvalue].\hat{\asymbol})$.
\item If $\sigmaabstractdecomposition_i = \asymbol(\sigmaabstractdecomposition_{j_1},\dots,\sigmaabstractdecomposition_{j_{\arity(\asymbol)}})$ for some symbol $\asymbol$ of arity $\arity(\asymbol)>0$, 
and $\astate$ is the unique state of $\treeautomaton_{\treewidthvalue}$ such that $\asymbol(\astate_{j_1},\dots,\astate_{j_{\arity(\asymbol)}})\rightarrow \astate$ is a transition of 
$\treeautomaton_\treewidthvalue$, 
we set $(\astate_i,\witnessset_i) = (\astate, \dpcore[\treewidthvalue].\hat{\asymbol}(\witnessset_{j_1},\dots,\witnessset_{j_{\arity(\asymbol)}}))$.
\end{itemize}

By construction, $\dprefutation$ satisfies Condition \ref{DPRefutation:Three} of Definition \ref{definition:DPRefutation}. 
Now, we know that $\abstractdecomposition\in \languageclass_\treewidthvalue$ and that 
$\astate_m$ is the state reached by $\abstractdecomposition$ in $\treeautomaton_k$, and therefore, $\astate_m$ is a final state. On the other hand, $\witnessset_m = \dynamizationfunction{\dpcore}{\treewidthvalue}(\abstractdecomposition)$ and  $\abstractdecomposition\notin \accepteddecompositions{\dpcore[\treewidthvalue]}$, and therefore, $\witnessset_m$ has no final local witness. Therefore, the pair $(\astate_m,\witnessset_m)$ is an $(\realizationclass,\dpcore,\treewidthvalue)$-inconsistent pair, so Condition \ref{DPRefutation:Two} of Definition \ref{definition:DPRefutation}
is satisfied. Consequently, the first direction of Theorem \ref{theorem:DP-CoreRefutation} is proved, i.e., $\dprefutation$ is an $(\realizationclass,\dpcore,\treewidthvalue)$-refutation.

For the converse, assume that
$\dprefutation \equiv (\astate_1,\witnessset_1)(\astate_2,\witnessset_2)\dots (\astate_m,\witnessset_m)$ is an $(\realizationclass,\dpcore,\treewidthvalue)$-refutation.
Using this refutation, we will construct a sequence 
of terms
$\sigmaabstractdecomposition_1,\sigmaabstractdecomposition_2,\dots,\sigmaabstractdecomposition_{m}$
with the following property: for each 
$i\in [m]$, $\sigmaabstractdecomposition_i \in \Terms{\alphabet_{\treewidthvalue}}$
and $\witnessset_i = \dynamizationfunction{\dpcore}{\treewidthvalue}(\sigmaabstractdecomposition_i)$.
Since $\astate_m$ is a final state for $\realizationclass_\treewidthvalue$ but 
$\witnessset_m$ has no final local witness for $\dpcore[\treewidthvalue]$, we have that 
$\sigmaabstractdecomposition_m$ is in $\languageclass_{\treewidthvalue}$ but not in $\accepteddecompositionsboundedwidth{\dpcore[\treewidthvalue]}$.
In other words, $\decompositiongraph{\sigmaabstractdecomposition_m}$ is in $\dpcoregraphproperty{\decompositionclass_\treewidthvalue}$ but not in  
$\dpcoregraphproperty{\dpcore[\treewidthvalue],\decompositionclass}$. 
Since $\dpcore$ is $\decompositionclass$-coherent, we have that for each $\treewidthvalue'\in \N$, there is no term $\sigma'\in \accepteddecompositionsboundedwidth{\dpcore[\treewidthvalue']}$ 
with $\decompositiongraph{\sigma_m} \simeq \decompositiongraph{\sigma'}$ (otherwise, $\sigma_m$ would belong to $\accepteddecompositionsboundedwidth{\dpcore[\treewidthvalue]}$).
 Therefore, $\decompositiongraph{\sigmaabstractdecomposition_m}$ is not 
in $\dpcoregraphproperty{\dpcore,\decompositionclass}$ either. We infer that $\dpcoregraphproperty{\decompositionclass_{\treewidthvalue}}\nsubseteq \dpcoregraphproperty{\dpcore,\decompositionclass}$.

Now, for each $i\in \N$, the construction of $\sigmaabstractdecomposition_i$ proceeds as follows. If there is a symbol $\asymbol\in \alphabet_{\treewidthvalue}$ of arity 
$0$ such that $\asymbol\rightarrow \astate_i$ is a transition of $\treeautomaton_{\treewidthvalue}$, and $\witnessset_i = \dpcore[\treewidthvalue].\hat{\asymbol}$, 
then  we let $\sigmaabstractdecomposition_i = \asymbol$. On the other hand, if there is a symbol $a\in \alphabet_{\treewidthvalue}$ of arity 
$\arity(\asymbol)>0$ and some $j_1,\dots,j_r<i$ such that $\asymbol(\astate_{j_1},\dots,\astate_{j_{\arity(\asymbol)}}) \rightarrow \astate_i$ is a transition of $\realizationclass_{\treewidthvalue}$
and $\witnessset_i = \dpcore[\treewidthvalue].\hat{\asymbol}(\witnessset_{j_1},\dots,\witnessset_{j_{\arity(\asymbol)}})$, 
then we let $\sigmaabstractdecomposition_i = \asymbol(\sigmaabstractdecomposition_{j_1},\dots,\sigmaabstractdecomposition_{j_{\arity(\asymbol)}})$.
It should be clear that for each $i\in [m]$, $\sigmaabstractdecomposition_i$ is a term in $\Terms{\alphabet_{\treewidthvalue}}$. Furthermore, 
using Definition \ref{definition:Dynamization}, it follows by induction on $i$ that for each $i\in [m]$, 
$\witnessset_i = \dynamizationfunction{\dpcore}{\treewidthvalue}(\sigmaabstractdecomposition_i)$.
This concludes the proof of the theorem.
\end{proof}

Theorem \ref{theorem:DP-CoreRefutation} implies 
the existence of a simple forward-chaining style algorithm 
for determining whether $\dpcoregraphproperty{\decompositionclass_{\treewidthvalue}} \subseteq \dpcoregraphproperty{\dpcore,\decompositionclass}$ when $\dpcore$
is a {\em finite} and $\decompositionclass$-coherent treelike DP-core. 

\begin{theorem}[Inclusion Test]
\label{theorem:Inclusion}
Let $\decompositionclass = \{(\alphabet_{\treewidthvalue},\languageclass_{\treewidthvalue},\graphfunction_{\treewidthvalue})\}_{\treewidthvalue\in \N}$ 
be a treelike decomposition class with automation $\realizationclass$, and let 
$\dpcore$ be a finite, internally polynomial $\decompositionclass$-coherent 
treelike DP core. Let $r$ be the arity of $\decompositionclass$ and $f(k)$ be 
the complexity of $\realizationclass$. Then, one can determine whether 
$\dpcoregraphproperty{\decompositionclass_\treewidthvalue}\subseteq \dpcoregraphproperty{\dpcore,\decompositionclass}$
in time 
%
%
$$f(\treewidthvalue)^{O(r)}\cdot 2^{O(r\cdot \bitlengthusefulwitnesssimple{\dpcore}(\treewidthvalue)\cdot \maxsizeusefulsetsimple{\dpcore}(\treewidthvalue))} 
\leq f(\treewidthvalue)^{O(r)}\cdot 2^{r\cdot 2^{O(\bitlengthusefulwitnesssimple{\dpcore}(\treewidthvalue))})}.$$

\end{theorem}

The proof of Theorem \ref{theorem:Inclusion} can be found in 
Appendix \ref{section:proofTheoremInclusionTest}.

\subsection{On the Size of Counter-Examples}
\label{subsection:SizeCounterExamples}

The proof of Theorem \ref{theorem:DP-CoreRefutation} provides us with an
algorithm to extract, from a given $(\realizationclass,\dpcore,\treewidthvalue)$-refutation $\dprefutation$, a $\decompositionclass$-decomposition $\abstractdecomposition$
of width at most $\treewidthvalue$ such that $\decompositiongraph{\abstractdecomposition}\notin \dpcoregraphproperty{\dpcore,\decompositionclass}$.
The graph $\decompositiongraph{\abstractdecomposition}$ corresponding to $\abstractdecomposition$ may be regarded as a counter-example for the
conjecture $\dpcoregraphproperty{\decompositionclass} \subseteq \dpcoregraphproperty{\dpcore,\decompositionclass}$. 
Note that if the refutation has length $m$, then the height of the $\decompositionclass$-decomposition $\abstractdecomposition$ is at most $m-1$. Here, by height we mean the length of the longest path between a leaf node of $\abstractdecomposition$ and the root node. This implies that if  $\decompositionclass$
is a treelike decomposition class of arity $\arityvalue$, the number of nodes of $\abstractdecomposition$ is at most $m$, if $\arityvalue=1$, and at most 
$\frac{\arityvalue^m-1}{\arityvalue-1}$ nodes, if $r>1$. 

\begin{corollary}
\label{corollary:DP-CoreRefutation}
Let $\decompositionclass = \{(\alphabet_k,\languageclass_k,\graphfunction_k)\}_{\treewidthvalue\in \N}$ 
be a treelike decomposition class of arity $\arityvalue$ with automation $\realizationclass$, $\dpcore$ be a $\decompositionclass$-coherent treelike DP-core, and 
$\dprefutation \equiv (\astate_1,\witnessset_1)(\astate_2,\witnessset_2)\dots 
(\astate_m,\witnessset_m)$ be a $(\realizationclass,\dpcore,\treewidthvalue)$-refutation. 
Then, there is a $\decompositionclass$-decomposition $\abstractdecomposition\in \languageclass_k$ of height at most $m-1$ such that
$\decompositiongraph{\abstractdecomposition} \in \dpcoregraphproperty{\decompositionclass_{\treewidthvalue}}\backslash \dpcoregraphproperty{\dpcore,\decompositionclass}$. Additionally, $\abstractdecomposition$ has at most $m$ nodes if $r=1$, and at most $\frac{\arityvalue^m-1}{\arityvalue-1}$ nodes if $r>1$.
\end{corollary}

Since the search space in the proof of Theorem \ref{theorem:Inclusion} has at most $f(\treewidthvalue)\cdot \numberusefulsetssimple{\dpcore}(\treewidthvalue)$
distinct $(\realizationclass,\dpcore,\treewidthvalue)$-pairs, a minimum-length  $(\realizationclass,\dpcore,\treewidthvalue)$-refutation
has length at most $f(\treewidthvalue)\cdot \numberusefulsetssimple{\dpcore}(\treewidthvalue)$. Therefore, this fact 
together with Corollary \ref{corollary:DP-CoreRefutation} implies the following result.

\begin{corollary}
\label{corollary:SizeCounterexample}
Let $\decompositionclass = \{(\alphabet_k,\languageclass_k,\graphfunction_k)\}_{\treewidthvalue\in \N}$ be a treelike decomposition class of complexity 
$f(\treewidthvalue)$, and let $\dpcore$ be a finite, $\decompositionclass$-coherent treelike DP core. 
If $\dpcoregraphproperty{\decompositionclass_\treewidthvalue}\nsubseteq \dpcoregraphproperty{\dpcore,\decompositionclass}$, 
then  
there is a $\decompositionclass$-decomposition $\abstractdecomposition\in \languageclass_k$ of height at most $f(\treewidthvalue)\cdot \numberusefulsetssimple{\dpcore}(\treewidthvalue)-1$ such that
$\decompositiongraph{\abstractdecomposition} \in \dpcoregraphproperty{\decompositionclass_{\treewidthvalue}}\backslash \dpcoregraphproperty{\dpcore,\decompositionclass}$. Additionally, $\abstractdecomposition$ has 
at most  $f(\treewidthvalue)\cdot \numberusefulsetssimple{\dpcore}(\treewidthvalue)$ nodes if $\arityvalue=1$,
 and at most $\arityvalue^{f(\treewidthvalue)\cdot \numberusefulsetssimple{\dpcore}(\treewidthvalue)}$ nodes if $\arityvalue>1$. 
\end{corollary}

The requirement that the DP-core $\dpcore$ in Theorem \ref{theorem:Inclusion} is finite can be relaxed if instead of asking whether
$\dpcoregraphproperty{\decompositionclass_{\treewidthvalue}}\subseteq \dpcoregraphproperty{\dpcore,\decompositionclass}$, we ask whether
all graphs in $\dpcoregraphproperty{\decompositionclass_{\treewidthvalue}}$ that can be represented by a $\decompositionclass$-decomposition
of size at most $\sizedecomposition$ belong to $\dpcoregraphproperty{\dpcore,\decompositionclass}$. 

\begin{corollary}[Bounded-Size Inclusion Test]
\label{corollary:InclusionBoundedSize}
Let $\decompositionclass$ be a treelike decomposition class of complexity $f(\treewidthvalue)$ and arity $r$,
and let $\dpcore$ be a (not necessarily finite) internally polynomial $\decompositionclass$-coherent treelike DP core. One can determine 
in time 
	$$f(\treewidthvalue)^{O(r)}\cdot 2^{O(r\cdot \bitlengthusefulwitnesssimple{\dpcore}(\treewidthvalue,\sizedecomposition)\cdot \maxsizeusefulsetsimple{\dpcore}(\treewidthvalue,\sizedecomposition))}\cdot \sizedecomposition^{O(1)}$$ 
whether every graph corresponding to a $\decompositionclass$-decomposition of width at most $\treewidthvalue$ and size at most $\sizedecomposition$ belongs to 
$\dpcoregraphproperty{\dpcore,\decompositionclass}$. 
\end{corollary}

We note that whenever $\maxsizeusefulsetsimple{\dpcore}(\treewidthvalue,\sizedecomposition) = h_1(\treewidthvalue)$ for some function $h_1:\N\rightarrow \N$, and 
$\bitlengthusefulwitnesssimple{\dpcore}(\treewidthvalue,\sizedecomposition) = h_2(\treewidthvalue)\cdot \log n$ for some function $h_2:\N\rightarrow \N$, then
the running time stated in Corollary \ref{corollary:InclusionBoundedSize} is of the form $\sizedecomposition^{h_3(\treewidthvalue)}$ for some function $h_3:\N\rightarrow \N$.
This is significantly faster than the naive brute-force approach of enumerating all terms of width at most $\treewidthvalue$, and size at most $n$, and subsequently 
testing whether these terms belong to $\dpcoregraphproperty{\dpcore,\decompositionclass}$.

\subsection{Combinators and Combinations}
\label{subsection:CombinatorsCombination}

Given a graph property $\graphproperty$, and a graph
$\agraph\in \allgraphs$, we let $\indicatorfunction{\graphproperty}{\agraph}$ 
denote the Boolean value {\em true} if
$\agraph\in \graphproperty$ and the value {\em false},
 if $\agraph\notin \graphproperty$. 
For each $\numberproperties\in \N$ we call a function of the form 
$$\combinator:\{0,1\}^{\numberproperties}\times (\{0,1\}^*)^{\numberproperties}\rightarrow \{0,1\}.$$
an {\em $\numberproperties$-combinator}. Given graph properties 
$\graphproperty_1,\dots,\graphproperty_{\numberproperties}$
and graph invariants $\invariant_1\dots,\invariant_{\numberproperties}$, 
we let
$$\combinatorproperty{\combinator}(\graphproperty_1,\dots,\graphproperty_{\numberproperties},\invariant_1,\dots,\invariant_{\numberproperties})$$
denote the graph property consisting of all graphs $\agraph$ such that 
$$\combinator(\indicatorfunction{\graphproperty_1}{\agraph},\dots,\indicatorfunction{\graphproperty_{\numberproperties}}{\agraph},\invariant_1(\agraph),\dots,\invariant_{\numberproperties}(\agraph))=1.$$
We say that $\combinator$ is {\em polynomial} if it can be computed
in time $O_{\numberproperties}(|X|^c)$ for some constant $c$ on any given input $X$. 

Intuitively, a combinator is a tool to define graph classes in terms 
of previously defined graph classes and previously defined graph invariants.
It is worth noting that Boolean combinations of graph classes can be straightforwardly defined 
using combinators. Nevertheless, one can do more than that, since 
combinators can also be used to establish relations between graph 
invariants. For instance, using combinators one can define the class of 
graphs whose {\em covering number} (the smallest size of a vertex-cover) is 
equal to the {\em dominating number} (the smallest size of a dominating set).
This is just a illustrative example. Other examples of invariants 
that can be related using combinators are: {\em clique number}, {\em independence number},
{\em chromatic number}, {\em diameter}, and many others.  
Next, we will use combinators as a tool to combine graph properties and 
graph invariants defined using DP-cores.

\begin{theorem}
\label{theorem:CombinatorsCore}
Let $\combinator$ be and $\numberproperties$-combinator, $\decompositionclass$ be a treelike decomposition class, and $\dpcore_1,\dots,\dpcore_{\numberproperties}$
be $\decompositionclass$-coherent treelike DP-cores. Then, there exists a $\decompositionclass$-coherent treelike DP-core 
$\dpcore = \dpcore(\combinator,\dpcore_1,\dots,\dpcore_{\numberproperties})$ 
satisfying the following properties: 
\begin{enumerate}
\item 
	$\dpcoregraphpropertyclass{\dpcore}{\decompositionclass} = \combinator(\dpcoregraphproperty{\dpcore_1,\decompositionclass},\dots,\dpcoregraphproperty{\dpcore_{\numberproperties},\decompositionclass},\invariant[\dpcore_1,\decompositionclass],\dots,\invariant[\dpcore_{\numberproperties},\decompositionclass])$. 
\item $\dpcore$ has bitlength $\bitlengthusefulwitnesssimple{\dpcore}(\treewidthvalue,\sizedecomposition) = \sum_{i = 1}^{\numberproperties} \bitlengthusefulwitnesssimple{\dpcore_i}(\treewidthvalue,\sizedecomposition)\cdot \maxsizeusefulsetsimple{\dpcore_i}(\treewidthvalue,\sizedecomposition)$. 
\item $\dpcore$ has multiplicity $\maxsizeusefulsetsimple{\dpcore}(\treewidthvalue,\sizedecomposition)=1$.  
\item $\dpcore$ has d.s.c. 
$\numberusefulsetssimple{\dpcore}(\treewidthvalue,\sizedecomposition) 
\leq
\prod_{i=1}^{\numberproperties}
\numberusefulsetssimple{\dpcore_{i}}(\treewidthvalue,\sizedecomposition).
$
\end{enumerate}
\end{theorem}
\begin{proof}
We let the $\combinator$-combination of $(\dpcore_1,\dots,\dpcore_{\numberproperties})$ be 
the DP-core $\dpcore$, where for each $\treewidthvalue\in \N$, the components of the tuple $\dpcore[\treewidthvalue]$ are specified below.
Here, we let $\vertexone,\vertextwo\in [\treewidthvalue+1]$, and 
$\boldS = (\witnessset_1,\dots,\witnessset_{\numberproperties})$ and 
$\boldS' = (\witnessset_1', \dots,\witnessset_{\numberproperties}')$  be 
tuples in $\finitepowerset{\dpcore_1[\treewidthvalue].\allwitnesses}\times \dots \times 
\finitepowerset{\dpcore_{\numberproperties}[\treewidthvalue].\allwitnesses}$. For each $i\in [\numberproperties]$, we let $F(\dpcore_i[\treewidthvalue],\witnessset_i)$ be the Boolean value $1$ if $\witnessset_i$ has a final local witness and $0$ otherwise. 
\begin{enumerate}
\item $\dpcore[\treewidthvalue].\allwitnesses = \finitepowerset{\dpcore_1[\treewidthvalue].\allwitnesses}\times \dots \times 
\finitepowerset{\dpcore_{\numberproperties}[\treewidthvalue].\allwitnesses}$.
\item $\dpcore[\treewidthvalue].\leaftype =  \{(\dpcore_1[\treewidthvalue].\leaftype,\dots,\dpcore_{\numberproperties}[\treewidthvalue].\leaftype)\}$. 
\item $\dpcore[\treewidthvalue].\introvertexgeneric{\vertexone}(\boldS) =\{(\dpcore_1[\treewidthvalue].\introvertexgeneric{\vertexone}(\witnessset_1),\dots,
	 \dpcore_{\numberproperties}[\treewidthvalue].\introvertexgeneric{\vertexone}(\witnessset_{\numberproperties}))\}$.
\item $\dpcore[\treewidthvalue].\forgetvertexgeneric{\vertexone}(\boldS) =\{(\dpcore_1[\treewidthvalue].\forgetvertexgeneric{\vertexone}(\witnessset_1),\dots,
	\dpcore_{\numberproperties}[\treewidthvalue].\forgetvertexgeneric{\vertexone}(\witnessset_{\numberproperties}))\}$.
\item $\dpcore[\treewidthvalue].\introedgegeneric{\vertexone}{\vertextwo}(\boldS) =   \{(\dpcore_1[\treewidthvalue].\introedgegeneric{\vertexone}{\vertextwo}(\witnessset_1),\dots,
\dpcore_{\numberproperties}[\treewidthvalue].\introedgegeneric{\vertexone}{\vertextwo}(\witnessset_{\numberproperties}))\}$.
\item $\dpcore[\treewidthvalue].\joingeneric(\boldS,\boldS') = \{(\dpcore_1[\treewidthvalue].\joingeneric(\witnessset_1,\witnessset_1'),\dots,
      \dpcore_{\numberproperties}[\treewidthvalue].\joingeneric(\witnessset_{\numberproperties},\witnessset_{\numberproperties}'))\}$.
\item $\dpcore[\treewidthvalue].\cleaningfunctioncore(\{\boldS\})=\{(\dpcore_1[\treewidthvalue].\cleaningfunctioncore(\witnessset_1),\dots, \dpcore_{\numberproperties}[\treewidthvalue].\cleaningfunctioncore(\witnessset_{\numberproperties}))\}$. 
\item $\dpcore[\treewidthvalue].\finalwitnessgenericcore(\boldS)=\combinator(F(\dpcore_1[\treewidthvalue],\witnessset_1),\dots,F(\dpcore_{\numberproperties}[\treewidthvalue],\witnessset_{\numberproperties}), 
\dpcore_1[\treewidthvalue].\invariantCore(\witnessset_1),\dots,\dpcore_{\numberproperties}[\treewidthvalue].\invariantCore(\witnessset_{\numberproperties}))$. 
\item $\dpcore[\treewidthvalue].\invariantCore(\{\boldS\})= (\dpcore_1[\treewidthvalue].\invariantCore(\witnessset_1),\dots,\dpcore_{\numberproperties}[\treewidthvalue].\invariantCore(\witnessset_{\numberproperties}))$. 
\end{enumerate}
The upper bounds for the functions $\bitlengthusefulwitnesssimple{\dpcore}(\treewidthvalue,\sizedecomposition)$, 
$\maxsizeusefulsetsimple{\dpcore}(\treewidthvalue,\sizedecomposition)$, 
and  $\numberusefulsetssimple{\dpcore}(\treewidthvalue,\sizedecomposition)$, 
can be inferred directly from this construction. Finally,
since the DP-cores $\dpcore_1,\dots,\dpcore_{\numberproperties}$
are $\decompositionclass$-coherent, the DP-core 
$\dpcore$ is also $\decompositionclass$-coherent. 
\end{proof}

We call the DP-core $\dpcore = \dpcore(\combinator,\dpcore_1,\dots,\dpcore_{\numberproperties})$ the $\combinator$-combination of 
$\dpcore_1,\dots,\dpcore_{\numberproperties}$.

If the DP-cores $\dpcore_1,\dots,\dpcore_{\numberproperties}$ are also {\em finite}, 
besides being $\decompositionclass$-coherent, and internally polynomial, then
Theorem \ref{theorem:CombinatorsCore} together with Theorem \ref{theorem:Inclusion} 
directly imply the following theorem, which will be used in Section
\ref{section:Applications} to establish analytic upper bound on the time necessary to 
verify long-standing conjectures on graphs of bounded treewidth. 

\begin{theorem}[Inclusion Test for Combinations]
\label{theorem:InclusionTestCombinators}
	Let $\decompositionclass$ be a treelike decomposition class of arity $\arityvalue$; $\dpcore_1,\dots,\dpcore_{\numberproperties}$ be {\em finite}, internally polynomial, 
$\decompositionclass$-coherent treelike DP-cores; and $\combinator$ be a polynomial $\numberproperties$-combinator. Let $\dpcore = \dpcore(\combinator,\dpcore_1,\dots,\dpcore_{\numberproperties})$
be the $\combinator$-combination of $\dpcore_1,\dots,\dpcore_{\numberproperties}$, $\beta(\treewidthvalue) = \max_{i}\bitlengthusefulwitnesssimple{\dpcore_i}(\treewidthvalue)$
and $\mu(\treewidthvalue) = \max_{i}\maxsizeusefulsetsimple{\dpcore_i}(\treewidthvalue)$.
	Then, for each $\treewidthvalue\in \N$, one can determine whether $\dpcoregraphproperty{\decompositionclass_{\treewidthvalue}} \subseteq \dpcoregraphproperty{\dpcore,\decompositionclass}$
in time 
$$f(\treewidthvalue)^{O(\arityvalue)}\cdot 2^{O(\numberproperties\cdot \arityvalue\cdot \beta(\treewidthvalue)\cdot \mu(\treewidthvalue))} 
	\leq f(\treewidthvalue)^{O(\arityvalue)}\cdot 2^{\numberproperties\cdot \arityvalue\cdot 2^{O(\beta(\treewidthvalue))})}.$$
	
\end{theorem}

We note that in typical applications, the parameters $\arityvalue$ and $\numberproperties$ are constant, while the growth of the function $f(\treewidthvalue)$ is negligible when 
compared with $2^{O(\beta(\treewidthvalue)\cdot \mu(\treewidthvalue))}$. Therefore, in these applications, the running time of our algorithm is of the 
form  $2^{O(\beta(\treewidthvalue)\cdot \mu(\treewidthvalue))} \leq 2^{2^{O(\beta(\treewidthvalue))}}$. It is also worth noting that if 
$\dpcoregraphproperty{\decompositionclass_k}\nsubseteq \dpcoregraphproperty{\dpcore,\decompositionclass}$, then 
there a $\decompositionclass$-decomposition $\abstractdecomposition\in \languageclass_k$ of height 
at most $2^{O(\beta(\treewidthvalue)\cdot \mu(\treewidthvalue))}$ such 
that $\decompositiongraph{\abstractdecomposition}\in \dpcoregraphproperty{\decompositionclass_{\treewidthvalue}}\backslash \dpcoregraphproperty{\dpcore,\decompositionclass}$. Intuitively, the graph 
$\decompositiongraph{\abstractdecomposition}$ is a counter-example 
for the statement $\dpcoregraphproperty{\decompositionclass_k}\subseteq \dpcoregraphproperty{\dpcore,\decompositionclass}$.

\section{Applications of Theorem \ref{theorem:InclusionTestCombinators}}
\label{section:Applications}
In this section, we show that Theorem \ref{theorem:InclusionTestCombinators} 
can be used to show that several long-standing graph-theoretic conjectures 
can be tested in time double exponential in $\treewidthvalue^{O(1)}$ on the class 
of graphs of treewidth at most $\treewidthvalue$. The next theorem enumerates
upper bounds on the bitlength and multiplicity of DP-cores deciding several
graph properties.
These upper bounds are obtained by translating combinatorial dynamic programming
algorithms for these properties parameterized by treewidth into 
internally polynomial, $\instructivetreedecompositionclass$-coherent, 
finite DP-cores. 
Here, $\instructivetreedecompositionclass$ is the class of instructive tree decompositions introduced in Section \ref{DPRealization}.
This class has complexity $2^{k}$. In our examples, we will use the following graph properties. 

\begin{enumerate}
	\item $\texttt{Simple}$: the set of all simple graphs (i.e. without multiedges). 
	\item $\texttt{MaxDeg}_{\geq}(c)$: the set of graphs containing at least one vertex of degree at least $c$. 
	\item $\texttt{MinDeg}_{\leq}(c)$: the set of 
	graphs containing at least one vertex of degree at most $c$. 
	\item $\texttt{Colorable(c)}$:  the set of graphs of chromatic number at most $c$.
	\item $\texttt{Conn}$:  set of connected graphs. 
	\item $\texttt{VConn}_{\leq}(c)$:  the set of graphs with vertex-connectivity at most $c$. A graph is $c$-vertex-connected 
	if it has at least $c$ vertices, and if it remains connected whenever fewer than $c$ vertices are deleted. 
	\item $\texttt{EConn}_{\leq}(c)$:  the  set of graphs with edge-connectivity at most $c$. A graph is $c$-edge-connected
	if it remains connected whenever fewer than $c$ edges are deleted. 
	\item $\texttt{Hamiltonian}$: the set of Hamiltonian graphs. A graph is Hamiltonian if it contains a cycle that spans all its vertices.  
	\item $\texttt{NZFlow}(\Z_{\flowValue})$:  the set of graphs that admit a $\Z_{\flowValue}$-flow. 
	Here, $\Z_{\flowValue} = \{0,\dots,\flowValue-1\}$ is the set of integers modulo $\flowValue$. 
	A graph $G$ admits a nowhere-zero $\Z_{\flowValue}$-flow if one can assign to each edge an orientation and a non-zero element of $\Z_{\flowValue}$
	in such a way that for each vertex, the sum of values associated with edges entering the vertex is equal 
	to the sum of values associated with edges leaving the vertex. 
	\item $\texttt{Minor}(\bgraph)$: the set of graphs containing $\bgraph$ as a minor.
	A graph $\bgraph$ is a {\em minor} of a graph $\agraph$ if $\bgraph$ can be obtained from
	$\agraph$ by a sequence of vertex/edge deletions and edge contractions.
	\end{enumerate}
 
\begin{theorem}
\label{theorem:Estimates}
Let $\instructivetreedecompositionclass$ be the instructive tree decomposition class defined in Section \ref{DPRealization}. 
The properties specified above have $\instructivetreedecompositionclass$-coherent DP-cores 
with complexity parameters (bitlength $\beta$, multiplicity $\mu$, state complexity $\nu$, deterministic state complexity $\delta$) as 
specified in Table \ref{table:MeasuresCores}. 
\end{theorem}
\vspace{-10pt}
{\small
\begin{table}[h]
\centering
\begin{tabular}{|ccc|} 
\hline 
Property & $\beta(\treewidthvalue)$ & $\mu(\treewidthvalue)$ \\
\hline 
$\texttt{Simple}$ & $O(k^2)$ & $1$  \\
$\texttt{MaxDeg}_{\geq}(c)$ & $O(k\cdot \log c)$ & $1$  \\
$\texttt{MinDeg}_{\leq}(c)$ & $O(k\cdot \log c)$ & $1$ \\
$\texttt{Colorable(c)}$ & $O(k\log c)$ & $2^{O(\beta(k))}$ \\
$\texttt{Conn}$ & $O(k\log k)$ & $2^{O(\beta(k))}$ \\
$\texttt{VConn(c)}$ & $O(\log c + k\log k)$ & $2^{O(\beta(k))}$ \\
$\texttt{EConn(c)}$ & $O(\log c + k\log k)$ & $2^{O(\beta(k))}$ \\
$\texttt{Hamiltonian}$ & $O(k\log k)$ & \textcolor{red}{$2^{O(k)}$} \\
$\texttt{NZFlow}(\Z_{\flowValue})$ & $O(k\log \flowValue)$ & $2^{O(\beta(k))}$ \\
$\texttt{Minor(H)}$ & $O(\treewidthvalue\log
\treewidthvalue+|\vertexset{H}|+|\edgeset{H}|)$ &  $2^{\beta(k)}$ \\ 
 \hline 
\end{tabular}
\caption{Complexity measures for DP-cores deciding several graph 
properties. \label{table:MeasuresCores}}
\end{table}
}

The proof of Theorem \ref{theorem:Estimates} can be found in Appendix \ref{section:ProofTheoremEstimates}. 
Note that 
in the case of the DP-core $\texttt{C-Hamiltonian}$ the multiplicity $2^{O(\treewidthvalue)}$ is smaller than the trivial 
upper bound of $2^{O(\treewidthvalue\cdot \log \treewidthvalue)}$ and consequently, the deterministic state complexity 
$2^{2^{O(k)}}$ is smaller than the trivial upper bound of $2^{2^{O(\treewidthvalue\cdot \log \treewidthvalue)}}$. We note 
that the proof of this fact is a consequence of the rank-based approach developed in \cite{bodlaender2015deterministic}. 
Next, we will show how Theorem \ref{theorem:InclusionTestCombinators} together with Theorem \ref{theorem:Estimates} can be used 
to provide double-exponential upper bounds on the time necessary to verify long-standing
graph-theoretic conjectures on graphs of treewidth at most $\treewidthvalue$. If such a conjecture 
is false, then Corollary \ref{corollary:SizeCounterexample} can be used to establish an upper bound on minimum height of a term representing 
a counterexample for the conjecture.\\ 

\noindent{\bf Hadwiger's Conjecture.} This conjecture states that for each $\numbercolors\geq 1$, every graph with 
no $K_{\numbercolors+1}$-minor has a $c$-coloring \cite{hadwiger1943klassifikation}. 
This conjecture, which suggests a far-reaching generalization of the $4$-colors theorem, 
is considered to be one of the most important open problems in graph theory. The 
conjecture has been resolved in the positive for the cases $\numbercolors<6$~\cite{robertson1993hadwiger}, but 
remains open for each value of $\numbercolors\geq 6$. By Theorem \ref{theorem:Estimates},
$\texttt{Colorable}(c)$ has DP-cores of deterministic state complexity $2^{2^{O(k\log c)}}$, while 
$\texttt{Minor}(K_{c+1})$ has DP-cores of deterministic state complexity $2^{2^{O(k\log k + c^2)}}$. Therefore, by 
using Theorem \ref{theorem:InclusionTestCombinators}, we have that the case $c$ of Hadwiger's conjecture can be 
tested in time $f(\numbercolors,\treewidthvalue) = 2^{2^{O(k\log k + c^2)}}$ on graphs of treewidth at most $\treewidthvalue$. 

Using the fact for each fixed $\numbercolors\in \N$, both the existence of 
$K_{\numbercolors+1}$-minors and the existence of $\numbercolors$-colorings 
are  MSO-definable, together with the fact that the MSO theory of graphs 
of bounded treewidth is definable one can estimate 
$f(\numbercolors,\treewidthvalue)$ by writing explicitly MSO sentences
and then by bounding the running time of the decision algorithm. 
This estimate is however very large (a tower of exponentials of height 10 in $k+c$ 
suffices).
In \cite{kawarabayashi2009hadwiger} 
Karawabayshi have estimated that $f(\numbercolors,\treewidthvalue)\leq p^{p^{p^{p}}}$, where $p=(\treewidthvalue+1)^{(\numbercolors-1)}$. 
It is worth noting that our estimate of  $2^{2^{O(k\log k + c^2)}}$ obtained by a combination Theorem \ref{theorem:InclusionTestCombinators} and Theorem \ref{theorem:Estimates} improves 
significantly on both the estimate obtained using the MSO approach and the estimate provided in \cite{kawarabayashi2009hadwiger}. 
\\ 

\noindent{\bf Tutte's Flow Conjectures.}
Tutte's $5$-flow, $4$-flow, and $3$-flow conjectures are some of the most 
well studient and important open problems in graph theory. The $5$-flow conjecture
states that every bridgeleass graph $G$ has a $\Z_{5}$-flow. This conjecture is
true if and only if every $2$-edge-connected graph
has a $\Z_{5}$-flow \cite{tutte1954contribution}. By Theorem \ref{theorem:Estimates}, both $\texttt{ECon}(2)$ and 
$\texttt{NZFlow}(\Z_5)$ have coherent DP-cores of deterministic state complexity $2^{2^{O(k\log k)}}$. Since 
Tutte's $5$-flow conjecture can be expressed in terms of a Boolen combination of these properties,
we have that this conjecture can be tested in time $2^{2^{O(k\log k)}}$ on graphs of treewidth at most $\treewidthvalue$. 
The $4$-flow conjecture states that every bridgeless graph with no Petersen minor 
has a nowhere-zero $4$-flow \cite{wang2009nowhere}. Since this conjecture can be formulated using a Boolean combination of 
the properties $ECon(2)$, $\texttt{Minor}(P)$ (where $P$ is the 
Pettersen graph), and $\texttt{NZFlow}(\Z_4)$, we have that this conjecture 
can be tested in time $2^{2^{O(k\log k)}}$ on graphs of treewidth at most $\treewidthvalue$. 
Finally, Tutte's $3$-Flow conjecture states that 
every $4$-edge connected graph has a nowhere-zero $3$-flow \cite{fan1993tutte}. Similarly to the 
other cases it can be expressed as a Boolean combination of $\texttt{ECon}(4)$ 
and $\texttt{NZFlow}(\Z_3)$. Therefore, it can be tested in time $2^{2^{O(k\log k)}}$
on graphs of treewidth at most $\treewidthvalue$. 
\\ 

\noindent{\bf Barnette's Conjecture.} This conjecture states
that every $3$-connected, $3$-regular, bipartite, planar graph is 
Hamiltonian. Since a graph is bipartite if and only if it is $2$-colorable, 
and since a graph is planar if and only if it does not 
contain $K_5$ or $K_{3,3}$ as minors, Barnette's conjecture
can be stated as a combination of the cores $\texttt{VCon(3)}$, 
$\texttt{MaxDeg}_{\geq}(3)$, $\texttt{MinDeg}_{\leq}(3)$, $\texttt{Colorable}(2)$, 
$\texttt{Minor}(K_5)$ and $\texttt{Minor}(K_{3,3})$. Therefore, by Theorem \ref{theorem:InclusionTestCombinators}, 
it can be tested in time $2^{2^{O(\treewidthvalue\log \treewidthvalue)}}$ on graphs of
treewidth at most $\treewidthvalue$. 

\section{Conclusion and Future Directions}
\label{section:Conclusion}

In this work, we have introduced a general and modular framework that allows one to 
combine width-based dynamic programming algorithms for model-checking graph-theoretic properties 
into algorithms that can be used to provide a width-based attack to long-standing conjectures
in graph theory. By generality, we mean that our framework can be applied with 
respect to any treelike width measure (Definition \ref{definition:WidthMeasure}), 
including treewidth \cite{bodlaender1997treewidth}, cliquewidth \cite{courcelle2000upper}, 
and many others \cite{korach1993tree,thilikos2000constructive,chung1989graphs,thilikos2005cutwidth,andrade2019width}.
By modularity, we mean that dynamic programming cores may be developed completely independently of each other as 
if they were plugins, and then combined either with the purpose of model-checking more 
complicated graph-theoretic properties, or with the purpose of attacking a given graph theoretic conjecture. 

As a concrete example, we have shown that the validity of several longstanding graph theoretic conjectures 
can be tested on graphs of treewidth at most $\treewidthvalue$ in time double exponential in ${\treewidthvalue}^{O(1)}$.
This upper bound follow from Theorem \ref{theorem:InclusionTestCombinators} together with upper bounds established
on the bitlength/multiplicity of DP-cores deciding several well studied graph properties. 
Although still high, this upper bound improves significantly on approaches based on quantifier elimination.
This is an indication that the expertise accumulated by parameterized complexity theorists 
in the development of more efficient width-based DP algorithms for model checking graph-theoretic properties 
have also relevance in the context of automated theorem proving. It is worth noting that this is the case even 
for graph properties that are computationally easy, such as connectivity and bounded degree.  

For simplicity, have defined the notion of a treelike width measure (Definition \ref{definition:WidthMeasure}) with 
respect to graphs. Nevertheless, this notion can be directly lifted to more general classes of relational structures. 
More specifically, given a class $\mathfrak{R}$ of relational structures over some signature $\mathfrak{s}$, 
we define the notion of a {\em treelike $\mathfrak{R}$-decomposition class} as a sequence 
of triples $\decompositionclass = \{(\alphabet_{\treewidthvalue},\languageclass_{\treewidthvalue},\graphfunction_{\treewidthvalue})\}_{\treewidthvalue\in \N}$
precisely as in Definition \ref{definition:TreelikeDecompositionClass}, with the only exception that now,
$\graphfunction_{\treewidthvalue}$ is a function from $\languageclass_{\treewidthvalue}$ to 
$\mathfrak{R}$. With this adaptation, and by letting the relation $\simeq$ in Definition \ref{definition:Coherency}
denote isomorphism between $\mathfrak{s}$-structures, all results in Section \ref{section:DPCores} generalize
smoothly to relational structures from $\mathfrak{R}$. This generalization is relevant because it shows that 
our notion of width-based automated theorem proving can be extended to a much larger context than graph theory.

In the field of parameterized complexity theory, the irrelevant vertex technique is a set of theoretical tools
\cite{DBLP:journals/jct/AdlerKKLST17,robertson1986graph} that can be used to show that for certain graph
properties $\graphproperty$ there is a constant $K_{\graphproperty}$, such that if $\agraph$ is a graph of treewidth at least $K_{\graphproperty}$ then it contains an {\em irrelevant vertex} for $\graphproperty$.
More specifically, there is a vertex $x$ such that $\agraph$ belongs to $\graphproperty$ if and only if the graph $\agraph\backslash x$ obtained
by deleting $x$ from $\agraph$ belongs to $\graphproperty$. This tool, that builds on Robertson and Seymour's celebrated 
excluded grid theorem \cite{robertson1986graph,chuzhoy2015excluded} and on the flat wall theorem \cite{DBLP:conf/soda/Chuzhoy15,kawarabayashi2018new,sau2021more}, 
has found several applications 
in structural graph theory and in the development of parameterized algorithms 
\cite{AdlerGroheKreutzer2008,robertson1986graph,DBLP:conf/lics/DawarGK07,DBLP:conf/focs/KawarabayashiMR08,DBLP:journals/combinatorica/KawarabayashiK10}. 
Interestingly, the irrelevant vertex technique has also theoretical relevance in the framework of width-based automated theorem proving. 
More specifically, the existence of vertices that are irrelevant for $\graphproperty$ on graphs of treewidth at least $K_{\graphproperty}$ implies
that if there is some graph $\agraph$ that {\em does not} belong to $\graphproperty$, then there is some graph of treewidth at most $K_{\graphproperty}$ that {\em also does not} belong to 
$\graphproperty$. As a consequence, $\graphproperty$ is equal to the class of {\em all graphs} (see Section \ref{section:Preliminaries} for a precise definition of this class) 
if and only if all graphs of treewidth at most $K_{\graphproperty}$ belong to $\graphproperty$. In other words, the irrelevant vertex technique allows one to show that 
certain conjectures are true in the class of all graphs if and only if they are true in the class of graphs of treewidth at most $K$ for 
some {\em constant} $K$. This approach has been considered for instance in the study of Hadwiger's conjecture (for each fixed number of colors $c$)
\cite{kawarabayashi2009hadwiger,seymour2016hadwiger,DBLP:journals/gc/KawarabayashiM07}. Identifying further conjectures that can be studied under the 
framework of the irrelevant vertex technique would be very relevant to the framework of width-based automated theorem proving.

\section*{Acknowledgements}
We acknowledge support from the Research Council of Norway in the context of the project {\em Automated Theorem Proving from the Mindset of Parameterized Complexity Theory} (proj. no. 288761).

\bibliographystyle{plainurl} 
\bibliography{widthBasedATP}

\begin{thebibliography}{10}

\bibitem{adler2011faster}
Isolde Adler, Frederic Dorn, Fedor~V Fomin, Ignasi Sau, and Dimitrios~M Thilikos.
\newblock Faster parameterized algorithms for minor containment.
\newblock {\em Theoretical Computer Science}, 412(50):7018--7028, 2011.

\bibitem{AdlerGroheKreutzer2008}
Isolde Adler, Martin Grohe, and Stephan Kreutzer.
\newblock Computing excluded minors.
\newblock In {\em Proc. of the 29th Annual ACM-SIAM Symposium on Discrete Algorithms {(SODA 2008)}}, pages 641--650. {SIAM}, 2008.

\bibitem{DBLP:journals/jct/AdlerKKLST17}
Isolde Adler, Stavros~G. Kolliopoulos, Philipp~Klaus Krause, Daniel Lokshtanov, Saket Saurabh, and Dimitrios~M. Thilikos.
\newblock Irrelevant vertices for the planar disjoint paths problem.
\newblock {\em J. Comb. Theory, Ser. {B}}, 122:815--843, 2017.

\bibitem{andrade2019width}
Alexsander Andrade~de Melo and Mateus~de Oliveira~Oliveira.
\newblock On the width of regular classes of finite structures.
\newblock In {\em Proc. of the 27th International Conference on Automated Deduction (CADE 2019)}, pages 18--34. Springer, 2019.

\bibitem{bannach2019positive}
Max Bannach and Sebastian Berndt.
\newblock Positive-instance driven dynamic programming for graph searching.
\newblock In {\em Workshop on Algorithms and Data Structures}, pages 43--56. Springer, 2019.

\bibitem{DBLP:journals/algorithms/BannachB19}
Max Bannach and Sebastian Berndt.
\newblock Practical access to dynamic programming on tree decompositions.
\newblock {\em Algorithms}, 12(8):172, 2019.

\bibitem{baste2022diversity}
Julien Baste, Michael~R. Fellows, Lars Jaffke, Tom{\'a}{\v{s}} Masa{\v{r}}{\'\i}k, Mateus de~Oliveira~Oliveira, Geevarghese Philip, and Frances~A. Rosamond.
\newblock Diversity of solutions: An exploration through the lens of fixed-parameter tractability theory.
\newblock {\em Artificial Intelligence}, 303:103644, 2022.

\bibitem{DBLP:conf/soda/BasteST20}
Julien Baste, Ignasi Sau, and Dimitrios~M. Thilikos.
\newblock A complexity dichotomy for hitting connected minors on bounded treewidth graphs: the chair and the banner draw the boundary.
\newblock In {\em {SODA}}, pages 951--970. {SIAM}, 2020.

\bibitem{BauderonC87}
Michel Bauderon and Bruno Courcelle.
\newblock Graph expressions and graph rewritings.
\newblock {\em Math. Syst. Theory}, 20(2-3):83--127, 1987.
\newblock \href {https://doi.org/10.1007/BF01692060} {\path{doi:10.1007/BF01692060}}.

\bibitem{bertele1973non}
Umberto Bertele and Francesco Brioschi.
\newblock On non-serial dynamic programming.
\newblock {\em J. Comb. Theory, Ser. A}, 14(2):137--148, 1973.

\bibitem{biedl2015triangulating}
Therese~C. Biedl.
\newblock On triangulating k-outerplanar graphs.
\newblock {\em Discret. Appl. Math.}, 181(1):275--279, 2015.

\bibitem{BlumensathC06}
Achim Blumensath and Bruno Courcelle.
\newblock Recognizability, hypergraph operations, and logical types.
\newblock {\em Inf. Comput.}, 204(6):853--919, 2006.
\newblock \href {https://doi.org/10.1016/j.ic.2005.11.006} {\path{doi:10.1016/j.ic.2005.11.006}}.

\bibitem{bodlaender1997treewidth}
Hans~L Bodlaender.
\newblock Treewidth: Algorithmic techniques and results.
\newblock In {\em Proc. of the 22nd International Symposium on Mathematical Foundations of Computer Science}, pages 19--36. Springer, 1997.

\bibitem{bodlaender1998partial}
Hans~L Bodlaender.
\newblock A partial k-arboretum of graphs with bounded treewidth.
\newblock {\em Theoretical computer science}, 209(1-2):1--45, 1998.

\bibitem{bodlaender2015deterministic}
Hans~L Bodlaender, Marek Cygan, Stefan Kratsch, and Jesper Nederlof.
\newblock Deterministic single exponential time algorithms for connectivity problems parameterized by treewidth.
\newblock {\em Information and Computation}, 243:86--111, 2015.

\bibitem{bodlaender2008combinatorial}
Hans~L Bodlaender and Arie~MCA Koster.
\newblock Combinatorial optimization on graphs of bounded treewidth.
\newblock {\em The Computer Journal}, 51(3):255--269, 2008.

\bibitem{bodlaender1986classes}
H.L. Bodlaender.
\newblock Classes of graphs with bounded tree-width.
\newblock Technical Report RUU-CS-86-22, Department of Information and Computing Sciences, Utrecht University, 1986.

\bibitem{PilipczukBojanczyk2016}
Miko{\l}aj Boja{\'{n}}czyk and Michal Pilipczuk.
\newblock Definability equals recognizability for graphs of bounded treewidth.
\newblock In {\em Proc. of the 31st Annual ACM/IEEE Symposium on Logic in Computer Science {(LICS 2016)}}, pages 407--416. {ACM}, 2016.

\bibitem{bonomo2016graph}
Flavia Bonomo, Luciano~N Grippo, Martin Milani{\v{c}}, and Mart{\'\i}n~D Safe.
\newblock Graph classes with and without powers of bounded clique-width.
\newblock {\em Discrete Applied Mathematics}, 199:3--15, 2016.

\bibitem{brandstadt1999graph}
Andreas Brandst{\"a}dt, Van~Bang Le, and Jeremy~P Spinrad.
\newblock {\em Graph classes: a survey}.
\newblock SIAM, 1999.

\bibitem{calabro2009complexity}
Chris Calabro, Russell Impagliazzo, and Ramamohan Paturi.
\newblock The complexity of satisfiability of small depth circuits.
\newblock In {\em Proc. of the 4th International Workshop on Parameterized and Exact Computation}, pages 75--85. Springer, 2009.

\bibitem{DBLP:conf/innovations/CarmosinoGIMPS16}
Marco~L. Carmosino, Jiawei Gao, Russell Impagliazzo, Ivan Mihajlin, Ramamohan Paturi, and Stefan Schneider.
\newblock Nondeterministic extensions of the strong exponential time hypothesis and consequences for non-reducibility.
\newblock In {\em Proc. of the 7th {ACM} Conference on Innovations in Theoretical Computer Science (ITCS 2016)}, pages 261--270. {ACM}, 2016.

\bibitem{chung1989graphs}
Fan~RK Chung and Paul~D Seymour.
\newblock Graphs with small bandwidth and cutwidth.
\newblock {\em Discrete Mathematics}, 75(1-3):113--119, 1989.

\bibitem{chuzhoy2015excluded}
Julia Chuzhoy.
\newblock Excluded grid theorem: Improved and simplified.
\newblock In {\em Proc. of the 47th annual ACM symposium on Theory of Computing}, pages 645--654, 2015.

\bibitem{DBLP:conf/soda/Chuzhoy15}
Julia Chuzhoy.
\newblock Improved bounds for the flat wall theorem.
\newblock In {\em Proc. of the 2015 {ACM-SIAM} Symposium on Discrete Algorithms}, pages 256--275. {SIAM}, 2015.

\bibitem{TATA2008}
Hubert Comon, Max Dauchet, R{\'e}mi Gilleron, Florent Jacquemard, Denis Lugiez, Christof L{\"o}ding, Sophie Tison, and Marc Tommasi.
\newblock Tree automata techniques and applications.
\newblock {\em HAL Inria}, (hal-03367725):1--262, 2008.

\bibitem{Courcelle1990MSO}
Bruno Courcelle.
\newblock The monadic second-order logic of graphs. {I}. {R}ecognizable sets of finite graphs.
\newblock {\em Information and Computation}, 85(1):12 -- 75, 1990.

\bibitem{DBLP:journals/tcs/CourcelleD16}
Bruno Courcelle and Ir{\`{e}}ne Durand.
\newblock Computations by fly-automata beyond monadic second-order logic.
\newblock {\em Theor. Comput. Sci.}, 619:32--67, 2016.

\bibitem{CourcelleEngelfriet2012}
Bruno Courcelle and Joost Engelfriet.
\newblock {\em Graph structure and monadic second-order logic: {A} language-theoretic approach}, volume 138.
\newblock Cambridge University Press, 2012.

\bibitem{courcelle2000upper}
Bruno Courcelle and Stephan Olariu.
\newblock Upper bounds to the clique width of graphs.
\newblock {\em Discrete Applied Mathematics}, 101(1-3):77--114, 2000.

\bibitem{cygan2015parameterized}
Marek Cygan, Fedor~V Fomin, {\L}ukasz Kowalik, Daniel Lokshtanov, D{\'a}niel Marx, Marcin Pilipczuk, Micha{\l} Pilipczuk, and Saket Saurabh.
\newblock {\em Parameterized algorithms}, volume~5.
\newblock Springer, 2015.

\bibitem{DBLP:conf/lics/DawarGK07}
Anuj Dawar, Martin Grohe, and Stephan Kreutzer.
\newblock Locally excluding a minor.
\newblock In {\em Proc. of the 22nd ACM/IEEE Anual Symposium on Logic in Computer Science ({LICS} 2007)}, pages 270--279. {IEEE} Computer Society, 2007.

\bibitem{de2016algorithmic}
Mateus de~Oliveira~Oliveira.
\newblock An algorithmic metatheorem for directed treewidth.
\newblock {\em Discrete Applied Mathematics}, 204:49--76, 2016.

\bibitem{downey2012parameterized}
Rodney~G Downey and Michael~R. Fellows.
\newblock {\em Parameterized complexity}.
\newblock Springer Science \& Business Media, 2012.

\bibitem{Elberfeld2016}
Michael Elberfeld.
\newblock Context-free graph properties via definable decompositions.
\newblock In {\em Proc. of the 25th Conference on Computer Science Logic ({CSL} 2016)}, volume~62 of {\em LIPIcs}, pages 17:1--17:16, 2016.

\bibitem{fan1993tutte}
GH~Fan.
\newblock Tutte's 3-flow conjecture and short cycle covers.
\newblock {\em Journal of Combinatorial Theory, Series B}, 57(1):36--43, 1993.

\bibitem{Flum2002query}
J{\"o}rg Flum, Markus Frick, and Martin Grohe.
\newblock Query evaluation via tree-decompositions.
\newblock {\em Journal of the ACM (JACM)}, 49(6):716--752, 2002.

\bibitem{glikson2003nce}
Alexander Glikson and Johann~A Makowsky.
\newblock Nce graph grammars and clique-width.
\newblock In {\em International Workshop on Graph-Theoretic Concepts in Computer Science}, pages 237--248. Springer, 2003.

\bibitem{gnesi1981dynamic}
Stefania Gnesi, Ugo Montanari, and Alberto Martelli.
\newblock Dynamic programming as graph searching: An algebraic approach.
\newblock {\em Journal of the ACM (JACM)}, 28(4):737--751, 1981.

\bibitem{hadwiger1943klassifikation}
Hugo Hadwiger.
\newblock {\"U}ber eine klassifikation der streckenkomplexe.
\newblock {\em Vierteljschr. Naturforsch. Ges. Z{\"u}rich}, 88(2):133--142, 1943.

\bibitem{halin1976s}
Rudolf Halin.
\newblock S-functions for graphs.
\newblock {\em Journal of geometry}, 8(1-2):171--186, 1976.

\bibitem{hicks2004branch}
Illya~V Hicks.
\newblock Branch decompositions and minor containment.
\newblock {\em Networks: An International Journal}, 43(1):1--9, 2004.

\bibitem{DBLP:journals/jcss/ImpagliazzoP01}
Russell Impagliazzo and Ramamohan Paturi.
\newblock On the complexity of k-{SAT}.
\newblock {\em J. Comput. Syst. Sci.}, 62(2):367--375, 2001.

\bibitem{DBLP:journals/jcss/ImpagliazzoPZ01}
Russell Impagliazzo, Ramamohan Paturi, and Francis Zane.
\newblock Which problems have strongly exponential complexity?
\newblock {\em J. Comput. Syst. Sci.}, 63(4):512--530, 2001.

\bibitem{kammer2007determining}
Frank Kammer.
\newblock Determining the smallest k such that g is k-outerplanar.
\newblock In {\em European Symposium on Algorithms}, pages 359--370. Springer, 2007.

\bibitem{DBLP:journals/combinatorica/KawarabayashiK10}
Ken{-}ichi Kawarabayashi and Yusuke Kobayashi.
\newblock Algorithms for finding an induced cycle in planar graphs.
\newblock {\em Comb.}, 30(6):715--734, 2010.

\bibitem{DBLP:journals/gc/KawarabayashiM07}
Ken{-}ichi Kawarabayashi and Bojan Mohar.
\newblock Some recent progress and applications in graph minor theory.
\newblock {\em Graphs Comb.}, 23(1):1--46, 2007.

\bibitem{DBLP:conf/focs/KawarabayashiMR08}
Ken{-}ichi Kawarabayashi, Bojan Mohar, and Bruce~A. Reed.
\newblock A simpler linear time algorithm for embedding graphs into an arbitrary surface and the genus of graphs of bounded tree-width.
\newblock In {\em {FOCS}}, pages 771--780. {IEEE} Computer Society, 2008.

\bibitem{kawarabayashi2009hadwiger}
Ken-ichi Kawarabayashi and Bruce Reed.
\newblock Hadwiger's conjecture is decidable.
\newblock In {\em Proc. of the 41st Annual ACM Symposium on Theory of Computing}, pages 445--454, 2009.

\bibitem{kawarabayashi2018new}
Ken-ichi Kawarabayashi, Robin Thomas, and Paul Wollan.
\newblock A new proof of the flat wall theorem.
\newblock {\em Journal of Combinatorial Theory, Series B}, 129:204--238, 2018.

\bibitem{korach1993tree}
Ephraim Korach and Nir Solel.
\newblock Tree-width, path-width, and cutwidth.
\newblock {\em Discrete Applied Mathematics}, 43(1):97--101, 1993.

\bibitem{kumar1988cdp}
Vipin Kumar and Laveen~N Kanal.
\newblock The cdp: A unifying formulation for heuristic search, dynamic programming, and branch-and-bound.
\newblock In {\em Search in Artificial Intelligence}, pages 1--27. Springer, 1988.

\bibitem{leon2009tools}
C~Le{\'o}n, G~Miranda, and C~Rodr{\'\i}guez.
\newblock Tools for tree searches: Dynamic programming.
\newblock {\em Optimization Techniques for Solving Complex Problems}, pages 209--230, 2009.

\bibitem{DBLP:journals/eatcs/LokshtanovMS11}
Daniel Lokshtanov, D{\'{a}}niel Marx, and Saket Saurabh.
\newblock Lower bounds based on the exponential time hypothesis.
\newblock {\em Bull. {EATCS}}, 105:41--72, 2011.

\bibitem{DBLP:journals/siamdm/LokshtanovMSZ19}
Daniel Lokshtanov, Amer~E. Mouawad, Saket Saurabh, and Meirav Zehavi.
\newblock Packing cycles faster than erdos-posa.
\newblock {\em {SIAM} J. Discret. Math.}, 33(3):1194--1215, 2019.

\bibitem{morales2000parallel}
D~Gonzalez Morales, Francisco Almeida, C~Rodr{\i}guez, Jose~L. Roda, I~Coloma, and A~Delgado.
\newblock Parallel dynamic programming and automata theory.
\newblock {\em Parallel computing}, 26(1):113--134, 2000.

\bibitem{DBLP:conf/stacs/NederlofPSW22}
Jesper Nederlof, Michal Pilipczuk, C{\'{e}}line M.~F. Swennenhuis, and Karol Wegrzycki.
\newblock Isolation schemes for problems on decomposable graphs.
\newblock In {\em {STACS}}, volume 219 of {\em LIPIcs}, pages 50:1--50:20. Schloss Dagstuhl - Leibniz-Zentrum f{\"{u}}r Informatik, 2022.

\bibitem{oliveira2013subgraphs}
Mateus~de Oliveira~Oliveira.
\newblock Subgraphs satisfying {MSO} properties on z-topologically orderable digraphs.
\newblock In {\em International Symposium on Parameterized and Exact Computation}, pages 123--136. Springer, 2013.

\bibitem{papusha2016automata}
Ivan Papusha, Jie Fu, Ufuk Topcu, and Richard~M Murray.
\newblock Automata theory meets approximate dynamic programming: Optimal control with temporal logic constraints.
\newblock In {\em Proc. of the 55th IEEE Conference on Decision and Control (CDC 2016)}, pages 434--440. IEEE, 2016.

\bibitem{pardo1997automata}
Miguel Angel~Alonso Pardo, Eric De~La~Clergerie, and Manuel~Vilares Ferro.
\newblock Automata-based parsing in dynamic programming for linear indexed grammars.
\newblock {\em Proc. of DIALOGUE}, 97:22--27, 1997.

\bibitem{parker1989partial}
D~Stott Parker.
\newblock Partial order programming.
\newblock In {\em Proceedings of the 16th ACM SIGPLAN-SIGACT symposium on Principles of programming languages}, pages 260--266, 1989.

\bibitem{pilipczuk2011problems}
Micha{\l} Pilipczuk.
\newblock Problems parameterized by treewidth tractable in single exponential time: A logical approach.
\newblock In {\em Proc. of the 36th International Symposium on Mathematical Foundations of Computer Science}, pages 520--531. Springer, 2011.

\bibitem{raymond2017recent}
Jean-Florent Raymond and Dimitrios~M Thilikos.
\newblock Recent techniques and results on the erd{\H{o}}s--p{\'o}sa property.
\newblock {\em Discrete Applied Mathematics}, 231:25--43, 2017.

\bibitem{robertson1993hadwiger}
Neil Robertson, Paul Seymour, and Robin Thomas.
\newblock Hadwiger's conjecture fork 6-free graphs.
\newblock {\em Combinatorica}, 13(3):279--361, 1993.

\bibitem{robertson1984graph}
Neil Robertson and Paul~D Seymour.
\newblock Graph minors. {III}. {P}lanar tree-width.
\newblock {\em Journal of Combinatorial Theory, Series B}, 36(1):49--64, 1984.

\bibitem{robertson1986graph}
Neil Robertson and Paul~D Seymour.
\newblock Graph minors. {V}. {E}xcluding a planar graph.
\newblock {\em Journal of Combinatorial Theory, Series B}, 41(1):92--114, 1986.

\bibitem{sau2021more}
Ignasi Sau, Giannos Stamoulis, and Dimitrios~M Thilikos.
\newblock A more accurate view of the flat wall theorem.
\newblock {\em arXiv preprint arXiv:2102.06463}, 2021.

\bibitem{seese1991structure}
Detlef Seese.
\newblock The structure of the models of decidable monadic theories of graphs.
\newblock {\em Annals of pure and applied logic}, 53(2):169--195, 1991.

\bibitem{seymour2016hadwiger}
Paul Seymour.
\newblock {\em Hadwiger’s conjecture}.
\newblock Springer, 2016.

\bibitem{thilikos2005cutwidth}
Dimitrios~M Thilikos, Maria Serna, and Hans~L Bodlaender.
\newblock Cutwidth {I}: A linear time fixed parameter algorithm.
\newblock {\em Journal of Algorithms}, 56(1):1--24, 2005.

\bibitem{thilikos2000constructive}
Dimitrios~M Thilikos, Maria~J Serna, and Hans~L Bodlaender.
\newblock Constructive linear time algorithms for small cutwidth and carving-width.
\newblock In {\em Proc. of the 11th International Symposium on Algorithms and Computation}, pages 192--203. Springer, 2000.

\bibitem{tutte1954contribution}
William~Thomas Tutte.
\newblock A contribution to the theory of chromatic polynomials.
\newblock {\em Canadian journal of mathematics}, 6:80--91, 1954.

\bibitem{tutte1969recent}
William~Thomas Tutte.
\newblock Recent progress in combinatorics.
\newblock In {\em Proceedings of the 3rd Waterloo Conference on Combinatorics}. Academic Press, 1969.

\bibitem{DBLP:conf/esa/RooijBR09}
Johan M.~M. van Rooij, Hans~L. Bodlaender, and Peter Rossmanith.
\newblock Dynamic programming on tree decompositions using generalised fast subset convolution.
\newblock In {\em {ESA}}, volume 5757 of {\em Lecture Notes in Computer Science}, pages 566--577. Springer, 2009.

\bibitem{wang2009nowhere}
Xiaofeng Wang, Cun-Quan Zhang, and Taoye Zhang.
\newblock Nowhere-zero 4-flow in almost petersen-minor free graphs.
\newblock {\em Discrete mathematics}, 309(5):1025--1032, 2009.

\bibitem{ziobro2019finding}
Micha{\l} Ziobro and Marcin Pilipczuk.
\newblock Finding hamiltonian cycle in graphs of bounded treewidth: Experimental evaluation.
\newblock {\em Journal of Experimental Algorithmics (JEA)}, 24:1--18, 2019.

\end{thebibliography}

\appendix
\newpage
\section{Proof of Theorem \ref{theorem:TreewidthTreelike}}
\label{appendix:proofTheoremTreewidthTreelike} The proof of Theorem
\ref{theorem:TreewidthTreelike} follows from results available in the
literature. A graph has cliquewidth at most $k$ if and only if it can be
defined as the graph associated with a $k$-expression as introduced in
\cite{courcelle2000upper}.
The fact that the set of all $k$-expressions is
regular follows directly by the definition of $k$-expression \cite{courcelle2000upper,DBLP:journals/tcs/CourcelleD16}. This shows that cliquewidth is a treelike width
measure. In Chapter 12 of \cite{downey2012parameterized} it is shown how to
define graphs of treewidth at most $\treewidthvalue$ using suitable parse
trees, and how to associate a graph of treewidth at most $\treewidthvalue$ to
each such parse tree. The fact that the set of all parse trees corresponding to
graphs of treewidth at most $\treewidthvalue$ is regular is a direct
consequence of the definition. This shows that treewidth is a treelike width
measure. Graphs of pathwidth at most $\treewidthvalue$ can be obtained by
considering a restricted version of the parse-trees considered in
\cite{downey2012parameterized}. This restriction preserves regularity.
Therefore, pathwidth is also a treelike width measure. Cutwidth can be shown to
be treelike using the notion of slice decompositions considered for instance in
\cite{oliveira2013subgraphs}, while the fact that carving-width is treelike
follows from the generalization of slice-decompositions to the context of trees
defined in \cite{de2016algorithmic}.
Finally, the fact that ODD-width is a treelike width measure stems from the fact 
that for each $k$ the set of all ordered decision diagrams (ODDs) of width 
at most $k$ can be defined using a finite automaton whose size depends only on $k$ \cite{andrade2019width}. 
 $\square$

\section{Formal Definition of Instructive Tree Decompositions}
\label{formalDefinitionITD}

For each $k\in \N$, we formally define the set $\allabstractdecompositionstreewidth{\treewidthvalue}$ of $k$-instructive tree decompositions as the
language of a suitable tree automaton $\treeautomaton_{\treewidthvalue}$ 
over the alphabet $\abstractalphabet{\treewidthvalue}$.
More specifically, we let $\allabstractdecompositionstreewidth{\treewidthvalue} = \lang(\treeautomaton_{\treewidthvalue})$ where
$\treeautomaton_k = (\alphabet_k,\statestreeautomaton_k,\finalstatestreeautomaton_k,\transitionstreeautomaton_k)$ is a tree automaton 
with $\statestreeautomaton_k = \finalstatestreeautomaton_k = \powerset{[\treewidthvalue+1]}$, and transitions
$$
\begin{array}{lcl}
\transitionstreeautomaton_k  & = & \{ \leaftype \rightarrow \emptyset\}\; \cup \\ 
& & \{\introvertextype{\vertexone}(\abag) \rightarrow \abag \cup \{\vertexone\} \;|\; \abag\subseteq [\treewidthvalue+1],\; 
\vertexone\in [\treewidthvalue+1]\backslash \abag\}\; \cup \; \\
& & \{\forgetvertextype{\vertexone}(\abag) \rightarrow \abag\backslash \{\vertexone\} \;|\; \abag\subseteq [\treewidthvalue+1],\; 
\vertexone\in \abag\} \; \cup \; \\ 
& & \{\introedgetype{\vertexone}{\vertextwo}(\abag) \rightarrow \abag\;|\; \abag\subseteq [\treewidthvalue+1], \vertexone,\vertextwo\in \abag, \vertexone\neq \vertextwo\} \; \cup \; \\
& & \{\jointype(\abag,\abag)\rightarrow \abag \;|\; \abag\subseteq [\treewidthvalue+1]\}.
\end{array}
$$

Intuitively, states of $\treeautomaton_{\treewidthvalue}$ are subsets of $[\treewidthvalue+1]$
corresponding to subsets of active labels. The set of transitions specify both which instructions
can be applied from a given set of active labels $\abag$, and which labels are active 
after the application of a given instruction. 
 
\begin{definition}
\label{instructiveDecomposition}
The terms in $\allabstractdecompositionstreewidth{\treewidthvalue}$ are called $\treewidthvalue$-instructive tree decompositions. Terms in $\allabstractdecompositionstreewidth{\treewidthvalue}$
that do not use the symbol $\jointype$ are called $\treewidthvalue$-instructive path decompositions. We let 
$\allabstractdecompositionspathwidth{\treewidthvalue}$ denote the set of all $\treewidthvalue$-instructive path decompositions. 
\end{definition}

Given numbers $n_1,n_2\in \N$ and a subset $Y\subseteq [n_2]$,
we let 
$\relabelingfunction[n_1,n_2,Y]:[n_2]\backslash Y\rightarrow \{n_1+1,\dots,n_1+n_2-|Y|\}$ be the unique injective function from 
$[n_2]\backslash Y$ to 
$\{n_1+1,\dots,n_1+n_2-|Y|\}$ with the property that for each $i,j\in [n_2]\backslash Y$, $i<j$ implies that $\relabelingfunction(i)<\relabelingfunction(j)$. Intuitively, for each $i\in [n_2]\backslash Y$, $\relabelingfunction[n_1,n_2,Y](i)$ is equal to $n_1$ plus the 
order of $i$ in the set $[n_2]\backslash Y$. 

Let $\treewidthvalue\in \N$. A $\treewidthvalue$-boundaried graph is a pair $(\agraph,\topmapname)$ where $\agraph$ is a graph and
$\topmapname:\bagset\rightarrow \vertexset{\agraph}$ is an injective map from some subset $\bagset\subseteq [\treewidthvalue+1]$ to
the vertex set of $\agraph$. We assume that $\vertexset{\agraph} = [n]$ for some $n\in \N$, and that $\edgeset{\agraph} = [m]$ for some $m\in\N$. 
Given $\treewidthvalue$-boundaried graphs $(\agraph_1,\topmapname_1)$ and $(\agraph_2,\topmapname_2)$
with $\domain(\topmapname_1)=\domain(\topmapname_2)$, the {\em join} of $(\agraph_1,\topmapname_1)$ and $(\agraph_2,\topmapname_2)$, 
denoted by $(\agraph_1,\topmapname_1)\oplus (\agraph_2,\topmapname_2)$, is the $\treewidthvalue$-boundaried 
graph $(\agraph,\topmapname)$ obtained by identifying, for each $\vertexone\in \bagset$, 
the vertex $\topmapname_1(\vertexone)$ of $\agraph_1$ with the vertex $\topmapname_2(\vertexone)$ of $\agraph_2$. 
More precisely, let $\vertexset{\agraph_1}=[n_1]$,
$\edgeset{\agraph_1} = [m_1]$, $\vertexset{\agraph_2}=[n_2]$,  
and $\edgeset{\agraph_2}=[m_2]$ for numbers $n_1,m_1,n_2,m_2\in \N$, and 
let $\relabelingfunction = \relabelingfunction[n_1,n_2,\topmapname_2(\topbagname)]$. Then, the $\treewidthvalue$-boundaried graph $(\agraph,\topmapname)$ is defined as follows.

\begin{enumerate}
\item $\vertexset{\agraph} = [n_1+n_2\,\backslash\,|\topmapname_2(\topbagname)|]$,
\item $\edgeset{\agraph} = [m_1+m_2]$
\item $\incidencerelation{\agraph} =  
\incidencerelation{\agraph_1} \cup
\{(e+m_1,\topmapname_1(\vertexone)) : (e,\topmapname_{2}(\vertexone))\in \incidencerelationname_2\}$ \\
\hphantom{aaaaaaaaa}$\cup\; \{(e+m_1,\relabelingfunction(x))\;:\; (e,x)\in \incidencerelationname_2, 
x\in \vertexset{\agraph_2}\backslash \image(\topmapname_2)\}$, 
\item For each $\vertexone\in [\treewidthvalue+1], \topmapname(\vertexone)=\topmapname_1(\vertexone)$. 
\end{enumerate}

For each $\treewidthvalue\in \N$, we let
$\mathcal{F}_{\treewidthvalue} = \{f:\abag\rightarrow \N\;|\; \abag\subseteq [\treewidthvalue+1],\; f\mbox{ is injective}\}$ 
be the set of injective functions from some subset $\abag\subseteq [\treewidthvalue+1]$ to $\N$. 
As a last step, we define a function 
$\graphfunction:\bigcup_{\treewidthvalue\in \N} \allabstractdecompositionstreewidth{\treewidthvalue} \rightarrow \allgraphs$ that assigns a graph 
$\graphfunction(\abstractdecomposition)$ to 
each $\abstractdecomposition\in \bigcup_{\treewidthvalue} \allabstractdecompositionstreewidth{\treewidthvalue}$. 
This function is defined inductively below, together with an auxiliary function 
$\topmapname:\bigcup_{\treewidthvalue} \allabstractdecompositionstreewidth{\treewidthvalue}\rightarrow \bigcup_{\treewidthvalue} \mathcal{F}_{\treewidthvalue}$ 
that assigns, for each $\treewidthvalue\in \N$, and each $\treewidthvalue$-instructive tree decomposition $\abstractdecomposition$, 
an injective map $\topmapname[\abstractdecomposition]:\abag\rightarrow \N$ in the set $\mathcal{F}_{\treewidthvalue}$. In this way, the 
pair $(\decompositiongraph{\abstractdecomposition},\topmapname[\abstractdecomposition])$ forms a $\treewidthvalue$-boundaried graph. 
Each element $\vertexone\in \abag$ is said to be an {\em active label} for $\decompositiongraph{\abstractdecomposition}$, and the vertex 
$\topmapname[\abstractdecomposition](\vertexone)$ is the active vertex labeled with $\vertexone$. The functions 
$\graphfunction$ and $\topmapname$ are inductively defined as follows.
Below, for each $\abstractdecomposition \in \bigcup_{\treewidthvalue} \allabstractdecompositionstreewidth{\treewidthvalue}$,
we specify the injective map $\topmapname[\abstractdecomposition]$ as a subset of pairs from $[\treewidthvalue+1]\times \N$.

\begin{enumerate}
\setlength\itemsep{0.5em}
\item If $\abstractdecomposition = \leaftype$, then $\graphfunction(\abstractdecomposition) = (\emptyset,\emptyset,\emptyset)$
 and $\topmapname[\abstractdecomposition] = \emptyset$. 
\item If $\abstractdecomposition = \introvertextype{\vertexone}(\sigmaabstractdecomposition)$ then 
$$\graphfunction(\abstractdecomposition) = (\vertexset{\graphfunction(\sigmaabstractdecomposition)} \cup 
		\{|\vertexset{\graphfunction(\sigmaabstractdecomposition)}| +1\}, \edgeset{\graphfunction(\sigmaabstractdecomposition)},
		\incidencerelation{\graphfunction(\sigmaabstractdecomposition)}),
\mbox{ and } 
\topmapname[\abstractdecomposition] = \topmapname[\sigmaabstractdecomposition] \cup 
\{(\vertexone,|\vertexset{\graphfunction(\sigmaabstractdecomposition)}|+1)\}.$$
\item If $\abstractdecomposition = \forgetvertextype{\vertexone}(\sigmaabstractdecomposition)$, then 
$\graphfunction(\abstractdecomposition) = \graphfunction(\sigmaabstractdecomposition)$, and
$\topmapname[\abstractdecomposition] = \topmapname[\sigmaabstractdecomposition] \backslash \{(\vertexone,\topmapname[\sigmaabstractdecomposition](\vertexone))\}.$

\item If $\abstractdecomposition = \introedgetype{\vertexone}{\vertextwo}(\sigmaabstractdecomposition)$, then
$$
\begin{array}l
\graphfunction(\abstractdecomposition) = (\vertexset{\graphfunction(\sigmaabstractdecomposition)},
\edgeset{\graphfunction(\sigmaabstractdecomposition)}\cup \{|\edgeset{\graphfunction(\sigmaabstractdecomposition)}|+1\}, \incidencerelation{\graphfunction(\sigmaabstractdecomposition)} \cup 
\{(|\edgeset{\graphfunction(\sigmaabstractdecomposition)}|+1,\topmapname[\sigmaabstractdecomposition](\vertexone)),(|\edgeset{\graphfunction(\sigmaabstractdecomposition)}|+1,\topmapname[\sigmaabstractdecomposition](\vertextwo))\}),
\\ 
\mbox{ and } 
\topmapname[\abstractdecomposition] = \topmapname[\sigmaabstractdecomposition].
\end{array}
$$
 
\item If $\abstractdecomposition = \jointype(\sigmaabstractdecomposition_1,\sigmaabstractdecomposition_2)\in \allabstractdecompositionstreewidth{\treewidthvalue}$, then
$(\graphfunction(\abstractdecomposition),\topmapname[\abstractdecomposition]) = (\graphfunction(\sigmaabstractdecomposition_1),\topmapname[\sigmaabstractdecomposition_1]) \oplus (\graphfunction(\sigmaabstractdecomposition_2),\topmapname[\sigmaabstractdecomposition_2])$.
\end{enumerate}

For illustration purposes, in Figure \ref{fig:my_label} we provide a step-by-step construction of the graph associated with the $2$-instructive tree decomposition of Figure \ref{figure:TreeDecomposition}. 

By letting, for each $\treewidthvalue\in \N$, $\graphfunction_{\treewidthvalue}$ be the restriction of $\graphfunction$ to the set 
$\allabstractdecompositionstreewidth{\treewidthvalue}$, 
we have that the sequence $\instructivetreedecompositionclass  = \{(\abstractalphabet{\treewidthvalue},\allabstractdecompositionstreewidth{\treewidthvalue},\graphfunction_{\treewidthvalue})\}_{\treewidthvalue\in \N}$ 
is a treelike decomposition class. We call this class the {\em instructive tree decomposition class}. 
Note that $\instructivetreedecompositionclass$ has complexity $2^{k}$, since as discussed above, for each $\treewidthvalue\in \N$, $\instructivetreedecompositionclass_{\treewidthvalue}$ is 
accepted by a tree automaton  $\treeautomaton_{\treewidthvalue}$ with $2^{\treewidthvalue}$ states. 

\begin{figure}[h]
    \centering
    \includegraphics[scale=0.8]{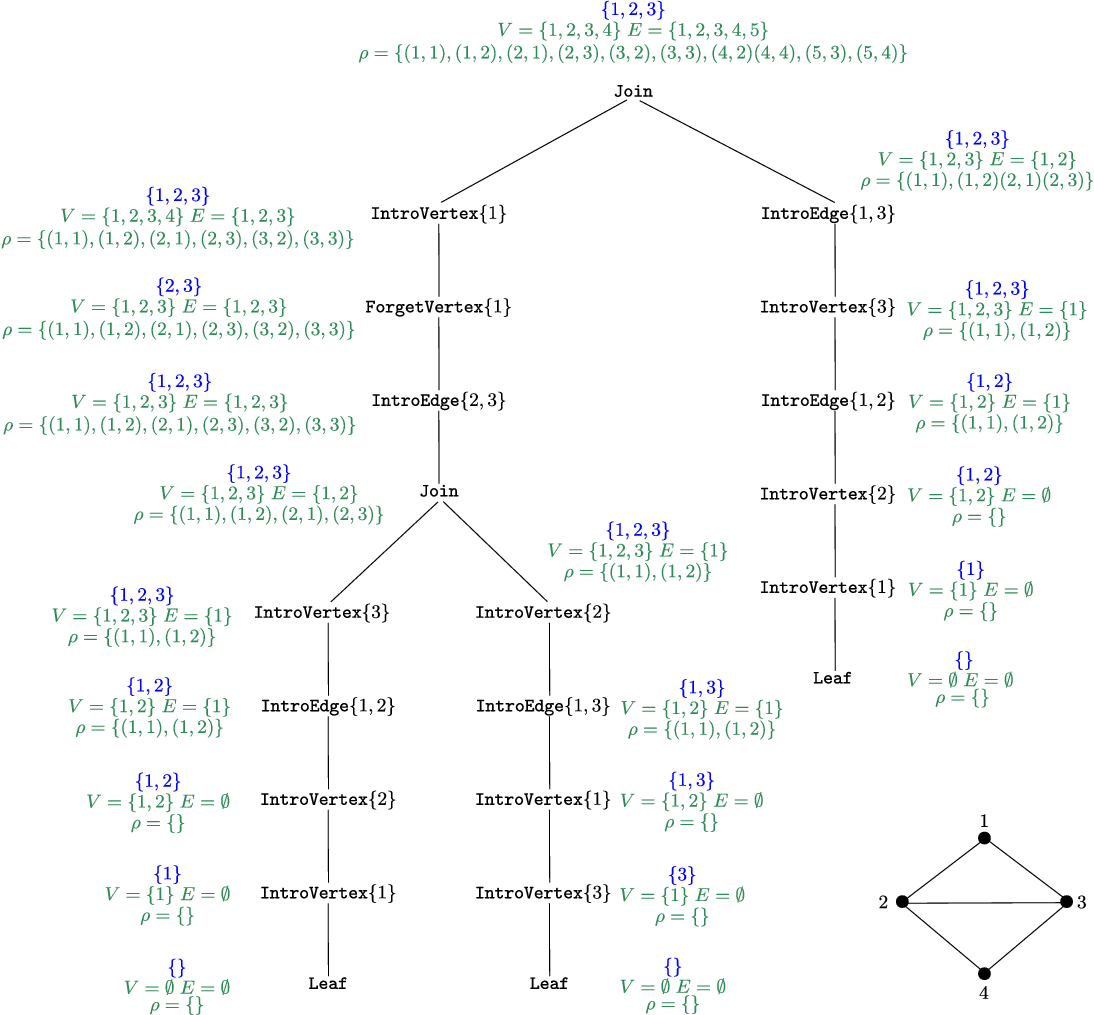}
    \caption{Construction of the graph associated to the $2$-instructive tree decomposition of Fig. \ref{figure:TreeDecomposition}. The set of active labels on each node is specified in blue, while the graph associated to the decomposition rooted at that node is specified in green. The boundary maps are omitted in the figure. Nevertheless, in each node the boundary map assigns each label $i$ to vertex $i$, except for the left child of the topmost node, where the boundary map assigns labels $1$, $2$, and $3$ to vertices $4$, $2$ and $3$ respectively.}
    \label{fig:my_label}
\end{figure}

\section{Proof of Lemma \ref{lemma:TreewidthIsTreelikeConfirmation}}
\label{section:ProofLemmaTreewidthIsTreelikeConfirmation}
\newcommand{\vertexTreeDec}{x}

Next, in Definition \ref{definition:TreeDecomposition}, we provide a slightly different, but equivalent definition of the standard notion of tree decomposition of a graph. The only difference is that we replace the condition that requires that for each edge of the graph there is a bag containing the endpoints of that edge, with a choice function $\beta$, which specifies, for each edge $\anedge$, the bag $X_{\beta(\anedge)}$ that contains the endpoints
of $\anedge$. We also assume that the tree is rooted, since choosing a root can be done without loss of generality. 

\begin{definition}[Tree Decomposition]
\label{definition:TreeDecomposition}
Let $\agraph$ be a graph. A tree decomposition of 
$\agraph$ is a triple $(\atree,X,\beta)$ where $\atree=(N,F,r)$ is a rooted tree with set of nodes $N$, set of arcs $F$, and root $r$;
$\beta:\edgeset{\agraph}\rightarrow N$ is a function mapping each edge $\anedge\in \edgeset{\agraph}$ to some node $\anode\in N$; and $X=\{X_{u}\}_{u\in N}$ is a collection of subsets of $\vertexset{\agraph}$ satisfying the following properties. 
\begin{enumerate}
    \item For each vertex $\vertexTreeDec\in \vertexset{\agraph}$, there is some $\anode\in N$ such that $\vertexTreeDec\in X_{\anode}$. 
    \item For each edge $\anedge\in \edgeset{\agraph}$, $\edgeendpoints{\agraph}(\anedge)\subseteq X_{\beta(\anedge)}$.
    \item For each vertex $\vertexTreeDec\in \vertexset{\agraph}$, the set
    $\{\anode\in N\;:\; \vertexTreeDec\in X_{\anode}\}$ induces a connected subtree of $\atree$. 
\end{enumerate}
\end{definition}

The subsets in $X$ are called {\em bags}. The width of $(\atree,X,\beta)$ is defined as the maximum size of a bag in $X$ minus one: $\mathrm{w}(\atree,X,\beta) = \max_{\anode\in N} |X_{\anode}|-1$. The treewidth of a graph $\agraph$ is the minimum width of a tree decomposition of $\agraph$. 

We say that $(\atree,X,\beta)$ is a
{\em nice edge-introducing tree decomposition of $\agraph$} if the function $\beta$ is injective, and additionally, each node $\anode\in N$ is of one of the following types. 

\begin{enumerate}
    \item $\mathtt{Leaf}$ node: $\anode$ has no child, and $X_{\anode} = \emptyset$. 
    \item $\mathtt{IntroVertex}$ node: $\anode$ has a unique child $\anode'$ and
    $X_{\anode} = X_{\anode'} \cup \{\vertexTreeDec\}$ for some vertex $\vertexTreeDec\in \vertexset{\agraph}$. We say that $\vertexTreeDec$ is {\em introduced} at $\anode$.
    \item $\mathtt{ForgetVertex}$ node: $\anode$ has a unique child $\anode'$ and $X_{\anode'}  = X_{\anode} \cup \{\vertexTreeDec\}$ for some vertex $\vertexTreeDec\in \vertexset{\agraph}$. We say that $\vertexTreeDec$ is {\em forgotten} at $\anode$.
    \item $\mathtt{IntroEdge}$ node: $\anode$ has 
    a unique child $\anode'$, $X_{\anode} = X_{\anode'}$, and $\anode = \beta(\anedge)$ for some $\anedge\in \edgeset{\agraph}$. We say that $\anedge$ is {\em introduced} at $\anode$. \item $\mathtt{Join}$ node: $\anode$ has two children $\anode'$ and $\anode''$, and 
    $X_{\anode} = X_{\anode'} = X_{\anode''}$. 
\end{enumerate}

We note that our notion of nice, edge-introducing tree decompositions of a graph is essentially identical to the notion of nice tree decompositions with introduce-edge nodes used in the literature (see \cite{cygan2015parameterized}, pages 161 and 168), except that we do not require the root bag to be empty.
It can be shown that any tree decomposition of width $\treewidthvalue$ can be efficiently transformed into a nice, edge-introducing tree decomposition of width at most $\treewidthvalue$ (see for instance Lemma 7.4 of \cite{cygan2015parameterized}). 

\begin{observation}
\label{observation:EdgeIntroducingDecompositionConversion}
Let $\agraph$ be a graph, and $(\atree,X,\beta)$ be a tree decomposition $\agraph$ of width at most $\treewidthvalue$, where $T = (N,F,r)$. Then, one can construct in time $O(\treewidthvalue^2\cdot \max\{|\vertexset{\agraph}|,|N|
\} + |\edgeset{\agraph}|)$ a nice, edge-introducing tree decomposition $(\atree',X',\beta')$ of $\agraph$ of width at most 
$\treewidthvalue$ that has  $O(\treewidthvalue\cdot|\vertexset{\agraph}|+|\edgeset{\agraph}|)$ nodes.  
\end{observation}

Now, we are in a position to prove that a graph $\agraph$ is isomorphic to the graph $\decompositiongraph{\abstractdecomposition}$ associated with some $\treewidthvalue$-instructive tree decomposition $\abstractdecomposition$ if and only if $\agraph$ has treewidth at most $\treewidthvalue$.

First, let $\abstractdecomposition$ be a  $\treewidthvalue$-instructive tree decomposition, and  $\decompositiongraph{\abstractdecomposition}$ be its associated graph. Then this graph can be obtained from smaller graphs of size at most $\treewidthvalue+1$ by successively identifying sub-graphs 
of size at most $\treewidthvalue+1$. This implies that $\decompositiongraph{\abstractdecomposition}$ is a subgraph of a $\treewidthvalue$-tree, and therefore
that $\decompositiongraph{\abstractdecomposition}$ has treewidth at most $\treewidthvalue$. In the next claim, we show how to actually extract from $\abstractdecomposition$ a nice edge-introducing tree decomposition of the graph $\decompositiongraph{\abstractdecomposition}$ of width at most $\treewidthvalue$. 

\begin{claim}
\label{claim:ConversionITD-NiceTD}
Let $\abstractdecomposition$ be a $\treewidthvalue$-instructive tree decomposition. Then, the graph $\decompositiongraph{\abstractdecomposition}$ has a nice edge-introducing tree decomposition 
$(\atree,X,\beta)$ of width at most $\treewidthvalue$ where the bag associated with the root node $r$ is the set $X_r = \topmap{\abstractdecomposition}(\topbag{\abstractdecomposition})$. 
\end{claim}
\begin{proof}
The proof is by induction on the height $h$ of $\abstractdecomposition$. In the base case, $h=0$. In this case, $\abstractdecomposition = \leaftype$. Therefore,
since $\decompositiongraph{\abstractdecomposition}$ is the empty graph, the claim is vacuously satisfied. Now, suppose the claim is valid for every 
$\treewidthvalue$-instructive tree decomposition $\abstractdecomposition$ of height at most $h$, and let $\abstractdecomposition$ be a $\treewidthvalue$-instructive
tree decomposition of height $h+1$. There are four cases to be analyzed. 
\begin{itemize}
\item Let  $\sigmaabstractdecomposition$ be a $\treewidthvalue$-instructive
tree decompositions of height at most $h$ and let $(T,X,\beta)$ be a nice edge-introducing tree decompositions with root $r$ and root-bag
$X_{r}$ satisfying the conditions of the claim. In this case we create a new root node $r'$ with root bag $X_{r'}$, and set $r$ as the child of $r'$.
\begin{enumerate}
\item If $\abstractdecomposition = \introvertextype{\vertexone}(\sigmaabstractdecomposition)$, we let $X_{r'} = X_{r} \cup \{\topmap{\abstractdecomposition}(\vertexone)\}$. 
\item If $\abstractdecomposition = \forgetvertextype{\vertexone}(\sigmaabstractdecomposition)$, we let $X_{r'} = X_{r}\backslash \{\topmap{\abstractdecomposition}(\vertexone)\}$. 
\item If $\abstractdecomposition = \introedgetype{\vertexone}{\vertextwo}(\sigmaabstractdecomposition)$, we let $X_{r'} = X_r$. 
\end{enumerate}
\item Let $\sigmaabstractdecomposition_1$ and $\sigmaabstractdecomposition_2$ be  $\treewidthvalue$-instructive
tree decompositions of height at most $h$ and let $(T_1,X_1,\beta_1)$ and $(T_2,X_2,\beta_2)$ be their respective nice edge-introducing tree decompositions with roots $r_1$ and $r_2$ and root-bags
$X_{r_1}$ and $X_{r_2}$ satisfying the conditions of the claim. 
\begin{enumerate}
\item[4.] Let $\abstractdecomposition = \jointype(\sigmaabstractdecomposition_1,\sigmaabstractdecomposition_2)$. In this case, we create a new root bag $r$, labeled with the bag $X_{r} = X_{r_1}$ and set 
$r_1$ and $r_2$ as the children of $r$. Additionally, in the new tree decomposition, for each node $\anode$ of $T_2$, the bag $X_{\anode}$ is replaced by the bag
$$X'_{\anode} = \{\topmap{\sigmaabstractdecomposition_1}(\vertexone)\;:\; \vertexone\in \topbag{\sigmaabstractdecomposition},\; \topmap{\sigmaabstractdecomposition_2}(\vertexone)\in X_{\anode} \}
\cup \{\relabelingfunction(x)\;:\; x\in X_{\anode}\backslash \topmap{\sigmaabstractdecomposition_2}(\topbag{\sigmaabstractdecomposition_2})\},$$
where 
$\relabelingfunction = \relabelingfunction[n_1,n_2,Y]$ is the relabeling function defined in Section \ref{formalDefinitionITD} with 
$n_1=|\vertexset{\agraph_1}|$, $n_2 = |\vertexset{\agraph_2}|$, and 
$Y = \topmap{\sigmaabstractdecomposition_2}(\topbag{\sigmaabstractdecomposition_2})$. 
\end{enumerate} 
\end{itemize}
In any of the four cases, the obtained tree decomposition $(T,X,\beta)$ is a
tree decomposition of the graph $\decompositiongraph{\abstractdecomposition}$ where the bag corresponding to the root node $r$ is the set $X_r = \topmapname[\abstractdecomposition](\topbag{\abstractdecomposition})$. 
\end{proof}

For the converse, let $\agraph$ be a graph and $\alpha:\vertexset{\agraph}\rightarrow [\treewidthvalue+1]$ be a proper $(\treewidthvalue+1)$-coloring of $\agraph$. More specifically, for each edge $\anedge\in \edgeset{\agraph}$ the endpoints of $\anedge$ are colored differently by $\alpha$. We say that $\alpha$ is bag-injective with respect to a  
nice, edge introducing tree decomposition $(\atree,X,\beta)$
of $\agraph$ if for each bag $X_{\anode}\in X$, the restriction $\alpha|_{X_{\anode}}$ is injective.

\begin{observation}
\label{observation:InjectiveColoring}
Let $\agraph$ be a graph and $(\atree,X,\beta)$ be a 
tree decomposition of $\agraph$ of width $\treewidthvalue$, where $T=(N,F,r)$. Then, one can construct in time $O(\treewidthvalue\cdot |N|)$ a proper $(\treewidthvalue+1)$-coloring $\alpha$ of $\agraph$
that  is bag-injective for $(\atree,X,\beta)$. 
\end{observation}
\begin{proof}
We will construct a coloring $\alpha:\vertexset{\agraph}\rightarrow [\treewidthvalue+1]$ of $\agraph$ that is bag-injective for $(\atree,X,\beta)$ by traversing the bags in $X$ from the root towards the leaves. We start by choosing an arbitrary injective coloring of the vertices in the root bag $X_r$. Since
$X_r$ has at most $\treewidthvalue+1$ vertices, such an injective coloring with at most $(\treewidthvalue+1)$ colors exists. 
Now assume that the vertices of a given bag $X_{\anode}$ have been injectively colored with at most $\treewidthvalue+1$ colors, and let $\anode'$ be a child of $\anode$. Then we have
three possibilities: (i) $X_{\anode'} = X_{\anode}$; (ii) $X_{\anode'} = X_{\anode} \backslash \{\vertexTreeDec\}$, for some vertex $\vertexTreeDec$; (iii) $X_{\anode'} = X_{\anode}\cup \{\vertexTreeDec\}$ for some vertex $\vertexTreeDec$. In the first two cases, an injective coloring of the vertices in $X_{\anode'}$ has already been chosen. In the last case, we just choose some arbitrary color for $\vertexTreeDec$ in the set $[\treewidthvalue+1]\backslash \alpha(X_{\anode})$. We proceed in this way until all bags have been visited. Since each vertex occurs in some bag, each vertex has received some color from the set $[\treewidthvalue+1]$. Since the endpoints of each edge
$\anedge$ are contained in some bag, and since by construction, the coloring in each bag is injective, we have that the endpoints of $\anedge$ receive distinct colors. Therefore, $\alpha$ is a proper $(\treewidthvalue+1)$-coloring of $\agraph$. 
\end{proof}

Now, let $\agraph$ be a graph of treewidth at most $\treewidthvalue$. Let $(\atree,X,\beta)$ be a nice, 
edge-introducing tree decomposition of $\agraph$ of width at most $\treewidthvalue$. Let $\alpha:\vertexset{\agraph}\rightarrow [\treewidthvalue+1]$ be a proper $(\treewidthvalue+1)$-coloring of $\agraph$ that is bag-injective for $(\atree,X,\beta)$. By Observation \ref{observation:InjectiveColoring} such a coloring exists
and can be constructed in time $O(\treewidthvalue\cdot |N|)$ where $N$ is the set of nodes of $\atree$. 
Then, a $\treewidthvalue$-instructive tree decomposition 
$\abstractdecomposition$ of $\agraph$ can be constructed as follows: We let $\abstractdecomposition$ have the
same structure as $T$, except that instead being labeled with bags, the nodes are now labeled with instructions from the alphabet $\abstractalphabet{\treewidthvalue}$.
More specifically, each leaf node is labeled with the instruction $\leaftype$. Additionally, for each non-leaf node $\anode$ of $\atree$:
If a vertex $x$ is introduced at $\anode$, then we label the node $\anode$ of $\abstractdecomposition$ with the instruction $\introvertextype{\alpha(x)}$.
If an edge $\{x,y\}$ is introduced at $\anode$, then the node $\anode$ of $\abstractdecomposition$ is labeled with the instruction $\introedgetype{\alpha(x)}{\alpha(y)}$. 
If a vertex $x$ is forgotten at $\anode$, then the node $\anode$ of $\abstractdecomposition$ is labeled with the instruction $\forgetvertextype{\alpha(x)}$.
Finally, $\anode$ is a join node, then the node $\anode$ of $\abstractdecomposition$ is labeled with the instruction  $\jointype$. 
It follows by induction on the height $\atree$ that there is an
isomorphism $\isomorphism  = (\isomorphismvertices,\isomorphismedges)$
from $\agraph$ to $\decompositiongraph{\abstractdecomposition}$ such that for each vertex $x\in \vertexset{\agraph}$, $x$ belongs to the root 
bag $X_r$ if and only $\alpha(x)\in \topbag{\abstractdecomposition}$ and 
$\topmap{\abstractdecomposition}(\alpha(x)) = \isomorphismvertices(x)$.
This concludes the proof of Lemma \ref{lemma:TreewidthIsTreelikeConfirmation}. 
$\square$

\section{Formal Definition of Internally Polynomial DP-Cores}
\label{section:formalDefinitionInternallyPolynomial} 

A DP-core $\dpcore$ is internally polynomial if 
there is an algorithm $\analgorithm$ that when given $\treewidthvalue\in \N$ as input simulates the behavior
of $\dpcore[\treewidthvalue]$ in such a way that the output of the invariant function has polynomially many bits in
the bitlength of $\dpcore[\treewidthvalue]$; $\analgorithm$ decides membership in the set $\dpcore[\treewidthvalue].\allwitnesses$ in time polynomial in $\treewidthvalue$ plus the 
size of the size of the queried string; for each symbol $a$ of arity $0$, $\analgorithm$ constructs the set $\dpcore[\treewidthvalue].\hat{a}$ 
in time polynomial in $\treewidthvalue$ plus the maximum number of bits needed to describe such a set; 
and $\analgorithm$ computes each function in $\dpcore[\treewidthvalue]$ in time polynomial in $\treewidthvalue$ plus the size of the input of 
the function.
Note that the fact that $\dpcore$ is internally polynomial does not imply that one can determine whether a given
term $\abstractdecomposition$ is accepted by $\dpcore$ in time polynomial in
$|\abstractdecomposition|$. The complexity of this test 
is governed by the bitlength and multiplicity of the DP-core in question
(see Theorem \ref{theorem:ModelChecking}).

\begin{definition}[Internally Polynomial DP-Cores]
\label{definition:PolynomialTimeComputable}
Let $\dpcore$ be a treelike DP-core. We say that $\dpcore$ is internally polynomial if the following conditions are satisfied. 
\begin{enumerate} 
\item For each $\treewidthvalue\in \N$, and each $(\dpcore,\treewidthvalue,\sizedecomposition)$-useful set $\witnessset$,
$|\dpcore[\treewidthvalue].\invariantCore(\witnessset)| = \bitlengthusefulwitnesssimple{\dpcore}(\treewidthvalue,\sizedecomposition)^{O(1)}$. 
\item There is a deterministic algorithm $\analgorithm$ such that the following conditions are satisfied. 
\begin{enumerate}
\item Given $k\in \N$, and a string $w \in \{0,1\}^*$, $\analgorithm$ decides whether $w \in \dpcore[\treewidthvalue].\allwitnesses$ in time $(\treewidthvalue+|w|)^{O(1)}$,
\item Given $\treewidthvalue\in \N$, and a symbol $\asymbol\in \alphabet_{\treewidthvalue}$ of arity $0$, the algorithm constructs the set $\dpcore[\treewidthvalue].\hat{\asymbol}$ in time 
$\bitlengthusefulwitnesssimple{\dpcore}(\treewidthvalue,0)\cdot \maxsizeusefulsetsimple{\dpcore}(\treewidthvalue,0)$. 
\item Given $\treewidthvalue\in \N$, a symbol $a\in \alphabet_k$, and an input $X$ for the function $\dpcore[\treewidthvalue].\hat{a}$, the algorithm constructs the set  
	$\dpcore[\treewidthvalue].\hat{a}(X)$ in time $(\treewidthvalue+|X|)^{O(1)}$.
\item Given $k\in \N$, an element $\mathtt{Function} \in \{\finalwitnessgenericcore,\cleaningfunctioncore,\invariantCore\}$, and an input $X$ for 
	the function $\dpcore[\treewidthvalue].\mathtt{Function}$, the algorithm computes the value $\dpcore[\treewidthvalue].\mathtt{Function}(X)$ in time 
	 $(\treewidthvalue+|X|)^{O(1)}$.
\end{enumerate}
\end{enumerate}
\end{definition}

\section{Proof of Theorem \ref{theorem:ModelChecking}}
\label{section:ProofTheoremModelChecking}
    Since $\dpcore$ is $\decompositionclass$-coherent, and since $\abstractdecomposition$ has $\decompositionclass$-width at most $\treewidthvalue$, 
    by Proposition \ref{proposition:CoherencySequenceTuples}, we have that $\decompositiongraph{\abstractdecomposition}$ belongs to 
    $\dpcoregraphpropertyclass{\dpcore}{\decompositionclass}$ if and only if $\abstractdecomposition\in \accepteddecompositionsboundedwidth{\dpcore[\treewidthvalue]}$.
    In other words, if and only if the set $\dynamizationfunction{\dpcore}{\treewidthvalue}(\abstractdecomposition)$ has some final local witness.
    Therefore, the model-checking algorithm consists of two steps: we first compute the set 
    $\dynamizationfunction{\dpcore}{\treewidthvalue}(\abstractdecomposition)$ inductively using Definition \ref{definition:Dynamization}, 
    and subsequently, we test whether this set has some final local witness. Since $|\abstractdecomposition| = n$, we have that 
    $\abstractdecomposition$ has at most $n$ sub-terms. Let $\abstractdecomposition$ be such a subterm. If $\sigmaabstractdecomposition=\asymbol$
    for some symbol of arity $0$, then the construction of the set $\dynamizationfunction{\dpcore}{\treewidthvalue}(\sigmaabstractdecomposition)$
    takes time at most $\bitlengthusefulwitnesssimple{\dpcore}(\treewidthvalue,\sizedecomposition)^{O(1)}\cdot \maxsizeusefulsetsimple{\dpcore}(\treewidthvalue,\sizedecomposition)^{O(1)}$.
    Now, suppose that $\sigmaabstractdecomposition = \asymbol(\sigmaabstractdecomposition_{1},\dots,\sigmaabstractdecomposition_{\arity(\asymbol)})$
    for some symbol $\asymbol\in \alphabet_{\treewidthvalue}$, and some terms $\sigmaabstractdecomposition_1,\dots,\sigmaabstractdecomposition_{\arity(\asymbol)}$ in
    $\Terms{\alphabet_{\treewidthvalue}}$, and assume that the sets $\dynamizationfunction{\dpcore}{\treewidthvalue}(\sigmaabstractdecomposition_1), \dots, 
    \dynamizationfunction{\dpcore}{\treewidthvalue}(\sigmaabstractdecomposition_{\arity(\asymbol)})$ have been computed. 
    We claim that using these precomputed sets, together with Equation \ref{equation:NextStepDynamization},
    the set $\dynamizationfunction{\dpcore}{\treewidthvalue}(\sigmaabstractdecomposition)$ can be constructed in time
    $\treewidthvalue^{O(1)}\cdot \arityvalue^{O(1)}\cdot \bitlengthusefulwitnesssimple{\dpcore}(\treewidthvalue,\sizedecomposition)^{O(1)} \cdot 
    \maxsizeusefulsetsimple{\dpcore}(\treewidthvalue,\sizedecomposition)^{\arityvalue+ O(1)}$. To see this, we note that in order to construct 
    $\dynamizationfunction{\dpcore}{\treewidthvalue}(\sigmaabstractdecomposition)$, 
    we need to construct, for each tuple $(\awitness_1,\awitness_2,\dots,\awitness_{\arity(\asymbol)})$ of local witnesses in 
    the Cartesian product $\dynamizationfunction{\dpcore}{\treewidthvalue}(\sigmaabstractdecomposition_1) \times \dots \times  \dynamizationfunction{\dpcore}{\treewidthvalue}(\sigmaabstractdecomposition_{\arity(\asymbol)})$, 
    the set $\dpcore[\treewidthvalue].\hat{\asymbol}(\awitness_1,\dots,\awitness_{\arity(\asymbol)})$. Since, by assumption, $\dpcore$ is internally polynomial,
    this set can be constructed in time at most $\treewidthvalue^{O(1)} \cdot  \arityvalue^{O(1)}\cdot \bitlengthusefulwitnesssimple{\dpcore}(\treewidthvalue,\sizedecomposition)^{O(1)}$. That is to say, polynomial 
    in $\treewidthvalue$ plus the size of the input to this function, which is upper bounded by $\arityvalue\cdot \bitlengthusefulwitnesssimple{\dpcore}(\treewidthvalue,\sizedecomposition)$. 
    In particular, this set has at most $\treewidthvalue^{O(1)}\cdot  \arityvalue^{O(1)} \cdot \bitlengthusefulwitnesssimple{\dpcore}(\treewidthvalue,\sizedecomposition)^{O(1)}$ local witnesses.
    Since we need to consider at most $\maxsizeusefulsetsimple{\dpcore}(\treewidthvalue,\sizedecomposition)^{\arityvalue}$ tuples, taking the union of all such sets 
    $\dpcore[\treewidthvalue].\hat{\asymbol}(\awitness_1,\dots,\awitness_{\arity(\asymbol)})$ takes time 
    $\treewidthvalue^{O(1)}\cdot\arityvalue^{O(1)}\cdot\bitlengthusefulwitnesssimple{\dpcore}(\treewidthvalue,\sizedecomposition)^{O(1)}\cdot\maxsizeusefulsetsimple{\dpcore}(\treewidthvalue,\sizedecomposition)^{\arityvalue + O(1)}$.
    Finally, once the union has been computed, since $\dpcore$ is internally polynomial, the application of the function $\dpcore[\treewidthvalue].\cleaningfunctioncore$ takes an additional (additive) factor of
    $\treewidthvalue^{O(1)}\cdot \bitlengthusefulwitnesssimple{\dpcore}(\treewidthvalue,\sizedecomposition)^{O(1)}\cdot \maxsizeusefulsetsimple{\dpcore}(\treewidthvalue,\sizedecomposition)^{O(1)}$. 
    Therefore, the overall computation of the set $\dynamizationfunction{\dpcore}{\treewidthvalue}(\abstractdecomposition)$, and subsequent determination of whether this set 
    contains a final local witness takes time 
    $$T(\treewidthvalue,\sizedecomposition) = \sizedecomposition\cdot \treewidthvalue^{O(1)}\cdot\arityvalue^{O(1)}\cdot\bitlengthusefulwitnesssimple{\dpcore}(\treewidthvalue,\sizedecomposition)^{O(1)}\cdot\maxsizeusefulsetsimple{\dpcore}(\treewidthvalue,\sizedecomposition)^{\arityvalue + O(1)}.$$
    
    For the second item, after having computed $\dynamizationfunction{\dpcore}{\treewidthvalue}(\abstractdecomposition)$ in 
    time $T(\treewidthvalue,\sizedecomposition)$, we need to compute the value of $\dpcore[\treewidthvalue].\invariantCore$ on this set.
    This takes time 
    $\treewidthvalue^{O(1)}\cdot \bitlengthusefulwitnesssimple{\dpcore}(\treewidthvalue,\sizedecomposition)^{O(1)}\cdot \maxsizeusefulsetsimple{\dpcore}(\treewidthvalue,\sizedecomposition)^{O(1)}$.
    Therefore, in overall, we need time 
    $T(\treewidthvalue,n) + 
    \treewidthvalue^{O(1)}\cdot \bitlengthusefulwitnesssimple{\dpcore}(\treewidthvalue,\sizedecomposition)^{O(1)}\cdot \maxsizeusefulsetsimple{\dpcore}(\treewidthvalue,\sizedecomposition)^{O(1)}$.
    to compute the value $\invariant[\dpcore,\decompositionclass](\decompositiongraph{\abstractdecomposition})$.
   $\square$
\section{Proof of Theorem \ref{theorem:Inclusion}}
\label{section:proofTheoremInclusionTest}
    Using breadth-first search, we successively enumerate the $(\realizationclass,\dpcore,\treewidthvalue)$-pairs that 
    can be derived using rule \ref{condition:dp-refutation-two} of Definition \ref{definition:DPRefutation}. We may assume that 
    the first traversed pairs are those in the set
    $$\{(\astate,\dpcore[\treewidthvalue].\hat{\asymbol})\; | \; \arity(\asymbol)=0,\; \asymbol\rightarrow \astate \mbox{ is a transition of $\realizationclass_{\treewidthvalue}$}\},$$ 
    which correspond to symbols of arity $0$. 	
    We run this search until either an inconsistent $(\realizationclass,\dpcore,\treewidthvalue)$-pair has been reached,
    or until there is no $(\realizationclass,\dpcore,\treewidthvalue)$-pair left to be enumerated. 
    In the first case, the obtained list of pairs
    \begin{equation}
    \label{equation:TriplesRefutation}
    \dprefutation = (\astate_1,\witnessset_1)(\astate_2,\witnessset_2)\dots(\astate_m,\witnessset_m)
    \end{equation}
    is, by construction, an $(\realizationclass,\dpcore,\treewidthvalue)$-refutation since 
    $(\astate_m,\witnessset_m)$ is inconsistent and for each $i\in [m]$, $(\astate_i,\witnessset_i)$
    has been obtained by applying the rule \ref{condition:dp-refutation-two} of Definition \ref{definition:DPRefutation}.
    In the second case, no such refutation exists. More specifically, in the beginning of the process, $\dprefutation$ is the empty list, and the pairs 
    $\{(\realizationclass_\treewidthvalue.\asymbol,\dpcore[\treewidthvalue].\hat{\asymbol}) : \; \arity(\asymbol)=0\}$ are added to a {\em buffer set} $Y$. While $Y$ is non-empty, we delete an arbitrary pair $(\astate,\witnessset)$ from $Y$ and append 
    this pair to $\dprefutation$. If the pair is inconsistent, we have constructed a refutation $\dprefutation$, and therefore, we return $\dprefutation$. 
    Otherwise, for each $(\astate',\witnessset')$ that can be obtained from $(\astate,\witnessset)$ using the rule \ref{condition:dp-refutation-two} of Definition \ref{definition:DPRefutation} (together with another pair from $\dprefutation$), we insert $(\astate',\witnessset')$ to the buffer $Y$ provided this pair is not already in $\dprefutation$. We repeat this process until either $\dprefutation$ has been returned 
    or until $Y$ is empty. In this case, we conclude that 
    $\dpcoregraphproperty{\decompositionclass_\treewidthvalue}\subseteq \dpcoregraphproperty{\dpcore[\treewidthvalue],\decompositionclass}$, and return 
    {\em Inclusion Holds}. This construction is detailed in Algorithm \ref{algorithm:InclusionTest}.
    
    Now, an upper bound on the running time of the algorithm can be established as follows. 
    Since $\realizationclass$ has complexity $f(\treewidthvalue)$, the breadth-first search constructs at most
    $f(\treewidthvalue)\cdot \delta_{\dpcore}(\treewidthvalue)$ pairs of the form $(\astate,\witnessset)$ where 
    $\astate$ is a state of $\treeautomaton_{\treewidthvalue}$, and $\witnessset\subseteq \dpcore[\treewidthvalue].\allwitnesses$. 
    Furthermore, since $\decompositionclass$ has arity at most $\arityvalue$, the creation of a new pair may require the analysis of at most 
    $f(\treewidthvalue)^{\arityvalue}\cdot \delta_{\dpcore}(\treewidthvalue)^{\arityvalue}$ tuples of previously created pairs. 
    For each such a tuple $[(\astate_1,\witnessset_1),(\astate_2,\witnessset_2),\dots,(\astate_{\arityvalue},\witnessset_{\arityvalue'})]$ with $r'\leq r$,
    the computation of the state $\astate'$ from the tuple $(\astate_1,\dots,\astate_{\arityvalue})$ takes time $f(k)^{O(1)}$ and,
    as argued in the proof of Theorem \ref{theorem:ModelChecking}, the computation of the set 
    $\witnessset' = \dpcore[\treewidthvalue].\hat{\asymbol}(\witnessset_{j_1},\dots,\witnessset_{j_{\arity(a)}})$ from the 
    tuple $(\witnessset_{j_1},\dots,\witnessset_{j_{\arity(a)}})$ takes time
    $\treewidthvalue^{O(1)}\cdot\arityvalue^{O(1)}\cdot\bitlengthusefulwitnesssimple{\dpcore}(\treewidthvalue)^{O(1)}\cdot
    \maxsizeusefulsetsimple{\dpcore}(\treewidthvalue)^{\arityvalue + O(1)}$. 
    Therefore, the whole process takes time at most 
    \begin{equation}
    \label{equation:Expression}
    \treewidthvalue^{O(1)}\cdot\arityvalue^{O(1)}\cdot\bitlengthusefulwitnesssimple{\dpcore}(\treewidthvalue)^{O(1)}\cdot
    \maxsizeusefulsetsimple{\dpcore}(\treewidthvalue)^{\arityvalue + O(1)} \cdot f(\treewidthvalue)^{\arityvalue+O(1)} 
    \cdot \numberusefulsetssimple{\dpcore}(\treewidthvalue)^{\arityvalue + O(1)}.
    \end{equation}
    
    By assumption, $f(k)\geq k$. Additionally, 
    $\maxsizeusefulsetsimple{\dpcore}(\treewidthvalue)\leq 2^{\bitlengthusefulwitnesssimple{\dpcore}(\treewidthvalue)}$ 
    and $\numberusefulsetssimple{\dpcore}(\treewidthvalue)  \leq  2^{\bitlengthusefulwitnesssimple{\dpcore}(\treewidthvalue)\cdot \maxsizeusefulsetsimple{\dpcore}(\treewidthvalue)+1}$. Therefore, Expression \ref{equation:Expression} can be simplified to 
    $$f(\treewidthvalue)^{O(r)}\cdot 2^{O(r\cdot \bitlengthusefulwitnesssimple{\dpcore}(\treewidthvalue)\cdot \maxsizeusefulsetsimple{\dpcore}(\treewidthvalue))} 
    \leq f(\treewidthvalue)^{O(r)}\cdot 2^{r\cdot 2^{O(\bitlengthusefulwitnesssimple{\dpcore}(\treewidthvalue))}}.$$ 
    $\square$
    
\SetKwComment{Comment}{/* }{ */}
\SetKwInOut{Input}{Input}
\SetKwInOut{Output}{Output}
\begin{algorithm}[h]
\caption{Inclusion Test}
\label{algorithm:InclusionTest}
\Input{An automation $\realizationclass$ for $\decompositionclass$, a finite, $\decompositionclass$-coherent treelike DP-core $\dpcore$, and an integer $\treewidthvalue\in \N$.}
\Output{An $(\realizationclass,\dpcore,\treewidthvalue)$-refutation $R$ if 
$\dpcoregraphproperty{\decompositionclass_\treewidthvalue}\nsubseteq \dpcoregraphproperty{\dpcore[\treewidthvalue],\decompositionclass}$, 
and "Inclusion Holds", otherwise.}
$\dprefutation \gets [\;]$ \Comment*[r]{$[\;]$ is the empty list.}
 $\hat{\dprefutation}\gets \{\;\}$ \Comment*[r]{$\{\;\}$ is the empty set.}
	$Y \gets \{(\astate,\dpcore.\hat{\asymbol}) : \; \arity(\asymbol)=0,\; a\rightarrow q \mbox{ is a transition of $\realizationclass_{\treewidthvalue}$}.\}$\;
\label{AlgoLineY}
$\mathrm{InconsistentPair} \gets \mathrm{false}$\;
\While{$Y\neq \emptyset$ and $\mathrm{InconsistentPair}= \mathrm{false}$
\label{AlgoWhileLoop}
}{
Remove some pair $(\astate,\witnessset)$ from $Y$, append it to $\dprefutation$, and insert it in $\hat{\dprefutation}$\;
\eIf{$(\astate,\witnessset)$ 
        is an inconsistent  $(\realizationclass,\dpcore,\treewidthvalue)$-pair,}{
        \Return $\dprefutation$ \Comment*[r]{In this case, $\dprefutation$ is a $(\realizationclass,\dpcore,\treewidthvalue)$-refutation.}
        }{
        \ForEach{$\asymbol \in \alphabet_\treewidthvalue$}{
		\ForEach{sequence $(\astate_{1},\witnessset_{1})\dots(\astate_{{\arity(\asymbol)}},\witnessset_{{\arity(\asymbol)}}) \mbox{ of pairs from } \hat{\dprefutation}^{\arity(\asymbol)}$
		having at least one occurrence of $(\astate,\witnessset)$}{
			$(\astate',\witnessset') \leftarrow (\asymbol(\astate_1,\astate_2,\dots,\astate_{\arity(\asymbol)}),\dpcore[\treewidthvalue].\hat{\asymbol}(\witnessset_1,\witnessset_2,\dots,\witnessset_{\arity(\asymbol)}))$\; 
		\lIf {$(\astate',\witnessset')$ is not in $\hat{\dprefutation}$}{
			$Y \leftarrow Y \cup \{(\astate',\witnessset')\}$
                }
            }
        }
        
        }
}
\Return "Inclusion Holds."
\end{algorithm}

\section{Example: A Coherent DP-Core for $\vertexcoverProperty_{\vertexcoverParameter}$}
\label{subsection:DPCoreVertexCover}

Let $\agraph=(\vertexsetname,\edgesetname,\incidencerelationname)$ be a graph.
A subset $\vertexcover$ of $\vertexsetname$ is a vertex cover of $\agraph$ if
every edge of $\agraph$ has at least one endpoint in $\vertexcover$. We let
$\vertexcoverProperty_{\vertexcoverParameter}$ be the graph property
consisting of all graphs that have a vertex cover of size at most
$\vertexcoverParameter$. 

In this section, for illustration purposes, we formally define a coherent DP-core
deciding the property $\vertexcoverProperty_{\vertexcoverParameter}$, and provide a proof of correctness for this DP-core.  
More specifically, let $\instructivetreedecompositionclass =
\{(\alphabet_k,\allabstractdecompositionstreewidth{\treewidthvalue},\graphfunction_{\treewidthvalue})\}_{\treewidthvalue\in
\N}$ be the decomposition class defined in Section \ref{DPRealization}, which
realizes treewidth. We will define an
$\instructivetreedecompositionclass$-coherent treelike DP-core
$\vertexcoverCore_{\vertexcoverParameter}$, with graph property
$\dpcoregraphproperty{\vertexcoverCore_{\vertexcoverParameter},\instructivetreedecompositionclass}
= \vertexcoverProperty_{\vertexcoverParameter}$.
It is enough to specify, for each $\treewidthvalue\in \N$, the components of $\vertexcoverCore_{\vertexcoverParameter}[\treewidthvalue]$.
A local witness for $\vertexcoverCore[\treewidthvalue]$ is a pair
$\awitness=(\vertexcoverset,\vertexcoversize)$ where $\vertexcoverset \subseteq
[\treewidthvalue+1]$ and $\vertexcoversize \in \N$. Intuitively,
$\vertexcoverset$ denotes the set of active labels associated with vertices of
a partial vertex cover, and $\vertexcoversize$ denotes the size of the partial
vertex cover. Therefore, we set
$$\vertexcoverCore_{\vertexcoverParameter}[\treewidthvalue].\allwitnesses =  \{(\vertexcoverset,\vertexcoversize)\;:\; 
\vertexcoverset \subseteq [\treewidthvalue+1],\;\vertexcoversize \in \{0,1,\dots,\vertexcoverParameter\}\}.$$
In this specific DP-core, each local witness is final. In other words, for each local 
witness $\awitness$, we have 
$$\vertexcoverCore_r[\treewidthvalue].\finalwitnessgenericcore(\awitness) = 1.$$
If $\witnessset$ is a set of local witnesses, and $(R,s)$ and $(R,s')$ are local witnesses in $S$ with 
$s<s'$, then $(R,s')$ is redundant.
The clean function of the DP-core takes a set of local 
witnesses as input and removes redundancies. More precisely, 
$$\vertexcoverCore_r[\treewidthvalue].\cleaningfunctioncore(\witnessset)=\{(R,s)\in S: \; \nexists s'<s,\; (R,s')\in S\}.$$
The invariant function of the core takes a set of witnesses as input and returns the smallest value $s$ with the property 
that there is some subset $R\subseteq [\treewidthvalue+1]$ with $(R,s)\in \witnessset$. If $S=\emptyset$, it returns the value $\infty$, indicating that the minimum size of a vertex cover in the graph is greater than $r$.  
$$
\vertexcoverCore_r[\treewidthvalue].\invariantCore(\witnessset)= 
\left\{
\begin{array}{l}
\min\{\vertexcoversize :\; \exists \vertexcoverset \text{ s.t. } (\vertexcoverset,\vertexcoversize)\in \witnessset \} \mbox{ if $S\neq \emptyset$.} \\ 
\infty \mbox{ if $S=\emptyset$}. 
\end{array}
\right.$$
Next, we define the transition functions of the DP-core.

\begin{definition}
\label{definition:vertex-cover}
	Let $\awitness= (\vertexcoverset,\vertexcoversize)$ and $\awitness'= (\vertexcoverset',\vertexcoversize')$ be local witnesses, and $\vertexone,\vertextwo\in [\treewidthvalue+1]$ be such that $\vertexone\neq \vertextwo$. 
\begin{enumerate} 
\setlength\itemsep{1em}
\item \label{definition:vertex-cover-leaf} $\vertexcoverCore_r[\treewidthvalue].\initialsetgenericcore = \{(\emptyset,0)\}$.
\item \label{definition:vertex-cover-introvertex} $\vertexcoverCore_r[\treewidthvalue].\introvertextype{\vertexone}(\awitness) = 
\{\awitness\}$.
\item \label{definition:vertex-cover-forgetvertex} $\vertexcoverCore_r[\treewidthvalue].\forgetvertextype{\vertexone}(\awitness) = \{(\vertexcoverset\setminus\{\vertexone\},\vertexcoversize)\}$. 
\item \label{definition:vertex-cover-introedge} $\vertexcoverCore_r[\treewidthvalue].\introedgegeneric{\vertexone}{\vertextwo}(\awitness) = 
\begin{cases}
	\{\awitness\} \hspace{4.4cm} \text{if } \vertexone \text{ or }\vertextwo \in \vertexcoverset,\\
\emptyset \hspace{3.6cm} \text{if }\vertexone,\vertextwo\notin \vertexcoverset \text{ and }\vertexcoversize = \vertexcoverParameter, \\
\{(\vertexcoverset\cup\{\vertexone\},\vertexcoversize+1),(\vertexcoverset\cup\{\vertextwo\},\vertexcoversize+1)\} 
\mbox{ otherwise.}
\end{cases}$
\item \label{definition:vertex-cover-join} $\vertexcoverCore_r[\treewidthvalue].\joingenericcore(\awitness,\awitness')= 
\begin{cases}
\{(\vertexcoverset\cup\vertexcoverset',\vertexcoversize+\vertexcoversize' - |\vertexcoverset\cap\vertexcoverset'|)\} & \text{If } \vertexcoversize+\vertexcoversize' - |\vertexcoverset\cap\vertexcoverset'| \leq \vertexcoverParameter, \\ 
\{\}  & \text{otherwise.}
\end{cases}
$
\end{enumerate}
\end{definition}

It should be clear that $\vertexcoverCore_{r}[\treewidthvalue]$ is
both finite and internally polynomial. The next lemma characterizes, 
for each $\treewidthvalue$-instructive tree decomposition $\abstractdecomposition$, 
the local witnesses $\awitness$ that are present in the set 
$\dynamizationfunction{\vertexcoverCore_r}{\treewidthvalue}(\abstractdecomposition)$. Given a $\treewidthvalue$-instructive tree decomposition $\abstractdecomposition$, we let $(\decompositiongraph{\abstractdecomposition},\topmap{\abstractdecomposition})$ be the 
$\treewidthvalue$-boundaried graph defined by $\abstractdecomposition$ according to the construction in Appendix \ref{formalDefinitionITD}. We
let $\topbag{\abstractdecomposition}$ denote the domain of $\topmap{\abstractdecomposition}$, that is to say, the set of active labels corresponding to the root node of $\abstractdecomposition$. We also let 
$\topmap{\abstractdecomposition}(\topbag{\abstractdecomposition})$ be the image of the map $\topmap{\abstractdecomposition}$, that is to say, the set of vertices of $\decompositiongraph{\abstractdecomposition}$ corresponding to 
labels in $\topbag{\abstractdecomposition}$. 

\begin{lemma}
\label{predicateLemma}
Let $\vertexcoverParameter\in \N$. For each $\treewidthvalue\in \N$, each $\treewidthvalue$-instructive tree decomposition $\abstractdecomposition$, and each 
local witness $\awitness = (\vertexcoverset,\vertexcoversize)$ in $\vertexcoverCore_{r}[\treewidthvalue].\allwitnesses$,
$\awitness\in \dynamizationfunction{\vertexcoverCore_r}{\treewidthvalue}(\abstractdecomposition)$ if and only if the following predicate is satisfied: 
\begin{itemize}
	\item $\vertexcoverPred_{\vertexcoverParameter}[\treewidthvalue](\abstractdecomposition,\awitness) \equiv $
\begin{enumerate}
	\item \label{condition:predicate-vertex-cover-one} $\vertexcoversize \leq \vertexcoverParameter$,
	\item \label{condition:predicate-vertex-cover-two} $\vertexcoversize$ is the minimum size of a vertex cover $\vertexcover$ in $\decompositiongraph{\abstractdecomposition}$ 
	with $\topmap{\abstractdecomposition}(\vertexcoverset) = \vertexcover \cap \topmap{\abstractdecomposition}(\topbag{\abstractdecomposition})$.
\end{enumerate}
\end{itemize}
\end{lemma}

The proof of Lemma \ref{predicateLemma} follows by induction on the structure of $\abstractdecomposition$. For completeness, and also for illustration purposes,
a detailed proof can be found in Appendix \ref{section:ProofPredicateLemma}. Lemma \ref{predicateLemma} implies that $\vertexcoverCore_r$ is coherent and that 
for each $\treewidthvalue\in \N$, the graph property $\dpcoregraphproperty{\vertexcoverCore_r[\treewidthvalue],\instructivetreedecompositionclass}$ 
is the set of all graphs of treewidth at most $\treewidthvalue$ with a vertex cover of size at most $r$.

\begin{theorem}
\label{theorem:CorrectnessVertexCover}
For each $r\in \N$, the DP-core $\vertexcoverCore_r$ is coherent. Additionally, for each $\treewidthvalue\in \N$, 
$\dpcoregraphproperty{\vertexcoverCore_r[\treewidthvalue],\instructivetreedecompositionclass} = \vertexcoverProperty_r \cap \allgraphstreewidth{\treewidthvalue}$. 
\end{theorem}
\begin{proof}
Let $\treewidthvalue\in \N$, $\agraph$ be a graph of treewidth at most $\treewidthvalue$, and 
$\abstractdecomposition\in \allabstractdecompositionstreewidth{\treewidthvalue}$ be such that $\decompositiongraph{\abstractdecomposition}\simeq \agraph$. 
Then, by Lemma \ref{predicateLemma}, $\dynamizationfunction{\vertexcoverCore_r}{\treewidthvalue}(\abstractdecomposition)$
is nonempty (i.e. has some local witness) if and only if $\agraph$ has a vertex cover of size at most $r$. Since, for this specific DP-core, 
every local witness is final, we have that $\abstractdecomposition\in \accepteddecompositionsboundedwidth{\vertexcoverCore_r[\treewidthvalue]}$ if and only
if $\agraph$ has a vertex cover of size at most $r$. Therefore, $\agraph\in\dpcoregraphproperty{\vertexcoverCore_r[\treewidthvalue],\instructivetreedecompositionclass}$ if and only
if $\agraph$ has a vertex cover of size at most $r$. 
\end{proof}

Now, consider the DP-core $\vertexcoverCore$ (i.e. without the subscript $r$) where for each $\treewidthvalue\in \N$, all 
components are identical to the components of $\vertexcoverCore_r$, except for the local witnesses $(\vertexcoverset,\vertexcoversize)$, 
where now $\vertexcoversize$ is allowed to be any number in $\N$, and 
for the edge introduction function 
$\vertexcoverCore[\treewidthvalue].\introedgegeneric{\vertexone}{\vertextwo}$ which is defined as follows on 
each local witness $\awitness = (\vertexcoverset,\vertexcoversize)$. 
$$\vertexcoverCore[\treewidthvalue].\introedgegeneric{\vertexone}{\vertextwo}(\awitness) = 
\begin{cases}
\{\awitness\} &  \text{if } \vertexone \text{ or }\vertextwo \in \vertexcoverset,\\
\{(\vertexcoverset\cup\{\vertexone\},\vertexcoversize+1),(\vertexcoverset\cup\{\vertextwo\},\vertexcoversize+1)\} &  \mbox{otherwise.}
\end{cases}$$

Then, in this case, the core is no longer finite because one can impose no bound on the value $s$ of a local witness $(R,s)$. 
Still, we have that the multiplicity of $\vertexcoverCore[\treewidthvalue]$ is bounded by $2^{\treewidthvalue+1}$ (i.e, a function of $\treewidthvalue$ only) because the clean function eliminates redundancies.
Additionally, one can show that this variant actually computes the minimum size of a vertex cover in the graph represented by a $\treewidthvalue$-instructive tree decomposition. Below, we let
$\dpcore = \vertexcoverCore$ and $\invariant[\dpcore,\decompositionclass]:\dpcoregraphpropertyclass{\dpcore}{\decompositionclass}\rightarrow \{0,1\}^*$ be the graph invariant computed by $\dpcore$.

\begin{theorem}
\label{theorem:InvariantComputation}
Let $\abstractdecomposition$ be a $\treewidthvalue$-instructive tree decomposition. Then, $\invariant[\dpcore,\decompositionclass](\decompositiongraph{\abstractdecomposition})$ is the (binary encoding of) the minimum size of a vertex cover in $\decompositiongraph{\abstractdecomposition}$. 
\end{theorem}

We omit the proof of Theorem \ref{theorem:InvariantComputation} given that the proof is very similar to the proof of Theorem \ref{theorem:CorrectnessVertexCover}. 
In particular, this theorem is a direct consequence of an analog of Lemma \ref{predicateLemma} where 
$\vertexcoverCore_r$ is replaced by $\vertexcoverCore$ and the predicate $\texttt{P-VertexCover}_r$  is replaced by 
the predicate $\texttt{P-VertexCover}$ obtained by omitting the first condition  ($s\leq r$).

\subsection{Proof of Lemma \ref{predicateLemma}}
\label{section:ProofPredicateLemma}

\newcommand{\vertexcoverPredicate}{\texttt{P-VertexCover}}

In this section, we prove Lemma \ref{predicateLemma} by induction on the 
height of a $\treewidthvalue$-instructive tree decomposition $\abstractdecomposition$.
Although the proof is straightforward, it serves as an illustration of how the 
objects $\topbag{\abstractdecomposition}$, $\decompositiongraph{\abstractdecomposition}$, 
$\topmap{\abstractdecomposition}$ and $\dynamizationfunction{\dpcore}{\treewidthvalue}(\abstractdecomposition)$
interact with each other in an inductive proof of correctness of a DP-core $\dpcore$. In particular,
the structure of this proof can be adapted to provide inductive proofs for the correctness of DP-cores
deciding other graph properties. \\

\noindent{\bf Base Case: } In the base case, the height of $\abstractdecomposition$
is $0$, and therefore $\abstractdecomposition = \leaftype$. By Definition \ref{definition:Dynamization} and
Definition \ref{definition:vertex-cover}.\ref{definition:vertex-cover-leaf},
$\dynamizationfunction{\vertexcoverCore_{\vertexcoverParameter}}{\treewidthvalue}(\abstractdecomposition) = \vertexcoverCore_{\vertexcoverParameter}[\treewidthvalue].\initialsetgenericcore = \{(\emptyset,0)\}$. 
Since $\decompositiongraph{\abstractdecomposition} = (\emptyset,\emptyset,\emptyset)$, we have that $\emptyset$ is the only vertex cover of $\decompositiongraph{\abstractdecomposition}$, 
and since this vertex cover has size $0$, $\awitness = (\emptyset,0)$ is the only local witness with $\vertexcoverPredicate_r(\abstractdecomposition,\awitness)=1$. Therefore, the lemma
holds in the base case. \\ 
		
\noindent{\bf Inductive Step: }	
Now, assume that the lemma holds for every $\treewidthvalue$-instructive tree decomposition of height at most $h$. Let $\abstractdecomposition$ be 
a $\treewidthvalue$-instructive tree decomposition of height $h+1$, and 
$\awitness = (\vertexcoverset,\vertexcoversize)$ be a local witness in $\vertexcoverCore_{\vertexcoverParameter}[\treewidthvalue].\allwitnesses$. 
There are four cases to be considered. 
\begin{enumerate} 
\item \label{case:introvertex}
	Let $\abstractdecomposition = \introvertextype{\vertexone}(\sigmaabstractdecomposition)$. Suppose $\awitness\in \dynamizationfunction{\vertexcoverCore_{\vertexcoverParameter}}{\treewidthvalue}(\abstractdecomposition)$.
By Definition \ref{definition:Dynamization} and Definition \ref{definition:vertex-cover}.\ref{definition:vertex-cover-introvertex}, we have that 
$\awitness$ also belongs to $\dynamizationfunction{\vertexcoverCore_{\vertexcoverParameter}}{\treewidthvalue}(\sigmaabstractdecomposition)$. 
By the induction hypothesis, $\vertexcoverPredicate_r[\treewidthvalue](\sigmaabstractdecomposition,\awitness) = 1$.
Then, $\vertexcoversize\leq \vertexcoverParameter$, and $\vertexcoversize$ is the minimum size of a vertex cover $\vertexcover$ in $\decompositiongraph{\sigmaabstractdecomposition}$
with $\topmap{\sigmaabstractdecomposition}(\vertexcoverset) = \vertexcover \cap \topmap{\sigmaabstractdecomposition}(\topbag{\sigmaabstractdecomposition})$. This implies that
$\vertexcover$ is also a vertex cover in $\decompositiongraph{\abstractdecomposition}$ with
$\topmap{\abstractdecomposition}(\vertexcoverset) = \vertexcover \cap \topmap{\abstractdecomposition}(\topbag{\abstractdecomposition})$. 
Therefore, $\vertexcoverPredicate_r[\treewidthvalue](\abstractdecomposition,\awitness) = 1$. 

For the converse, let $\vertexcoverPredicate_r[\treewidthvalue](\abstractdecomposition,\awitness) = 1$.
Then, there is a vertex cover $\vertexcover$ of $\decompositiongraph{\abstractdecomposition}$ of size $\vertexcoversize$ where $\topmap{\abstractdecomposition}(\vertexcoverset) = \vertexcover \cap \topmap{\abstractdecomposition}(\topbag{\abstractdecomposition})$ and $\vertexcoversize$ is the minimum size of such vertex cover. Since $\edgesetname_{\decompositiongraph{\abstractdecomposition}} = \edgesetname_{\decompositiongraph{\sigmaabstractdecomposition}}$, the vertex cover $\vertexcover$ is a vertex cover of $\decompositiongraph{\sigmaabstractdecomposition}$ such that $\topmap{\sigmaabstractdecomposition}(\vertexcoverset) = \vertexcover \cap \topmap{\sigmaabstractdecomposition}(\topbag{\sigmaabstractdecomposition})$. Therefore, 
$\vertexcoverPredicate_{\vertexcoverParameter   }[\treewidthvalue](\sigmaabstractdecomposition,\awitness) = 1$. By the induction hypothesis,
$\awitness \in \dynamizationfunction{\vertexcoverCore_{\vertexcoverParameter}}{\treewidthvalue}(\sigmaabstractdecomposition)$. By Definition \ref{definition:vertex-cover}.\ref{definition:vertex-cover-introvertex}, $\vertexcoverCore_{\vertexcoverParameter}[\treewidthvalue].\introvertextype{\vertexone}(\awitness) = \{\awitness\}$. Therefore, $\awitness \in \dynamizationfunction{\vertexcoverCore_{\vertexcoverParameter}}{\treewidthvalue}(\abstractdecomposition)$.

\item Let $\abstractdecomposition = \forgetvertextype{\vertexone}(\sigmaabstractdecomposition)$. Suppose $\awitness \in \dynamizationfunction{\vertexcoverCore_{\vertexcoverParameter}}{\treewidthvalue}(\abstractdecomposition)$. By Definition \ref{definition:vertex-cover}.\ref{definition:vertex-cover-forgetvertex}, there exists
	$\awitness' = (\vertexcoverset',\vertexcoversize')\in \dynamizationfunction{\vertexcoverCore_{\vertexcoverParameter}}{\treewidthvalue}(\sigmaabstractdecomposition)$ such that \\ 
$\vertexcoverCore_{\vertexcoverParameter}[\treewidthvalue].\forgetvertextype{\vertexone}(\awitness') = \{\awitness\}$. Note that $\vertexcoverset'\setminus \{\vertexone\} = \vertexcoverset $ and $\vertexcoversize ' = \vertexcoversize$.
By the induction hypothesis, $\vertexcoverPredicate_{\vertexcoverParameter}[\treewidthvalue](\sigmaabstractdecomposition,\awitness') = 1$, and therefore, there exists a vertex cover $\vertexcover$ of minimum size
	 $\vertexcoversize'$ in $\decompositiongraph{\sigmaabstractdecomposition}$ where $\topmap{\sigmaabstractdecomposition}(\vertexcoverset') = \vertexcover \cap \topmap{\sigmaabstractdecomposition}(\topbag{\sigmaabstractdecomposition})$.
Since $\decompositiongraph{\sigmaabstractdecomposition} = \decompositiongraph{\abstractdecomposition}$, the set $\vertexcover$ is a vertex cover of $\decompositiongraph{\abstractdecomposition}$. 
We can infer the following by the fact that $\topbag{\abstractdecomposition } = \topbag{\sigmaabstractdecomposition}\setminus\{\vertexone\}$, $\vertexcoverset = \vertexcoverset'\setminus \{\vertexone\}$, and $\topmap{\abstractdecomposition} = \topmap{\sigmaabstractdecomposition}|_{\topbag{\abstractdecomposition}}$:
\begin{equation}\label{equation:vertex-cover-predicate-forgetvertex-one}
   \topmap{\abstractdecomposition}(\vertexcoverset) = \vertexcover \cap \topmap{\abstractdecomposition}(\topbag{\abstractdecomposition}).
\end{equation}
Since $\vertexcover$ is a vertex cover of minimum size in $\decompositiongraph{\sigmaabstractdecomposition}$ such that  $\topmap{\sigmaabstractdecomposition}(\vertexcoverset') = \vertexcover \cap \topmap{\sigmaabstractdecomposition}(\topbag{\sigmaabstractdecomposition})$ and since $\decompositiongraph{\abstractdecomposition} = \decompositiongraph{\sigmaabstractdecomposition}$, we have that $\vertexcover$ is also a vertex cover of minimum size in $\decompositiongraph{\abstractdecomposition}$ that satisfies Equation \ref{equation:vertex-cover-predicate-forgetvertex-one}. 
Consequently, 
$\vertexcoverPredicate_{\vertexcoverParameter}[\treewidthvalue](\abstractdecomposition,\awitness) = 1$.

For the converse, suppose $\vertexcoverPredicate_{\vertexcoverParameter}[\treewidthvalue](\abstractdecomposition,\awitness) = 1$.
Then,
there is a vertex cover $\vertexcover$ of minimum size such that  $\topmap{\abstractdecomposition}(\vertexcoverset) = \vertexcover \cap \topmap{\abstractdecomposition}(\topbag{\abstractdecomposition})$ and $|\vertexcover| = \vertexcoversize \leq \vertexcoverParameter$. Since $\decompositiongraph{\abstractdecomposition} = \decompositiongraph{\sigmaabstractdecomposition}$, the set $\vertexcover$ is a vertex cover of $\decompositiongraph{\sigmaabstractdecomposition}$. There are two cases to be considered:
\begin{enumerate}
    \item $\topmap{\sigmaabstractdecomposition}(\vertexone) \in \vertexcover$. In this case, 
    let $\awitness'= (\vertexcoverset',\vertexcoversize)$ where $\vertexcoverset' = \vertexcoverset \cup \{\vertexone\}$.
    \item $\topmap{\sigmaabstractdecomposition}(\vertexone) \notin \vertexcover$. In this case, let $\awitness' =  (\vertexcoverset',\vertexcoversize)$ where $\vertexcoverset' = \vertexcoverset$.
\end{enumerate}
		In both cases we will show that $\awitness' \in \dynamizationfunction{\vertexcoverCore_{\vertexcoverParameter}}{\treewidthvalue}(\sigmaabstractdecomposition)$.
  Recall that $\topbag{\abstractdecomposition } = \topbag{\sigmaabstractdecomposition}\setminus\{\vertexone\}$ and $ \topmap{\abstractdecomposition} = \topmap{\sigmaabstractdecomposition}|_{\topbag{\abstractdecomposition}}$, and therefore we have that
\begin{equation}\label{equation:vertex-cover-predicate-forgetvertex-two}
   \topmap{\sigmaabstractdecomposition}(\vertexcoverset') = \vertexcover \cap \topmap{\sigmaabstractdecomposition}(\topbag{\sigmaabstractdecomposition}).
\end{equation}
Equation \ref{equation:vertex-cover-predicate-forgetvertex-two} holds for the two cases and $\vertexcover$ is a vertex cover of minimum size satisfying Equation \ref{equation:vertex-cover-predicate-forgetvertex-two}.
		Therefore, $\vertexcoverPredicate_{\vertexcoverParameter}[\treewidthvalue](\sigmaabstractdecomposition,\awitness') = 1 $, and by the induction hypothesis, $\awitness' \in \dynamizationfunction{\vertexcoverCore_{\vertexcoverParameter}}{\treewidthvalue}(\sigmaabstractdecomposition)$.
		By Definition \ref{definition:vertex-cover}.\ref{definition:vertex-cover-forgetvertex}, $\vertexcoverCore_{\vertexcoverParameter}[\treewidthvalue].\forgetvertextype{\vertexone}(\awitness') = \{(\vertexcoverset'\setminus \{\vertexone\},\vertexcoversize)\}$ and we have that $(\vertexcoverset'\setminus \{\vertexone\},\vertexcoversize) = \awitness$, and therefore,
$\awitness\in \dynamizationfunction{\vertexcoverCore_{\vertexcoverParameter}}{\treewidthvalue}(\abstractdecomposition)$.

\item Let $\abstractdecomposition = \introedgetype{\vertexone}{\vertextwo}(\sigmaabstractdecomposition)$.
Suppose $\awitness \in \dynamizationfunction{\vertexcoverCore_{\vertexcoverParameter}}{\treewidthvalue}(\sigmaabstractdecomposition)$. By Definition \ref{definition:vertex-cover}.\ref{definition:vertex-cover-introedge}, there exists $\awitness'=(\vertexcoverset',\vertexcoversize') \in \dynamizationfunction{\vertexcoverCore_{\vertexcoverParameter}}{\treewidthvalue}(\sigmaabstractdecomposition)$ where \\ 
$\vertexcoverCore_{\vertexcoverParameter}[\treewidthvalue].\introedgetype{\vertexone}{\vertextwo}(\awitness') = \{\awitness\}$ and either $(\vertexcoverset,\vertexcoversize)  = (\vertexcoverset',\vertexcoversize')$ or $(\vertexcoverset,\vertexcoversize) = (\vertexcoverset'\cup\{\xvertex\},\vertexcoversize'+1)$ where $\xvertex\in\{\vertexone,\vertextwo\}$.
By the induction hypothesis, $\vertexcoverPredicate_{\vertexcoverParameter}[\treewidthvalue](\sigmaabstractdecomposition,\awitness') = 1$, and therefore, there is a vertex cover $\vertexcover$ of minimum size where $\topmap{\sigmaabstractdecomposition}(\vertexcoverset') = \vertexcover \cap \topmap{\sigmaabstractdecomposition}(\topbag{\sigmaabstractdecomposition})$.
We consider two cases and define $\vertexcover'$ with regards to these cases.
\begin{enumerate}
    \item $\vertexcoverset = \vertexcoverset'\cup\{\xvertex\}$ where $\xvertex \in \{\vertexone,\vertextwo\}$. In this case let $\vertexcover' = \vertexcover \cup \{\topmap{\abstractdecomposition}(\xvertex)\}$.
    \item $\vertexcoverset = \vertexcoverset'$. In this case let $\vertexcover' = \vertexcover$.
\end{enumerate}
Note that $\edgesetname_{\decompositiongraph{\abstractdecomposition}} = \edgesetname_{\decompositiongraph{\sigmaabstractdecomposition}}\cup \{\anedge\}$ where $(\anedge,\topmap{\abstractdecomposition}(\vertexone)), (\anedge,\topmap{\abstractdecomposition}(\vertextwo)) \in \incidencerelationname_{\decompositiongraph{\abstractdecomposition}}$ so that $\vertexcover'$ is a vertex cover of $\decompositiongraph{\abstractdecomposition}$. Recall that $\topbag{\abstractdecomposition} = \topbag{\sigmaabstractdecomposition}$ and $\topmap{\abstractdecomposition} = \topmap{\sigmaabstractdecomposition}$. We can infer the following equation by the fact that we have mentioned.
\begin{equation}\label{equation:vertex-cover-predicate-introedge-one}
    \topmap{\abstractdecomposition}(\vertexcoverset) = \vertexcover' \cap\topmap{\abstractdecomposition}(\topbag{\abstractdecomposition})
\end{equation}
If we show $\vertexcover'$ is a vertex cover of minimum size satisfying \ref{equation:vertex-cover-predicate-introedge-one}, then we have proved the statement. Next, we will prove the vertex cover $\vertexcover'$ is minimum size. Suppose $\vertexcover' = \vertexcover$ (the second case), then obviously $\vertexcover'$ is minimum size. 
Otherwise, suppose $\vertexcover' = \vertexcover \cup \{\topmap{\abstractdecomposition}(\xvertex)\}$ (the first case) and suppose $\vertexcover'$ is not a vertex cover of minimum size and there is a vertex cover $|\vertexcover''| < |\vertexcover'|$ where $\topmap{\abstractdecomposition}(\vertexcoverset) = \vertexcover'' \cap \topmap{\abstractdecomposition}(\topbag{\abstractdecomposition})$. Therefore, we have $\topmap{\abstractdecomposition}(\vertexcoverset') = (\vertexcover''\setminus\{\topmap{\abstractdecomposition}(\xvertex)\}) \cap \topmap{\abstractdecomposition}(\topbag{\sigmaabstractdecomposition})$ and we know $|\vertexcover''\setminus\{\topmap{\abstractdecomposition}(\xvertex)\}|<|\vertexcover|$, and therefore, this contradicts the assumption that  $\vertexcover$ is the minimum size vertex cover which satisfies 
the second condition of the definition of $\texttt{P-VertexCover}_r$.

For the converse, suppose $\vertexcoverPredicate_{\vertexcoverParameter}[\treewidthvalue](\abstractdecomposition,\awitness) = 1$. 
Then, there exists a vertex cover $\vertexcover$ where $\topmap{\abstractdecomposition}(\vertexcoverset) = \vertexcover \cap \topmap{\abstractdecomposition}(\topbag{\abstractdecomposition})$ where $\vertexcover$ is the minimum size of such vertex cover. In the following let $\anedge $ be the edge that $\topmap{\abstractdecomposition}(\vertexone)$ and $\topmap{\abstractdecomposition}(\vertextwo)$ are the endpoints of it.
\begin{enumerate}
    \item $\vertexcover\setminus \{\topmap{\abstractdecomposition}(\xvertex)\}$ where $\xvertex\in\{\vertexone,\vertextwo\} $ is a vertex cover of $\decompositiongraph{\abstractdecomposition} \setminus \{\anedge\}$. In this case let $\awitness' = (\vertexcoverset',\vertexcoversize')$ where $\vertexcoverset' = \vertexcoverset\setminus \{\vertexone,\vertextwo\}$ and $\vertexcoversize' = \vertexcoversize - 1$, and let $\vertexcover' = \vertexcover\setminus \{\topmap{\abstractdecomposition}(\xvertex)\} $.
    \item $\vertexcover\setminus \{\topmap{\abstractdecomposition}(\xvertex)\}$ where $\xvertex\in\{\vertexone,\vertextwo\}$ is not a vertex cover of $\decompositiongraph{\abstractdecomposition} \setminus \{\anedge\}$. In this case, we set $\awitness'=(\vertexcoverset',\vertexcoversize')$ where $\vertexcoverset'=\vertexcoverset$ and $\vertexcoversize' = \vertexcoversize$, and let $\vertexcover' = \vertexcover$.
\end{enumerate}
In both cases, we show that $\vertexcoverPredicate_{\vertexcoverParameter}[\treewidthvalue](\sigmaabstractdecomposition,\awitness') = 1$. 
The condition $\vertexcoversize'\leq \vertexcoverParameter$, is satisfied. Since $\decompositiongraph{\abstractdecomposition}  = \decompositiongraph{\sigmaabstractdecomposition}\setminus \{\anedge\}$, $\vertexcover'$ in both cases is a vertex cover of $\decompositiongraph{\sigmaabstractdecomposition}$. We have $\topbag{\abstractdecomposition} = \topbag{\sigmaabstractdecomposition}$ and $\topmap{\abstractdecomposition} = \topmap{\sigmaabstractdecomposition}$, and therefore in the both cases the following equation is satisfied.

\begin{equation}\label{equation:vertex-cover-predicate-introedge-two}
\topmap{\sigmaabstractdecomposition}(\vertexcoverset') = \vertexcover' \cap \topmap{\sigmaabstractdecomposition}(\topbag{\sigmaabstractdecomposition})
\end{equation}
Now it remains to show that $\vertexcover'$ is a minimum size vertex cover satisfying Equation \ref{equation:vertex-cover-predicate-introedge-two}. If $\vertexcover' = \vertexcover$, then since $\vertexcover$ has minimum size, $\vertexcover'$ has minimum size. 
Now, let $\vertexcover' = \vertexcover\setminus \{\topmap{\abstractdecomposition}(\xvertex)\}$ and suppose $\vertexcover'$ does not have minimum size. Then there exists a vertex cover $\vertexcover''$ where $\topmap{\sigmaabstractdecomposition}(\vertexcoverset') = \vertexcover'' \cap\topmap{\sigmaabstractdecomposition}(\topbag{\sigmaabstractdecomposition})$ such that $|\vertexcover''| < |\vertexcover'| $, and consequently, we have $|\vertexcover''\cup\{\topmap{\abstractdecomposition}(\xvertex)\}| < |\vertexcover|$ and  $\topmap{\abstractdecomposition}(\vertexcoverset) = (\vertexcover''\cup\{\topmap{\abstractdecomposition}(\xvertex)\}) \cap \topmap{\abstractdecomposition}(\topbag{\abstractdecomposition})$ and this contradicts the assumption that $\vertexcover$ has minimum size. Therefore, $\vertexcover'$ has minimum size, and consequently, $\vertexcoverPredicate_{\vertexcoverParameter}[\treewidthvalue](\sigmaabstractdecomposition,\awitness')=1$. By the induction hypothesis, $\awitness'\in \dynamizationfunction{\vertexcoverCore_{\vertexcoverParameter}}{\treewidthvalue}(\sigmaabstractdecomposition)$. By Definition \ref{definition:vertex-cover}.\ref{definition:vertex-cover-introedge}, $\vertexcoverCore_{\vertexcoverParameter}[\treewidthvalue](\awitness') = \{\awitness\}$, and therefore, $\awitness \in \dynamizationfunctionname{\vertexcoverCore_{\vertexcoverParameter}}{\treewidthvalue}(\abstractdecomposition)$.

\item Let $\abstractdecomposition = \jointype(\sigmaabstractdecomposition_1,\sigmaabstractdecomposition_2)$. Suppose $\awitness \in \dynamizationfunction{\vertexcoverCore_{\vertexcoverParameter}}{\treewidthvalue}(\abstractdecomposition)$. Let $\awitness_1=(\vertexcoverset_1,\vertexcoversize_1)\in \dynamizationfunction{\vertexcoverCore_{\vertexcoverParameter}}{\treewidthvalue}(\sigmaabstractdecomposition_1)$ and $\awitness_2=(\vertexcoverset_2,\vertexcoversize_2)\in \dynamizationfunction{\vertexcoverCore_{\vertexcoverParameter}}{\treewidthvalue}(\sigmaabstractdecomposition_2)$ be local witnesses such that $\vertexcoverCore_{\vertexcoverParameter}[\treewidthvalue].\jointype(\awitness_1,\awitness_2)=\{\awitness\}$. By Definition \ref{definition:vertex-cover}.\ref{definition:vertex-cover-join}, $\vertexcoverset = \vertexcoverset_1\cup\vertexcoverset_2$ and $\vertexcoversize = \vertexcoversize_1+\vertexcoversize_2-|\vertexcoverset_1\cap\vertexcoverset_2|$ where $\vertexcoversize \leq \vertexcoverParameter$.
	By the induction hypothesis, since $\awitness_i \in \dynamizationfunction{\vertexcoverCore_{\vertexcoverParameter}}{\treewidthvalue}(\sigmaabstractdecomposition_i)$, we have $\vertexcoverPredicate_{\vertexcoverParameter}[\treewidthvalue](\sigmaabstractdecomposition_i,\awitness_i) = 1$ for $i\in \{1,2\}$.
Therefore, for each $i\in \{1,2\}$, 
there exists a vertex cover $\vertexcover_i$ of minimum size in $\decompositiongraph{\sigmaabstractdecomposition_i}$ where $\topmap{\sigmaabstractdecomposition_i}(\vertexcoverset_i) = \vertexcover_i \cap \topmap{\sigmaabstractdecomposition_i}(\topbag{\sigmaabstractdecomposition_i})$.
Let $\vertexcover = \vertexcover_1 \cup \{ \relabelingfunction(x)\; : \; x\in \vertexcover_2\setminus\topmap{\sigmaabstractdecomposition_2}(\topbag{\sigmaabstractdecomposition_2}) \}$ where $\relabelingfunction = \relabelingfunction[n_1,n_2,Y]$ is the relabeling function from Section 
\ref{formalDefinitionITD} with $n_1 = |\vertexset{\decompositiongraph{\sigmaabstractdecomposition_1}}|$, $n_2 = |\vertexset{\decompositiongraph{\sigmaabstractdecomposition_2}}|$, and $Y = \topbag{\sigmaabstractdecomposition_2}$. 
Let $\zeta:\vertexset{\decompositiongraph{\sigmaabstractdecomposition_2}}\to \vertexset{\decompositiongraph{\abstractdecomposition}}$ be the map 
such that for each $\xvertex\in \vertexset{\decompositiongraph{\sigmaabstractdecomposition_2}}$
\[
\zeta(\xvertex)= 
\begin{cases}
	\topmap{\abstractdecomposition}\big(\topmap{\sigmaabstractdecomposition_2}^{-1}(\xvertex)\big) & \mbox{ if } \xvertex \in \topmap{\sigmaabstractdecomposition_2}(\topbag{\sigmaabstractdecomposition_2})\\
	\relabelingfunction(\xvertex) & \mbox{ if } \xvertex \notin \topmap{\sigmaabstractdecomposition_2}(\topbag{\sigmaabstractdecomposition_2})
\end{cases}
\]

By the fact that $\topbag{\abstractdecomposition} = \topbag{\sigmaabstractdecomposition_2}$ and that $\vertexcoverset_2\subseteq \topbag{\sigmaabstractdecomposition_2}$, we can infer that 
$\zeta(\topmap{\sigma_2}(\vertexcoverset_2)) = \topmap{\abstractdecomposition}(\vertexcoverset_2)$, $ \zeta(\topmap{\sigmaabstractdecomposition_2}(\topbag{\sigmaabstractdecomposition_2})) = \topmap{\abstractdecomposition}(\topbag{\abstractdecomposition})$, $\zeta(\vertexcover_2) = (\vertexcover\setminus \vertexcover_1)\cup \topmap{\abstractdecomposition}(\vertexcoverset_2)$, and $\vertexcover = \vertexcover_1 \cup \zeta(\vertexcover_2)$. Therefore, we have the following.
\begin{equation}\label{equation:vertex-cover-predicate-join-one}
    \begin{array}{cl}
        &\zeta(\topmap{\sigmaabstractdecomposition_2}(\vertexcoverset_2)) = \zeta(\vertexcover_2)\cap \zeta(\topmap{\sigmaabstractdecomposition_2}(\topbag{\sigmaabstractdecomposition_2}))
         \\
         \implies& \topmap{\abstractdecomposition}(\vertexcoverset_2) = \zeta(\vertexcover_2) \cap \topmap{\abstractdecomposition}(\topbag{\abstractdecomposition}).  \\
         
    \end{array}
\end{equation}
The set $\vertexcover$ is a vertex cover of $\decompositiongraph{\abstractdecomposition}$ since each edge is either in $\decompositiongraph{\sigmaabstractdecomposition_1}$ or $\decompositiongraph{\sigmaabstractdecomposition_2}$. 
If an edge is in $\decompositiongraph{\sigmaabstractdecomposition_1}$, then at least one of its endpoints is in $\vertexcover_1$ and if an edge is in $\decompositiongraph{\sigmaabstractdecomposition_2}$, then at least one of its endpoints is in $\zeta(\vertexcover_2)$. 
By the fact $\topmap{\sigmaabstractdecomposition_1}(\vertexcoverset_1) = \vertexcover_1 \cap \topmap{\sigmaabstractdecomposition_1}(\topbag{\sigmaabstractdecomposition_1})$, $\topbag{\abstractdecomposition} = \topbag{\sigmaabstractdecomposition_1} = \topbag{\sigmaabstractdecomposition_2}$, $\topmap{\abstractdecomposition} = \topmap{\sigmaabstractdecomposition_1}$, and Equation \ref{equation:vertex-cover-predicate-join-one}, we can infer the following. 
\begin{equation} \label{equation:vertex-cover-predicate-join-two}
    \begin{array}{l}
	    \topmap{\abstractdecomposition}(\vertexcoverset_1)\cup
	    \topmap{\abstractdecomposition}(\vertexcoverset_2)= 
	    \big(\vertexcover_1\cap\topmap{\abstractdecomposition}(\topbag{\abstractdecomposition})
	    \big) \cup \big(\zeta(\vertexcover_2)\cap
	    \topmap{\abstractdecomposition}(\topbag{\abstractdecomposition})\big)
	    \\ \implies
	    \topmap{\abstractdecomposition}(\vertexcoverset_1)\cup
	    \topmap{\abstractdecomposition}(\vertexcoverset_2)=
	    (\vertexcover_1\cup\zeta(\vertexcover_2))\cap\topmap{\abstractdecomposition}(\topbag{\abstractdecomposition})\\
	    \implies \topmap{\abstractdecomposition}(\vertexcoverset) =
	    \vertexcover\cap\topmap{\abstractdecomposition}(\topbag{\abstractdecomposition})
    \end{array}
\end{equation}

\newcommand{\vertexcoverM}{\mathcal{M}}
The vertex cover $\vertexcover$ satisfies Equation \ref{equation:vertex-cover-predicate-join-two} and
if we show that $\vertexcoversize$ is the minimum size of such vertex cover satisfying Equation \ref{equation:vertex-cover-predicate-join-two}, then we have proved the statement. 
Suppose $\vertexcover$ is not a vertex cover of minimum size and there exists a vertex cover $\vertexcoverM$ such that $|\vertexcoverM| < |\vertexcover|$ where $\topmap{\abstractdecomposition}(\vertexcoverset) = \vertexcoverM \cap \topmap{\abstractdecomposition}(\topbag{\abstractdecomposition})$. Let $\vertexcoverM_1 = \vertexcoverM \cap \vertexset{\decompositiongraph{\sigmaabstractdecomposition_1}}$  and $\vertexcoverM_2 = \zeta^{-1}(\vertexcoverM \cap \zeta(\vertexset{\decompositiongraph{\sigmaabstractdecomposition_2}}))$ be vertex covers of $\decompositiongraph{\sigmaabstractdecomposition_1}$ and $\decompositiongraph{\sigmaabstractdecomposition_2}$ satisfying the following equation.
\begin{equation}\label{equation:vertex-cover-predicate-join-three}
    \topmap{\sigmaabstractdecomposition_i}(\vertexcoverset_i) = \vertexcoverM_i \cap \topmap{\sigmaabstractdecomposition_i}(\topbag{\sigmaabstractdecomposition_i})
\end{equation} 
for $i\in \{1,2\}$. We have $\vertexcoverM_1\cup \zeta(\vertexcoverM_2)=\vertexcoverM$, and consequently, the following equation holds.
\begin{equation}\label{equation:vertex-cover-predicate-join-four}
    |\vertexcoverM| = |\vertexcoverM_1\cup \zeta(\vertexcoverM_2)| = |\vertexcoverM_1|+|\zeta(\vertexcoverM_2)|-|\vertexcoverM_1\cap\zeta(\vertexcoverM_2)|
\end{equation}
Also we have the following,
\begin{equation}\label{equation:vertex-cover-predicate-join-five}
  |\vertexcover| = |\vertexcover_1\cup \zeta(\vertexcover_2)| = |\vertexcover_1|+|\zeta(\vertexcover_2)|- |\vertexcover_1\cap\zeta(\vertexcover_2)|
\end{equation}

 We know $\vertexcover_i$ is a vertex cover of minimum size of $\decompositiongraph{\sigmaabstractdecomposition_i}$ satisfying Equation \ref{equation:vertex-cover-predicate-join-three}, and therefore, $|\vertexcover_i|\leq |\vertexcoverM_i|$, and also by $\topmap{\abstractdecomposition}(\vertexcoverset_1\cap\vertexcoverset_2) = \vertexcoverM_1\cap\zeta(\vertexcoverM_2) = \vertexcover_1\cap\zeta(\vertexcover_2)$, Equation \ref{equation:vertex-cover-predicate-join-four}, and Equation \ref{equation:vertex-cover-predicate-join-five} we can infer  $|\vertexcover| \leq |\vertexcoverM|$.
This contradicts the assumption that $|\vertexcoverM| < |\vertexcover|$, and therefore, $\vertexcover$ is minimum size. As a consequence, 
$\vertexcoverPredicate_{\vertexcoverParameter}[\treewidthvalue](\abstractdecomposition,\awitness) = 1$.

For the converse, suppose $\vertexcoverPredicate_{\vertexcoverParameter}[\treewidthvalue](\abstractdecomposition,\awitness)=1$. 
Then, there is a vertex cover $\vertexcover$ that satisfies the conditions 
of the predicate $\texttt{P-VertexCover}_r$. 
Let $\awitness_1=(\vertexcoverset_1,\vertexcoversize_1)$ be such that 
$\vertexcoverset_1 = \topmap{\abstractdecomposition}^{-1}(\vertexcover\cap\topmap{\sigmaabstractdecomposition_1}(\topbag{\sigmaabstractdecomposition_1}))$ and $\vertexcoversize_1 = |\vertexcover\cap \vertexset{\decompositiongraph{\sigmaabstractdecomposition_1}}|$. Analogously, 
let $\awitness_2=(\vertexcoverset_2,\vertexcoversize_2)$ be 
such that 
$\vertexcoverset_2 = \topmap{\abstractdecomposition}^{-1}(\vertexcover\cap \zeta(\topmap{\sigmaabstractdecomposition_2}(\topbag{\sigmaabstractdecomposition_2})))$,
 $\vertexcoversize_2 = |\vertexcover\cap \zeta(\vertexset{\decompositiongraph{\sigmaabstractdecomposition_2}})|$.
Let $\vertexcover_1 = \vertexcover\cap \vertexset{\decompositiongraph{\sigmaabstractdecomposition_1}}$, and $\vertexcover_2 = \zeta^{-1}(\vertexcover\cap \zeta(\vertexset{\decompositiongraph{\sigmaabstractdecomposition_2}}))$. Then, for each $i\in \{1,2\}$, 
we have that
\begin{equation}\label{equation:vertex-cover-predicate-six}
    \topmap{\sigmaabstractdecomposition_i}(\vertexcoverset_i) = \vertexcover_i \cap \topmap{\sigmaabstractdecomposition_i}(\topbag{\sigmaabstractdecomposition_i}). 
\end{equation}
Therefore, $\vertexcover_1$ is a minimal-size vertex cover in the graph $\decompositiongraph{\sigmaabstractdecomposition_1}$, and 
$\vertexcover_2$ is a minimal-size vertex cover in the graph $\decompositiongraph{\sigmaabstractdecomposition_2}$, since the non-minimality of either set would contradict the assumption that $\vertexcover$ has minimum size. Therefore, $\vertexcoverPredicate_{\vertexcoverParameter}[\treewidthvalue](\sigmaabstractdecomposition_i,\awitness_i) = 1 $. By the induction hypothesis, $\awitness_i \in \dynamizationfunction{\vertexcoverCore_{\vertexcoverParameter}}{\treewidthvalue}(\sigmaabstractdecomposition_i)$ for $i\in\{1,2\}$. By Definition \ref{definition:vertex-cover}.\ref{definition:vertex-cover-join}, we have that $\vertexcoverCore_{\vertexcoverParameter}[\treewidthvalue].\jointype(\awitness_1,\awitness_2) = \{\awitness\}$, and therefore, $\awitness \in \dynamizationfunction{\vertexcoverCore_{\vertexcoverParameter}}{\treewidthvalue}(\abstractdecomposition)$.
$\square$

\end{enumerate}

\section{Proof of Theorem \ref{theorem:Estimates}}
\label{section:ProofTheoremEstimates}

In this section, we prove Theorem \ref{theorem:Estimates}. More specifically,
for each graph property $\graphproperty$ listed in Theorem \ref{theorem:Estimates}, 
we provide an upper bound for the bit-length $\beta_{\dpcore}(\treewidthvalue)$ of a
suitable $\instructivetreedecompositionclass$-coherent DP-core $\dpcore$ with 
$\dpcoregraphproperty{\dpcore,\instructivetreedecompositionclass} = \graphproperty$. The multiplicity of
such a DP-core $\dpcore$ is trivially upper bounded by $2^{O(\beta_{\dpcore}(\treewidthvalue))}$, 
the state complexity of $\dpcore$ is trivially upper bounded by $2^{O(\beta_{\dpcore}(\treewidthvalue))}$, 
and the deterministic state complexity of $\dpcore$ is trivially upper bounded by $2^{2^{O(\beta_{\dpcore}(\treewidthvalue))}}$. 
In some cases, we show that  these two last upper bounds can be improved using additional arguments.  

The process of specifying a coherent DP-core $\dpcore$ deciding a given graph property $\graphproperty$ can be split into four 
main steps, which should hold for each $\treewidthvalue\in \N$. 

\begin{enumerate}
\item\label{stepOne} The specification of a suitable set $\dpcore[\treewidthvalue].\allwitnesses$ of local 
witnesses, and of a Boolean function $\dpcore[\treewidthvalue].\finalwitnessgeneric$ determining which are the final local witnesses in 
$\dpcore[\treewidthvalue].\allwitnesses$. 
\item\label{stepTwo} The specification of a predicate 
$\texttt{P}[\treewidthvalue] \subseteq \allabstractdecompositionstreewidth{\treewidthvalue} \times \dpcore[\treewidthvalue].\allwitnesses$ such that 
for each  $\abstractdecomposition\in \allabstractdecompositionstreewidth{\treewidthvalue}$, and each {\em final local witness} $\awitness$,
$(\abstractdecomposition,\awitness) \in \texttt{P}[\treewidthvalue]$ if and only if the graph $\decompositiongraph{\abstractdecomposition}$
belongs to $\graphproperty$.
\item\label{stepThree} The specification of the functions belonging to $\dpcore[\treewidthvalue]$. 
\item\label{stepFour} An inductive proof that for each pair $(\abstractdecomposition,\awitness) \in \allabstractdecompositionstreewidth{\treewidthvalue}\times \dpcore[\treewidthvalue].\allwitnesses$, 
$(\abstractdecomposition,\awitness)\in \dynamizationfunction{\dpcore}{\treewidthvalue}(\abstractdecomposition)$ if and only if $(\abstractdecomposition,\awitness)\in \texttt{P}[\treewidthvalue]$. 
This last step together with Step 2 guarantees that the DP-core is coherent and that $\dpcoregraphproperty{\dpcore[\treewidthvalue]} = \graphproperty$. 
\end{enumerate}

The process described above was exemplified in full detail with respect to the DP-core $\vertexcoverCore_{r}$ defined in
Appendix \ref{subsection:DPCoreVertexCover}. In this section, we are 
only interested in providing an upper bound on the bitlength of the DP-cores listed in Theorem \ref{theorem:Estimates}. Therefore, we will only carry out
in details Steps \ref{stepOne} and \ref{stepTwo}.
Our choices of local witnesses is based on simplicity, instead of efficiency, and our goal primarily to justify the asymptotic upper bounds 
listed in Theorem \ref{theorem:Estimates}.  

\subsection{$\texttt{C-Simple}$}
	 The graph property of the DP-core $\texttt{C-Simple}$ is the set $\texttt{Simple}$ of all simple graphs. 
		For each $\treewidthvalue\in \N$, a local witness for $\texttt{C-Simple}[\treewidthvalue]$
		is a subset $\awitness\subseteq \binom{[\treewidthvalue+1]}{2}$. All local witnesses are final.
		The transitions of $\texttt{C-Simple}[\treewidthvalue]$ are defined in such a way
		that for each $\treewidthvalue$-instructive tree decomposition $\abstractdecomposition$ and
		each local witness $\awitness$, $\awitness\in \dynamizationfunction{\texttt{C-Simple}}{\treewidthvalue}(\abstractdecomposition)$
		if and only if the following predicate is satisfied: 
	\begin{itemize}
		\item $\texttt{P-Simple}[\treewidthvalue](\abstractdecomposition,\awitness) \equiv$ $\decompositiongraph{\abstractdecomposition}$ is a simple graph, and for each two distinct elements $\vertexone$ and
		$\vertextwo$ in $[\treewidthvalue+1]$, $\{\vertexone,\vertextwo\}\in \awitness$ if and only if there is a unique edge
		$\anedge\in \edgeset{\decompositiongraph{\abstractdecomposition}}$ such that $\edgeendpoints{\decompositiongraph{\abstractdecomposition}}(\anedge) = 
		\{\topmap{\abstractdecomposition}(\vertexone),\topmap{\abstractdecomposition}(\vertextwo)\}$.
	\end{itemize}
		Each local witness $\awitness$ may be represented as a 
		Boolean vector in $\{0,1\}^{\binom{\treewidthvalue+1}{2}}$ which has one coordinate 
		for each pair $\{\vertexone,\vertextwo\}\in \awitness$. The DP-core 
		$\texttt{C-Simple}[\treewidthvalue]$ can be defined 
		in such a way that it is deterministic, in the 
		sense that the set obtained from each local witness $\awitness$ upon the application of 
		each transition is a singleton. 
		Therefore, the multiplicity of $\texttt{C-Simple}[\treewidthvalue]$ is $1$. This implies that 
		the deterministic state complexity of $\texttt{C-Simple}[\treewidthvalue]$ is at most $2^{\binom{\treewidthvalue+1}{2}}$.

\subsection{$\texttt{C-MaxDeg}_{\geq}(d)$}
 The graph property of the DP-core $\texttt{C-MaxDeg}_{\geq}(d)$ is the set $\texttt{MaxDeg}_{\geq}(d)$ of 
		all graphs with maximum degree at least $d$. For each $\treewidthvalue\in \N$, a local witness for 
		$\texttt{C-MaxDeg}_{\geq}(d)[\treewidthvalue]$ is a 
		pair $(x,y) \in \{0,1\} \times \{0,1,\dots,d+1\}^{\treewidthvalue+1}$.
		Such local witness is final if and only if $x=1$ or there is some $\vertexone\in [\treewidthvalue+1]$ such that
		$y_{\vertexone} \geq d$. The transitions of $\texttt{C-MaxDeg}_{\geq}(d)[\treewidthvalue]$ 
		are defined in such a way that for each $\treewidthvalue$-instructive tree decomposition 
		$\abstractdecomposition$ and each local witness $(x,y)$, 
		$(x,y)\in \dynamizationfunction{\texttt{C-MaxDeg}_{\geq}(d)}{\treewidthvalue}(\abstractdecomposition)$
		if and only if the following predicate is satisfied: 
	\begin{itemize}
		\item $\texttt{P-MaxDeg}_{\geq}(d)[\treewidthvalue](\abstractdecomposition,(x,y)) \equiv $ for each $\vertexone\in [\treewidthvalue+1]$, the vertex $\topmap{\abstractdecomposition}(\vertexone)$
		in the graph $\decompositiongraph{\abstractdecomposition}$ has degree 
		$y_{\vertexone}$ if $y_{\vertexone}\leq d$ and degree at least $d+1$ if
		$y_{\vertexone}=d+1$; $x = 1$ if and only if there is some vertex of degree at least $d$ in 
		$\vertexset{\decompositiongraph{\abstractdecomposition}}\backslash
		\topmap{\abstractdecomposition}(\topbag{\abstractdecomposition})$. 
	\end{itemize}
		Each local witness $(x,y)$ can be represented as a binary string
		containing $1+(\treewidthvalue+1)\cdot\lceil \log (d+1)  \rceil$ bits. The DP-core 
		$\texttt{C-MaxDeg}_{\geq}(d)[\treewidthvalue]$ can be defined in such a way that it is deterministic. 
		Therefore, the multiplicity of $\texttt{C-MaxDeg}_{\geq}(d)[\treewidthvalue]$ is $1$. This implies that 
		the deterministic state complexity of $\texttt{C-MaxDeg}_{\geq}(d)[\treewidthvalue]$ is upper bounded by $2^{O(\treewidthvalue\log d)}$. 

\subsection{$\texttt{C-MinDeg}_{\leq}(d)$}

The graph property of the DP-core $\texttt{C-MinDeg}_{\leq}(d)$ is the set $\texttt{MinDeg}_{\leq}(d)$ of all graphs
		with minimum degree at most $d$. For each $\treewidthvalue\in \N$, a local witness for $\texttt{C-MinDeg}_{\leq}(d)[\treewidthvalue]$ is a 
		pair $(x,y) \in \{0,1\} \times \{0,1,\dots,d+1\}^{\treewidthvalue+1}$. Such local witness is final 
		if and only if $x=1$ or $y_{\vertexone}\leq d$ for some $\vertexone\in [\treewidthvalue+1]$.	
		The transitions of $\texttt{C-MinDeg}_{\leq}(d)[\treewidthvalue]$ are defined in such a way
		that for each $\treewidthvalue$-instructive tree decomposition $\abstractdecomposition$ and each local
		witness $y$, $(x,y)\in \dynamizationfunction{\texttt{C-MinDeg}_{\leq}(d)}{\treewidthvalue}(\abstractdecomposition)$
		if and only if the following predicate is satisfied:
 	\begin{itemize}
		\item $\texttt{P-MinDeg}_{\leq}(d)[\treewidthvalue](\abstractdecomposition,(x,y)) \equiv $ 
		 for each $\vertexone\in [\treewidthvalue+1]$, the vertex $\topmap{\abstractdecomposition}(\vertexone)$
		in the graph $\decompositiongraph{\abstractdecomposition}$ has degree 
		$y_{\vertexone}$ if $y_{\vertexone}\leq d$ and degree at least $d+1$ if
		$y_{\vertexone}=d+1$; $x = 1$ if and only if there is some vertex of degree at most $d$ in 
		$\vertexset{\decompositiongraph{\abstractdecomposition}}\backslash
		\topmap{\abstractdecomposition}(\topbag{\abstractdecomposition})$. 
	\end{itemize}
	Each local witness $(x,y)$ can be represented as a binary string
		containing $1+(\treewidthvalue+1)\cdot\lceil \log (d+1)  \rceil$ bits. The DP-core 
		$\texttt{C-MinDeg}_{\leq}(d)[\treewidthvalue]$ can be defined in such a way that it is deterministic.
		Therefore, the multiplicity of $\texttt{C-MinDeg}_{\leq}(d)[\treewidthvalue]$ is $1$. This implies that 
		the deterministic state complexity of $\texttt{C-MinDeg}_{\leq}(d)[\treewidthvalue]$ is upper bounded by $2^{O(\treewidthvalue\log d)}$. 

\subsection{$\texttt{C-Colorable(c)}$} 

 The graph property of the DP-core $\texttt{C-Colorable(c)}$ is the set $\texttt{Colorable(c)}$ of all graphs that
	    are $c$-colorable. For each $\treewidthvalue\in \N$, a local witness for $\texttt{C-Colorable(c)}[\treewidthvalue]$
	    is a vector in $\awitness\in \{0,1,\dots,c\}^{\treewidthvalue+1}$. All local witnesses are final. 
	    The transitions of  $\texttt{C-Colorable(c)}[\treewidthvalue]$ are defined in such a way 
	    that for each $\treewidthvalue$-instructive tree decomposition $\abstractdecomposition$, and each local witness
	    $\awitness$, $\awitness\in \dynamizationfunction{\texttt{C-Colorable(c)}}{\treewidthvalue}(\abstractdecomposition)$ if and only if
	    the following predicate is satisfied: 
	\begin{itemize}
	\item $\texttt{P-Colorable}(c)[\treewidthvalue](\abstractdecomposition,\awitness) \equiv $  there is  a proper $c$-coloring
	$\vertexcoloring:\vertexset{\decompositiongraph{\abstractdecomposition}}\rightarrow [c]$ where for each $\vertexone\in \topbag{\abstractdecomposition}$, 
			$\awitness_{\vertexone}= \vertexcoloring(\topmap{\abstractdecomposition}(\vertexone))$ and for each $\vertexone \in [\treewidthvalue+1]\backslash \topbag{\abstractdecomposition}$, 
			$\awitness_{\vertexone}=0$. 
	\end{itemize}
	Each local witness can be represented using $\lceil \log(c + 1) \rceil \cdot (\treewidthvalue+1)$ bits.
	It is worth noting that every graph of treewidth at most $\treewidthvalue$ is $(\treewidthvalue+1)$-colorable, and therefore, for $c\geq \treewidthvalue+1$, 
	we can define $\texttt{C-Colorable}(c)[\treewidthvalue]$ as the trivial DP-core which accepts all $\treewidthvalue$-instructive tree decompositions.
	This DP-core has a unique local witness, and this unique local witness is final. 

\subsection{$\texttt{C-Conn}$} 
\label{subsection:ConnectivityCore}

The graph property of the DP-core $\texttt{C-Conn}$ is the set $\texttt{Conn}$ of all connected graphs. 
For each $\treewidthvalue\in \N$, a local witness for $\texttt{C-Conn}[\treewidthvalue]$ is a pair $(\qconn,\Pconn)$
	  where $\qconn\in \{0,1,2,3\}$, and $\Pconn$ is a partition of some subset of $[\treewidthvalue+1]$. Such a local witness
	  $(\qconn,\Pconn)$ is final if $\qconn\neq 3$ and $\Pconn$ has at most one cell (the empty partition with no cell is a legal partition of the empty set).
	  The transitions of $\texttt{C-Conn}[\treewidthvalue]$ are defined in such a way that for each 
	 $\treewidthvalue$-instructive tree decomposition $\abstractdecomposition$, and each local witness
	  $(\qconn,\Pconn)$, $(\qconn,\Pconn)\in \dynamizationfunction{\texttt{C-Conn}}{\treewidthvalue}(\abstractdecomposition)$ if and only if the following predicate is satisfied: 
	\begin{itemize}
		\item $\texttt{P-Conn}[\treewidthvalue](\abstractdecomposition,(\qconn,\Pconn)) \equiv$
		$U(P)=\topbag{\abstractdecomposition}$; for each two labels $\vertexone,\vertextwo\in \topbag{\abstractdecomposition}$, $\vertexone$ and $\vertextwo$ are in
		the same cell of $\Pconn$ if and only if $\topmap{\abstractdecomposition}(\vertexone)$
		and $\topmap{\abstractdecomposition}(\vertextwo)$ belong to the same connected component in the
		graph $\decompositiongraph{\abstractdecomposition}$; furthermore,
 		$$
		\qconn = \left\{ 
		\begin{array}{ll}
		0 & \mbox{if $\decompositiongraph{\abstractdecomposition}$ is the empty graph;} \\
		1 & \mbox{if $P\neq \emptyset$ and every vertex in $\decompositiongraph{\abstractdecomposition}$
		is reachable from $\topmap{\abstractdecomposition}(U(P))$;} \\ 

		2 & \mbox{if $P=\emptyset$, $\decompositiongraph{\abstractdecomposition}$ is connected and not the empty graph;}\\
		3 & \mbox{if $P\neq \emptyset$ and some vertex in $\decompositiongraph{\abstractdecomposition}$ is not 
				reachable from $\topmap{\abstractdecomposition}(U(P))$,} \\ 
		  & \mbox{or $P=\emptyset$ and $\decompositiongraph{\abstractdecomposition}$ is disconnected.}
		\end{array}
		\right.
		$$
	\end{itemize}
	Each local witness $(\qconn,\Pconn)$ can be represented using
	$2+  (\treewidthvalue+1)\cdot \lceil \log (\treewidthvalue+2) \rceil = O(\treewidthvalue\log\treewidthvalue)$ bits. 
	Additionally, $\texttt{C-Conn}$ can be defined in such a way that it has multiplicity $1$. Therefore, its 
	deterministic state complexity is upper bounded by $2^{O(\treewidthvalue\log \treewidthvalue)}$. 

\subsection{$\texttt{C-VConn}_{\leq}(c)$}
The graph property of the DP-core $\texttt{C-VConn}_{\leq}(c)$ is the set of all graphs with vertex-connectivity at 
most $c$. A local witness is a triple $(r,\qconn,\Pconn)$ where $r\in \{0,1,\dots,c\}$, $\qconn\in \{0,1,2,3\}$, and $\Pconn$ is a partition of some subset of 
$[\treewidthvalue+1]\backslash R$. A witness is final if $\qconn=3$ or $\Pconn$ has more than one cell (this means that after removing $r\leq c$ vertices the graph gets disconnected). 
The transitions of $\texttt{C-VCon}_{\leq}(c)[\treewidthvalue]$ are defined in such a way that 
for each $\treewidthvalue$-instructive tree decomposition $\abstractdecomposition$, and each local witness
$(r,\qconn,\Pconn)$, $(r,\qconn,\Pconn)\in \dynamizationfunction{\texttt{C-VCon}_{\leq}(c)}{\treewidthvalue}(\abstractdecomposition)$ if and only if
    the following predicate is satisfied: 
	\begin{itemize}
		\item $\texttt{P-VCon}_{\leq}(c)[\treewidthvalue](\abstractdecomposition,\awitness) \equiv $ there is a subset of vertices 
		$X\subseteq \vertexset{\decompositiongraph{\abstractdecomposition}}$ of size $r$ 
		such that:
		$U(P)\subseteq \topbag{\abstractdecomposition}$ and for each $u\in \topbag{\abstractdecomposition}$, $u\in U(P)$ if and only if $\topmap{\abstractdecomposition}(u)\notin X$;
		for each two labels $\vertexone,\vertextwo\in \topbag{\abstractdecomposition}$, $\vertexone$ and $\vertextwo$ are in
		the same cell of $\Pconn$ if and only if $\topmap{\abstractdecomposition}(\vertexone)$
		and $\topmap{\abstractdecomposition}(\vertextwo)$ belong to the same connected component in the
		graph $\decompositiongraph{\abstractdecomposition}\backslash X$; furthermore,
 		$$
		\qconn = \left\{ 
		\begin{array}{ll}
		0 & \mbox{if $\decompositiongraph{\abstractdecomposition}\backslash X$ is the empty graph;} \\
		1 & \mbox{if $P\neq \emptyset$ and every vertex in $\decompositiongraph{\abstractdecomposition}\backslash X$
		is reachable from $\topmap{\abstractdecomposition}(U(P))$;} \\ 
		2 & \mbox{if $P=\emptyset$, $\decompositiongraph{\abstractdecomposition}\backslash X$ is connected and not the empty graph;}\\
		3 & \mbox{if $P\neq \emptyset$ and some vertex in $\decompositiongraph{\abstractdecomposition}\backslash X$ is not 
				reachable from $\topmap{\abstractdecomposition}(U(P))$,} \\ 
		  & \mbox{or $P=\emptyset$ and $\decompositiongraph{\abstractdecomposition}\backslash X$ is disconnected.}
		\end{array}
		\right.
		$$
	\end{itemize}
		Each local witness can be represented using $2 +\lceil \log c \rceil +  (\treewidthvalue+1)\cdot (1 + \lceil\log (\treewidthvalue+2) \rceil) = O(\log c + \treewidthvalue\log\treewidthvalue)$ bits.
		It is worth noting that every graph of treewidth at most $\treewidthvalue$ has connectivity at most $\treewidthvalue+1$, and therefore, for $c\geq \treewidthvalue+1$, 
		we can define $\texttt{C-VCon}_{\leq}(c)[\treewidthvalue]$ as the trivial DP-core which accepts all $\treewidthvalue$-instructive tree decompositions. This DP-core has a unique local witness, and this 
		unique local witness is final. 

\subsection{$\texttt{C-EConn}_{\leq}(c)$} 
The graph property of the DP-core $\texttt{C-EConn}_{\leq}(c)$ is the set of all graphs with edge-connectivity at 
most $c$. A local witness is a tuple $(r,\qconn,\Pconn)$ where $r\in \{0,1,\dots,c\}$, $\qconn\in \{0,1,2,3\}$, and $\Pconn$ is a partition of some subset of 
$[\treewidthvalue+1]$. A witness is final if $\qconn=3$ or $\Pconn$ has more than one cell (this means that after removing $r\leq c$ edges, the graph gets disconnected). 
The transitions of $\texttt{C-EConn}_{\leq}(c)[\treewidthvalue]$ are defined in such a way that 
for each $\treewidthvalue$-instructive tree decomposition $\abstractdecomposition$, and each local witness
$(r,\qconn,\Pconn)$, $(r,\qconn,\Pconn)\in \dynamizationfunction{\texttt{C-EConn}_{\leq}(c)}{\treewidthvalue}(\abstractdecomposition)$ if and only if
the following predicate is satisfied: 
	\begin{itemize}
		\item $\texttt{P-EConn}_{\leq}(c)[\treewidthvalue](\abstractdecomposition,\awitness) \equiv $
		there is a subset of edges $Y\subseteq \vertexset{\decompositiongraph{\abstractdecomposition}}$
		of size $r$ such that:
		$U(P)=\topbag{\abstractdecomposition}$; for each two labels $\vertexone,\vertextwo\in \topbag{\abstractdecomposition}$, $\vertexone$ and $\vertextwo$ are in
		the same cell of $\Pconn$ if and only if $\topmap{\abstractdecomposition}(\vertexone)$
		and $\topmap{\abstractdecomposition}(\vertextwo)$ belong to the same connected component in the
		graph $\decompositiongraph{\abstractdecomposition}\backslash Y$; furthermore,
 		$$
		\qconn = \left\{ 
		\begin{array}{ll}
		0 & \mbox{if $\decompositiongraph{\abstractdecomposition}\backslash Y$ is the empty graph;} \\
		1 & \mbox{if $P\neq \emptyset$ and every vertex in $\decompositiongraph{\abstractdecomposition}\backslash Y$
		is reachable from $\topmap{\abstractdecomposition}(U(P))$;} \\ 
		2 & \mbox{if $P=\emptyset$, $\decompositiongraph{\abstractdecomposition}\backslash Y$ is connected and not the empty graph;}\\
		3 & \mbox{if $P\neq \emptyset$ and some vertex in $\decompositiongraph{\abstractdecomposition}\backslash Y$ is not 
				reachable from $\topmap{\abstractdecomposition}(U(P))$,} \\ 
		  & \mbox{or $P=\emptyset$ and $\decompositiongraph{\abstractdecomposition}\backslash Y$ is disconnected.}
		\end{array}
		\right.
		$$
	\end{itemize}
		Each local witness can be represented using $2 +\lceil \log c \rceil +  (\treewidthvalue+1)\cdot \lceil \log (\treewidthvalue+2) \rceil = O(\log c + \treewidthvalue\log\treewidthvalue)$ bits.

\subsection{$\texttt{C-Hamiltonian}$}

The graph property of the DP-core $\texttt{C-Hamiltonian}$ is the set $\texttt{Hamiltonian}$ of all graphs that 
are Hamiltonian. This DP-core is defined by a straightforward adaptation of the standard algorithm for 
	testing Hamiltonicity parameterized by treewidth. See for instance \cite{ziobro2019finding}. 

For each $\treewidthvalue\in \N$, a local witness for $\texttt{C-Hamiltonian(c)}[\treewidthvalue]$
is a pair $(\beta,M)$ where $\beta:S\rightarrow \{0,1,2\}$ is a function whose domain $S$ is a subset of $[\treewidthvalue+1]$
and $M\subseteq \powersetchoosek{\beta^{-1}(1)}{2}$ is a matching that relates pairs of labels in $S$ that are sent 
to the value $1$. The transitions of $\texttt{C-Hamiltonian(c)}[\treewidthvalue]$ are defined in such a way 
that for each $\treewidthvalue$-instructive tree decomposition $\abstractdecomposition$, and each local witness
$\awitness$, $\awitness\in \dynamizationfunction{\texttt{C-Hamiltonian(c)}}{\treewidthvalue}(\abstractdecomposition)$
if and only if the following predicate is satisfied: 
\begin{itemize}
\item $\texttt{P-Hamiltonian}(c)[\treewidthvalue](\abstractdecomposition,\awitness) \equiv$
either $\decompositiongraph{\abstractdecomposition}$ is Hamiltonian and $\beta^{-1}(0) = \beta^{-1}(1) = \emptyset$, 
or there is a partition $\mathcal{P}$ of $\vertexset{\decompositiongraph{\abstractdecomposition}}$ into 
vertex-disjoint paths such that for each $\vertexone\in \topbag{\abstractdecomposition}$,
$\topmap{\abstractdecomposition}(\vertexone)$ has degree $\beta(\vertexone)$ in some path of $\mathcal{P}$.  
\end{itemize}
Each local witness can be represented using $\lceil \log\treewidthvalue \rceil \cdot (\treewidthvalue+1) = O(\treewidthvalue\log\treewidthvalue)$
bits. This yields a multiplicity of at most $2^{O(\treewidthvalue\cdot \log \treewidthvalue)}$, and consequently a deterministic state complexity of 
$2^{2^{O(\treewidthvalue\cdot \log \treewidthvalue)}}$. 
Dynamic programming algorithms that can be translated into DP-cores with better 
multiplicity have been devised in \cite{cygan2015parameterized,bodlaender2015deterministic}. For example, using a clever $\dpcore[\treewidthvalue].\cleaningfunctioncore$
function, that applies the rank-based approach introduced in \cite{bodlaender2015deterministic} to eliminate redundancies, 
one can guarantee that the number
of local witnesses in a useful witness set (i.e. the multiplicity of the DP-core) is always bounded by $2^{O(\treewidthvalue)}$. 
In other words, this clean function decreases the multiplicity of the DP-core without affecting the existence of 
a solution. As a consequence, the deterministic state complexity of the DP-core is upper bounded by
$\binom{2^{O(k\log k)}}{2^{O(k)}} = 2^{2^{O(k)}}$. 
A nice discussion about the rank based approach and other approaches to solve the Hamiltonian cycle problem on graphs 
of bounded treewidth is also present in \cite{ziobro2019finding}.

\subsection{$\texttt{C-NZFlow}(\Z_\flowValue)$}

We let $\Z_{\flowValue}=\{0,\dots,m-1\}$ be the set of integers modulo $\flowValue$. Let  $\agraph$ be a graph and $(\tailmap,\headmap)$
be an orientation of $\agraph$. In other words, $\tailmap:\edgeset{G}\to \vertexset{G}$ and $\headmap:\edgeset{G}\to \vertexset{G}$ are maps 
that specify the tail  $\tailmap(\anedge)$ and the head $\headmap(\anedge)$ of each edge $\anedge\in \edgeset{\agraph}$ ($\tailmap(\anedge)\neq \headmap(\anedge)$). For each $v\in \vertexset{\agraph}$, we let $\delta^-(v)=\{e | \; \tailmap(v)=e \}$ be the set of edges whose tail is $v$, and $\delta^+(v)=\{e | \; \headmap(v)=e \}$ be the set of edges whose head is $v$. We say that 
$v$ satisfies the flow equation if the following condition is satisfied. 

\begin{equation}
\label{equation:FlowEquation}
\sum_{e\in \delta^+(v)} \phi(e) = \sum_{e\in \delta^-(v)} \phi (e)
\end{equation}

\begin{definition}[Nowhere-Zero $\Z_{\flowValue}$-Flow]\label{definition:nowhereZeroFlow}
Let $\agraph$ be a graph and $(\tailmap,\headmap)$ be an orientation of $\agraph$. A nowhere-zero $\Z_{\flowValue}$-flow in $(\agraph,\tailmap,\headmap)$ is a function $\phi:\edgeset{\agraph}\to \Z_{\flowValue}$ 
satisfying the following conditions:
\begin{enumerate}
    \item \label{condition:flow equality} for each $v\in V$, $v$ satisfies the 
    flow equation, and
    \item for each $e\in E$, $\phi(e) \neq 0$.
\end{enumerate}

We say that $\phi:\edgeset{\agraph}\to \Z_{\flowValue}$ is a nowhere-zero $\Z_{\flowValue}$-flow in $\agraph$ if there is an orientation $(\tailmap,\headmap)$ such that $\phi$ is a
$\Z_{\flowValue}$-flow in $(\agraph,\tailmap,\headmap)$. 
\end{definition}

For each $\flowValue\in \N$, we let $\texttt{NZFlow}(\Z_\flowValue)$ be the graph property consisting of all graphs that admit a
nowhere-zero $\Z_\flowValue$-flow. For each $\treewidthvalue\in \N$, a local witness for $\texttt{C-NZFlow}(\Z_\flowValue)[\treewidthvalue]$
is a set of triples of the form $(\vertexone,\vertextwo,\flowAssignedValue)$ where $\vertexone$ and $\vertextwo$ belong to
$[\treewidthvalue+1]$ and $\flowAssignedValue\in \Z_{\flowValue}$. Such a local witness is final, meaning that 
$\texttt{C-NZFlow}(\Z_\flowValue)[\treewidthvalue].\finalwitnessgenericcore(\vertexone,\vertextwo,\flowAssignedValue) = 1$ if 
and only if $\flow(\flowWitness,\vertexone)$ is true for each $\vertexone\in \labels(\flowWitness)$. 
The transitions of $\texttt{C-MaxDeg}_{\geq}(d)[\treewidthvalue]$ 
are defined in such a way that for each $\treewidthvalue$-instructive tree decomposition 
$\abstractdecomposition$ and each local witness $\awitness$,
$\awitness \in \dynamizationfunction{\texttt{C-NZFlow}(\Z_m)}{\treewidthvalue}(\abstractdecomposition)$
if and only if the following predicate is satisfied: 
\begin{itemize}
\item $\texttt{P-NZFlow}(\Z_m)[\treewidthvalue](\abstractdecomposition,\awitness) \equiv $ 
there is an orientation $(\tailmap,\headmap)$ of the graph $\decompositiongraph{\abstractdecomposition}$,
		and a function $\phi:\edgeset{\decompositiongraph{\abstractdecomposition}}\to \Z_{\flowValue}$
such that the following conditions are satisfied. 
\begin{enumerate}
\item For each $\vertexone$ and $\vertextwo$ in 
$\topbag{\abstractdecomposition}$, and each $\flowAssignedValue\in \Z_{\flowValue}$, $(\vertexone,\vertextwo,\flowAssignedValue)\in \flowWitness$ if and only if $\sum_{\anedge} \phi(\anedge) = \flowAssignedValue$, where 
$\anedge$ ranges over all edges $\anedge\in \edgeset{\decompositiongraph{\abstractdecomposition}}$ with
$\tailmap(\anedge) = \topmap{\abstractdecomposition}(\vertexone)$, 
and
$\headmap(\anedge) =\topmap{\abstractdecomposition}(\vertextwo)$.  
\item For each $\vertexone\in \topbag{\abstractdecomposition}$, 
$\flow(\flowWitness,\vertexone) = \true$ if and only if
$\topmap{\abstractdecomposition}(\vertexone)$ satisfies the flow equation (Equation \ref{equation:FlowEquation}). In particular, 
if $\awitness$ is a final local witness, then for each $\vertexone \in \topbag{\abstractdecomposition}$, the vertex $\topmap{\abstractdecomposition}(\vertexone)$ satisfies the flow equation. 
\item Each vertex $\xvertex \in \vertexset{\decompositiongraph{\abstractdecomposition}}\setminus \image(\topmap{\abstractdecomposition})$, satisfies the flow equation
(Equation \ref{equation:FlowEquation}).
\end{enumerate}
\end{itemize}

Since a  local witness for $\texttt{C-NZFlow}(\Z_{\flowValue})[\treewidthvalue]$ is a set with $O(\treewidthvalue^2)$ triples, such a
witness can be represented by a binary vector of length $O(\treewidthvalue^2\cdot \log \flowValue)$, where $\log \flowValue$
bits are used to represent each flow value.

\subsection{$\texttt{C-Minor(H)}$} 

Parameterized algorithms for determining whether a pattern graph $H$ is a minor of a host graph $G$ parameterized by the branchwidth of 
the host graph have been devised in \cite{hicks2004branch},
and subsequently improved in \cite{adler2011faster}. Both algorithms operate with branch decompositions instead of tree decompositions.
The later algorithm works in time $O(2^{(2k+1)\cdot \log k} \cdot |\vertexset{H}|^{2k} \cdot 2^{2{|\vertexset{H}|}^2}\cdot |\vertexset{G}|)$,
where $\treewidthvalue$ is the width of the input branch decomposition. Partial solutions in both algorithms are combinatorial objects called 
{\em rooted packings}. 

In this section, we describe the structure of local witnesses of a DP-core $\texttt{C-Minor(H)}$ that solves the minor containment problem by analyzing 
$\treewidthvalue$-instructive tree decompositions, instead of branching decompositions. In our setting, partial solutions are objects 
called {\em quasimodels}. A local witness may be regarded as an encoding of the restriction of a quasimodel to the active vertices of a given $\treewidthvalue$-instructive
tree decomposition. We note that our local witnesses have bitlength $O(\treewidthvalue\cdot \log \treewidthvalue + |\vertexset{H}| + |\edgeset{H}|)$, and 
as a consequence, our DP-core decides whether a graph $\agraph$ of treewidth at most $\treewidthvalue$ has an $\bgraph$-minor 
in time $2^{O(\treewidthvalue\cdot \log \treewidthvalue + |\vertexset{H}| + |\edgeset{H}|)} \cdot (|\vertexset{G}| + |\edgeset{G}|)$. 


We
will need the following variant of the predicate
$\texttt{P-Conn}[\treewidthvalue](\abstractdecomposition,(\qconn,P))$ defined in Subsection \ref{subsection:ConnectivityCore}. This variant
is parameterized by a subset $X\subseteq \N$, where $X$ is meant to be a subset of $\vertexset{\decompositiongraph{\abstractdecomposition}}$. Note that this variant is obtained essentially, by
replacing $\decompositiongraph{\abstractdecomposition}$ with $\decompositiongraph{\abstractdecomposition}[X]$ and by requiring that $P$ is 
a partition of some subset of $\topbag{\abstractdecomposition}$, instead of a partition of $\topbag{\abstractdecomposition}$. 

	\begin{itemize}
		\item $\texttt{P-Conn}^{X}[\treewidthvalue](\abstractdecomposition,(\qconn,\Pconn)) \equiv$ 
		$X\subseteq \vertexset{\decompositiongraph{\abstractdecomposition}}$;
		$U(P)\subseteq \topbag{\abstractdecomposition}$ and for each $u\in \topbag{\abstractdecomposition}$, $u\in U(P)$ if and only if $\topmap{\abstractdecomposition}(u)\in X$;
		for each two labels $\vertexone,\vertextwo\in \topbag{\abstractdecomposition}$, $\vertexone$ and $\vertextwo$ are in
		the same cell of $\Pconn$ if and only if $\topmap{\abstractdecomposition}(\vertexone)$
		and $\topmap{\abstractdecomposition}(\vertextwo)$ belong to the same connected component in the
		graph $\decompositiongraph{\abstractdecomposition}[X]$; furthermore,
 		$$
		\qconn = \left\{ 
		\begin{array}{ll}
		0 & \mbox{if $\decompositiongraph{\abstractdecomposition}[X]$ is the empty graph;} \\
		1 & \mbox{if $P\neq \emptyset$ and every vertex in $\decompositiongraph{\abstractdecomposition}[X]$
		is reachable from $\topmap{\abstractdecomposition}(U(P))$;} \\ 
		2 & \mbox{if $P=\emptyset$, $\decompositiongraph{\abstractdecomposition}[X]$ is connected and not the empty graph;}\\
		3 & \mbox{if $P\neq \emptyset$ and some vertex in $\decompositiongraph{\abstractdecomposition}[X]$ is not 
				reachable from $\topmap{\abstractdecomposition}(U(P))$,} \\ 
		  & \mbox{or $P=\emptyset$ and $\decompositiongraph{\abstractdecomposition}[X]$ is disconnected.}
		\end{array}
		\right.
		$$
	\end{itemize}

Below, we define the notions of quasimodel and model of a graph  $H$ in a graph $G$. Intuitively, a model should be regarded as a 
certificate that $H$ is a minor of $G$, while a quasi-model is a substructure that can be potentially extended to a model. 

\begin{definition}[Quasimodels and Models]
\label{definition:Model}
Let $H$ and $G$ be graphs. A quasimodel of $H$ in $G$ is a pair of sequences
$M = ([X_x]_{x\in \vertexset{H}}, [y_e]_{e\in |\edgeset{H}|})$ 
satisfying the following
conditions: 
\begin{enumerate}
\item for each $x\in \vertexset{H}$, $X_x$ is a subset of $\vertexset{G}$, 
\item for each two distinct $x,x'\in \vertexset{H}$, $X_x$ is disjoint from $X_{x'}$, 
\item for each $e\in \edgeset{H}$ with endpoints $\{x,x'\}$, either $y_{e}$ is an edge in $\edgeset{G}$ with one endpoint in $X_{x}$ 
	and another endpoint in $X_{x'}$, or $y_e = 0$, 
\item for each two distinct $e, e'\in \edgeset{H}$, $y_e\neq y_{e'}$. 
\end{enumerate}
We say that $M$ is a {\em model} of $H$ in $G$ if additionally, for each $x\in \vertexset{H}$, $G[X_x]$ is a connected subgraph of $G$, and 
	for each $e\in \edgeset{H}$, $y_e > 0$.
\end{definition}

We say that a graph $H$ is a minor of a graph $G$ if there is a model of $H$ in $G$. 
The graph property of the DP-core $\texttt{C-Minor}(H)$ is the set of all graphs that have a model of $H$.
A local witness is a pair of tuples
\begin{equation}
\label{equation:LocalWitnessMinor}
\awitness  = ([(\qconn_x,\Pconn_x)]_{x\in \vertexset{H}}, [b_e]_{e\in E_H}) 
\end{equation}

where for each $x\in \vertexset{H}$, $\qconn_x\in \{0,1\}$, $\Pconn_x$ is a partition of some subset of $[\treewidthvalue+1]$,
and for each $e\in \edgeset{H}$, $b_e\in \{0,1\}$. Such a local witness is final if and only if $b_e=1$ for every $e\in \edgeset{H}$,
and for each $x\in \vertexset{H}$, $\qconn_x\neq 3$ and $\Pconn_x$ has at most one cell. Note that this last condition,
imposed on the pair $(\qconn_x,\Pconn_x)$, is just the condition for a local witness to be final with respect to the DP-core $\texttt{C-Conn}$ defined 
in Subsection \ref{subsection:ConnectivityCore}. Intuitively, 
this will be used to certify that for some $X_x\subseteq \vertexset{\decompositiongraph{\abstractdecomposition}}$, the 
induced subgraph $\decompositiongraph{\abstractdecomposition}[X_x]$ is connected. 

The transitions of $\texttt{C-Minor}(H)[\treewidthvalue]$ are defined in such a way that 
for each $\treewidthvalue$-instructive tree decomposition $\abstractdecomposition$, and each local witness
$\awitness = ([(\qconn_x,\Pconn_x)]_{x\in \vertexset{H}},[b_{e}]_{e\in \edgeset{H}})$, 
$\awitness\in \dynamizationfunction{\texttt{C-Minor}(H)}{\treewidthvalue}(\abstractdecomposition)$ if and only if
the following predicate is satisfied: 
\begin{itemize}
\item $\texttt{P-Minor}(H)[\treewidthvalue](\abstractdecomposition,\awitness) \equiv $ there exists a quasi-model $([X_{x}]_{x\in \vertexset{H}}, [y_e]_{e\in \edgeset{H}})$ 
of $H$ in $G$ such that
\begin{enumerate}
\item for each $x\in \vertexset{H}$, $\texttt{P-Conn}^{X_{x}}[\treewidthvalue](\abstractdecomposition,(\qconn_{x},\Pconn_{x}))=1$, and  
\item for each $e\in \edgeset{H}$, $b_e = 1$ if and only if $y_{e} > 0$. 
\end{enumerate}
\end{itemize}

Each local witness $([(\qconn_x,\Pconn_x)]_{x\in \vertexset{H}},[b_{e}]_{e\in \edgeset{H}})$ 
can be straightforwardly represented using 
$$|\vertexset{H}|\cdot (2+ \treewidthvalue\cdot \lceil \log (\treewidthvalue +2) \rceil) + \edgeset{H} = O(|\vertexset{H}|\cdot \treewidthvalue\cdot \log\treewidthvalue + |\edgeset{H}|)$$
bits, given that each partition $\Pconn_x$ can be represented by at most $\treewidthvalue\log \treewidthvalue$ bits. Nevertheless, this upper bound can be improved to 
$O(k\log k + |V_H| + |E_H|)$ by noting that one can assume that the vertices of $V_H$ are ordered (any fixed ordering) and that 
for each $x,x'\in \vertexset{H}$ with $x\neq x'$, each cell of $P_x$ is disjoint from each cell of $P_{x'}$. More specifically, we may represent 
the sequence $[P_x]_{x\in \vertexset{H}}$ as a sequence of symbols from the alphabet $\{1,\dots,k+1\}\cup \{ |, :\}$, where two consecutive occurrences of 
the symbol $|$ enclose labels occurring in some $P_x$, while the symbol $:$ is used to separate cells inside the partition. 
Numbers occurring after the $x$-th occurrence symbol $|$ and before the $(x+1)$-th occurrence of this symbol correspond to partition $P_x$. The symbol $:$ is used to separate cells 
within two consecutive occurrences of the symbols $|$. Since no number in the set $\{1,\dots,k+1\}$ occurs twice in the sequence,
the symbol $:$ occurs at most $k$ times, and the symbol $|$ occurs at most $|V_{H}|$ times, 
the length of this sequence is at most $(k+1) + k + |V_H|$. Note that the numbers in $\{1,\dots,k+1\}$ require each $O(\log k)$ bits to be 
represented, while the symbols $|$ and $:$ can be represented with $2$ bits. 
Therefore, the overall representation of the witness, including the vector $[b_{e}]_{e\in \edgeset{H}}$ has
$O(k\log k + |V_H| + |E_H|)$ bits.

\end{document}